\documentclass[english,runningheads,11pt]{llncs}
\usepackage{graphicx,amssymb,amsmath,amssymb}
\usepackage[linesnumbered, vlined, ruled]{algorithm2e}
\usepackage[absolute]{textpos}

\usepackage[in]{fullpage}

\def\calP{\mathcal{P}}

\def\calD{\mathcal{D}}
\def\calV{\mathcal{V}}
\def\calO{\mathcal{O}}
\def\calB{\mathcal{B}}
\def\bbB{\mathbb{B}}
\def\calR{\mathcal{R}}
\def\calM{\mathcal{M}}
\def\calS{\mathcal{S}}
\def\calL{\mathcal{L}}
\def\calC{\mathcal{C}}

\newcommand{\spm}{\mbox{$S\!P\!M$}}
\newcommand{\spt}{\mbox{$S\!P\!T$}}
\newcommand{\vis}{\mbox{$V\!i\!s$}}
\newcommand{\Tri}{\mbox{$T\!r\!i$}}

\newtheorem{observation}{Observation}

\begin{document}

\title{Quickest Visibility Queries in Polygonal Domains
}
\author{
Haitao Wang
}
\institute{
Department of Computer Science\\
Utah State University, Logan, UT 84322, USA\\
\email{haitao.wang@usu.edu}
}

\maketitle

\pagenumbering{arabic}
\setcounter{page}{1}

\begin{abstract}
Let $s$ be a point in a polygonal domain $\calP$ of $h-1$ holes and $n$
vertices. We consider a {\em quickest visibility
query} problem. Given a query point $q$ in $\calP$, the goal is to find a
shortest path in $\calP$ to move from $s$ to {\em see} $q$ as quickly
as possible. Previously, Arkin et al. (SoCG 2015) built a data
structure of size $O(n^22^{\alpha(n)}\log n)$ that can answer each
query in $O(K\log^2 n)$ time, where $\alpha(n)$ is the
inverse Ackermann function and $K$ is the size of the visibility polygon
of $q$ in $\calP$ (and $K$ can be $\Theta(n)$ in the worst case). In this paper, we present a new data structure of
size $O(n\log h + h^2)$ that can answer each query in $O(h\log h\log n)$ time.
Our result improves the previous work when $h$ is relatively small.
In particular, if $h$ is a constant, then our result even matches the best result for
the simple polygon case (i.e., $h=1$), which is optimal. As a by-product, we also have a
new algorithm for a {\em shortest-path-to-segment
query} problem. Given a query line segment $\tau$ in $\calP$, the query seeks a shortest path
from $s$ to all points of $\tau$. Previously,  Arkin et al. gave a data
structure of size $O(n^22^{\alpha(n)}\log n)$ that can answer each
query in $O(\log^2 n)$ time, and another data structure of size $O(n^3\log n)$ with $O(\log n)$ query time. We present a data structure of size
$O(n)$ with query time $O(h\log \frac{n}{h})$, which also favors
small values of $h$ and is optimal when $h=O(1)$.
\end{abstract}


\section{Introduction}
\label{sec:intro}

Let $\calP$ be a polygonal domain with $h-1$ holes and a total of $n$ vertices,
i.e., there is an outer simple polygon containing $h-1$ pairwise
disjoint holes and each hole itself is a simple polygon. If $h=1$, then $\calP$ becomes a simple polygon. For any two
points $s$ and $t$ in $\calP$, a {\em shortest path} from $s$ to $t$ is a path
in $\calP$ connecting $s$ and $t$ with the minimum Euclidean length.
Two points $p$ and $q$ are {\em visible} to each other if the line
segment $\overline{pq}$ is in $\calP$. For any point $q$ in $\calP$,
its {\em visibility polygon} consists of all points of $\calP$ visible
to $q$, denoted by $\vis(q)$.

We consider the following {\em quickest visibility
query} problem. Let $s$ be a source point in $\calP$. Given any point
$q$ in $\calP$, the query asks for a path to
move from $s$ to {\em see} $q$ as quickly as possible. Such a
``quickest path'' is actually a shortest path from $s$ to any point of
$\vis(q)$.  The problem has been recently studied by
Arkin et al.~\cite{ref:ArkinSh16}, who built a data structure of size
$O(n^22^{\alpha(n)}\log n)$
that can answer each query in $O(K\log^2 n)$ time,
where $K$ is the size of $\vis(q)$.
In this paper, we present a new data structure of
$O(n\log h + h^2)$ size with $O(h\log h\log n)$ query time.
Our result improves the previous work when $h$ is relatively small.
Interesting is that the query time is independent of $K$, which can be
$\Theta(n)$ in
the worst case. Our result is also interesting in that when
$h=O(1)$, the data structure has $O(n)$ size and $O(\log n)$
query time, which even matches the best result for the simple
polygon case~\cite{ref:ArkinSh16} and is optimal.



As in \cite{ref:ArkinSh16}, in order to solve the quickest visibility queries,
we also solve a {\em shortest-path-to-segment query} problem (or {\em segment
query} for short), which may have independent interest.
Given any line segment $\tau$ in $\calP$, the
segment query asks for a shortest path from $s$ to all points of
$\tau$. Arkin et al.~\cite{ref:ArkinSh16} gave a data structure of
size $O(n^22^{\alpha(n)}\log n)$ that can answer each query in
$O(\log^2 n)$ time,  and another data structure of size $O(n^3\log n)$ with $O(\log n)$ query time. We present a new data structure of $O(n)$ size with
$O(h\log\frac{n}{h})$ query time. Our result
again favors small values of $h$ and attains optimality when $h=O(1)$,
which also matches the best result for the simple polygon case~\cite{ref:ArkinSh16,ref:ChiangOp97}.

Given the shortest path map of $s$, our quickest visibility query data
structure can be built in $O(n\log h+h^2\log h)$ time and our segment
query data structure can be built in $O(n)$ time.
Arkin et al.'s quickest visibility query data structure and their first segment query data structure can both be built in $O(n^22^{\alpha(n)}\log n)$ time, and their second segment query data structure can be built in $O(n^3\log n)$ time~\cite{ref:ArkinSh16}.

Throughout the paper, 
whenever we talk about a query related to paths in $\calP$, the query time always refers to the time for computing the path length, and to output the actual path, it needs additional time linear in the number of edges of the path by standard techniques (we will omit the details about this).

\subsection{Related Work}

The traditional shortest path query problem has been studied
extensively, which is to compute a shortest path to move from $s$ to ``reach'' a query point. Each shortest path query can be answered in $O(\log n)$ time by using the {\em shortest path map} of $s$, denoted by $\spm(s)$, which is of $O(n)$ size. To build
$\spm(s)$, Mitchell~\cite{ref:MitchellSh96} gave an algorithm of
$O(n^{3/2+\epsilon})$ time for any $\epsilon>0$ and $O(n)$
space, and later Hershberger and Suri \cite{ref:HershbergerAn99} presented an
algorithm of $O(n\log n)$ time and space.
If $\calP$ is a simple polygon (i.e., $h=1$), $\spm(s)$ can be built in $O(n)$
time, e.g., see \cite{ref:GuibasLi87}.

For the quickest visibility queries, Arkin et al.~\cite{ref:ArkinSh16} also built a
``quickest visibility map'' of $O(n^7)$ size in $O(n^8\log n)$ time, which
can answer each query in $O(\log n)$ time.
In addition, Arkin et al.~\cite{ref:ArkinSh16} gave a conditional lower bound on the problem by showing that the 3SUM problem on $n$ numbers can be solved in $O(\tau_1+n\cdot \tau_2)$ time, where $\tau_1$ is the preprocessing time and $\tau_2$ is the query time. Therefore, a data structure of $o(n^2)$ preprocessing time and $o(n)$ query time would lead to an $o(n^2)$ time algorithm for 3SUM.


In the simple polygon case (i.e., $h=1$), better results are
possible for both the quickest visibility queries and the segment
queries. For the quickest visibility queries, Khosravi and
Ghodsi~\cite{ref:KhosraviTh05} first proposed a data structure of $O(n^2)$
size that can answer each query in $O(\log n)$ time. Arkin et
al.~\cite{ref:ArkinSh16} gave an improved result and they built a data
structure of $O(n)$ size in $O(n)$ time, with $O(\log n)$ query time.
For the segment queries, Arkin et al.~\cite{ref:ArkinSh16} built a data
structure of $O(n)$ size in $O(n)$ time, with $O(\log n)$ query time.
Chiang and Tamassia~\cite{ref:ChiangOp97} achieved the same
result for the segment queries and they also gave some more general
results (e.g., when the query is a convex polygon).

Similar in spirit to the ``point-to-segment'' shortest path problem, Cheung and Daescu~\cite{ref:CheungAp10} considered a ``point-to-face'' shortest path problem in 3D and approximation algorithms were given for the problem.

\subsection{Our Techniques}
\label{sec:tech}

We first propose a decomposition $\calD$
of $\calP$ by $O(h)$ shortest paths from $s$ to certain vertices of
$\spm(s)$. The decomposition $\calD$, whose size is $O(n)$,
has $O(n)$ cells with the following three key properties. First, any
segment $\tau$ in $\calP$ can intersect at most $O(h)$ cells of
$\calD$. Second, for each cell $\Delta$ of $\calD$,
$\tau\cap \Delta$ consists of at most
two sub-segments of $\tau$.  Third, after $O(n)$ time preprocessing,
for each sub-segment $\tau'$ of $\tau$ in any cell of
$\calD$, the shortest path from $s$ to $\tau'$ can be computed in
$O(\log n)$ time. With $\calD$, we can easily answer each segment query
in $O(h\log \frac{n}{h})$ time by a ``pedestrian'' algorithm.

To solve the quickest visibility queries, an observation is that the shortest path from $s$ to see $q$ is a shortest path from $s$ to a {\em window} of $\vis(q)$, i.e., an extension of the segment $\overline{qu}$ for some reflex vertex $u$ of $\calP$. Hence, the query can be answered by calling segment queries on all $O(K)$ windows of $\vis(s)$ and returning the shortest path. This leads to the $O(K\log^2 n)$ time query algorithm in \cite{ref:ArkinSh16}.

If we follow the same algorithmic scheme and using our new segment
query algorithm, then we would obtain an algorithm of $O(K\cdot h\cdot
\log \frac{n}{h})$ time for the quickest
visibility queries. We instead present a ``smarter'' algorithm.
We propose a ``pruning algorithm'' that prunes some ``unnecessary'' portions of the windows such that it suffices to consider the remaining parts of the windows. Further, with the help of
the decomposition $\calD$, we show that
a shortest path from $s$ to the remaining windows can be found in $O((K+h)\log h\log n)$ time.
We refer to it as {\em the preliminary result}.
To achieve this result, we solve many other
problems, which may be of independent interest. For example, we build a data
structure of $O(n\log h)$ size such that given any query point $t$ and
line segment $\tau$ in $\calP$, we can compute in $O(\log h\log n)$ time
the intersection between $\tau$ and the shortest path from $s$ to $t$ in $\calP$
(or report none if they do not intersect). Our above pruning algorithm is based on a new and interesting technique of using ``bundles''.


To further reduce the query time to $O(h\log h\log n)$,
the key idea is that by using the extended corridor structure of $\calP$~\cite{ref:ChenL113STACS,ref:ChenCo17}, we show that there exists a set $\calS(q)$ of
$O(h)$ {\em candidate windows} such that a shortest path from $s$ to
see the query point $q$ must be a shortest path from $s$ to a window in $\calS(q)$.
This is actually quite consistent with the result in the
simple polygon case, where only one window is needed
for answering each quickest visibility
query~\cite{ref:ArkinSh16}. Once the set $\calS(q)$ is computed, we
can apply our pruning algorithm discussed above on $\calS(q)$ to answer the quickest visibility
query in additional $O(h\log h\log n)$ time. To compute $\calS(q)$,
we give an algorithm of $O(h\log n)$ time, without having to
explicitly compute $\vis(s)$. The algorithm is based on a
modification of the algorithm given in~\cite{ref:ChenVi15} that
can compute $\vis(q)$ in $O(K\log n)$ time for any point $q$, after $O(n+h^2)$ space and
$O(n+h^2\log h)$ time preprocessing.



The rest of the paper is organized as follows. In
Section~\ref{sec:pre}, we introduce notation and review some
concepts. In Section~\ref{sec:segment}, we introduce the
decomposition $\calD$ of $\calP$, and present our algorithm for
the segment queries. We present our preliminary result for the quickest visibility
queries in Section~\ref{sec:first} and give the improved
result in Section~\ref{sec:second}.  Section~\ref{sec:con} concludes the paper.

\section{Preliminaries}
\label{sec:pre}

For any subset $A$ of $\calP$, we say that a point $p$ is {\em (weakly)
visible} to $A$ if $p$ is visible to at least one point of $A$. For any
point $t\in \calP$, we use $\pi(s,t)$ to denote a shortest path from $s$
to $t$ in $\calP$, and in the case where the shortest path is not
unique, $\pi(s,t)$ may refer to an arbitrary such path.
With a little abuse of notation, for any subset $A$ of $\calP$, we
use $\pi(s,A)$ to denote a shortest path from $s$ to all points of
$A$; we use $d(s,A)$ to denote the length of $\pi(s,A)$, i.e., $d(s,A)=\min_{t\in A}d(s,t)$.

Let $\calV$ denote the set of all vertices of $\calP$.

\paragraph{The shortest path map $\spm(s)$.}
$\spm(s)$ is a decomposition of $\calP$ into regions (or cells) such that in each
cell $\sigma$,
the sequence of obstacle vertices along $\pi(s,t)$ is fixed for all $t$ in
$\sigma$~\cite{ref:HershbergerAn99,ref:MitchellSh96}.
Further, the {\em root} of $\sigma$, denoted by $r(\sigma)$,
 is the last vertex of $\calV\cup \{s\}$ in $\pi(s,t)$  for any point $t\in \sigma$
(hence $\pi(s,t)=\pi(s,r(\sigma))\cup \overline{r(\sigma)t}$;
note that $r(\sigma)$ is $s$ if $s$ is visible to $t$).
We classify each edge of a cell $\sigma$ into three types: a portion of
an edge of $\calP$, {\em an extension segment}, which is a line segment extended
from $r(\sigma)$ along the opposite direction from  $r(\sigma)$ to the
vertex of $\pi(s,t)$ preceding $r(\sigma)$, and {\em a bisector
curve/edge} that is a hyperbolic arc. For each point $t$ on
a bisector edge of $\spm(s)$, $t$ is on the common boundary of two cells and
there are two different shortest paths from $s$ to $t$ through the roots of the
two cells, respectively. The {\em vertices} of $\spm(s)$
include $\calV\cup \{s\}$ and all intersections of edges of $\spm(s)$.
The intersection of two bisector edges is called a {\em triple
point}, which has more than two shortest paths from $s$. The map
$\spm(s)$ has $O(n)$ vertices, edges, and cells~\cite{ref:HershbergerAn99,ref:MitchellSh96}.

For differentiation, we call the vertices and edges of the polygonal
domain $\calP$ the {\em obstacle vertices} and the {\em obstacle
edges}, respectively. The holes and the outer polygon of $\calP$ are
also called {\em obstacles}.

The {\em shortest path tree} $\spt(s)$ is the union of shortest paths
from $s$ to all obstacle vertices of $\calP$. $\spt(s)$ has $O(n)$
edges~\cite{ref:HershbergerAn99,ref:MitchellSh96}. Given $\spm(s)$,
$\spt(s)$ can be obtained in linear time. We somethings consider
a further decomposition of $\spm(s)$ by having all edges of $\spt(s)$
in it.

For ease of exposition, we make a general position
assumption that no obstacle vertex has more than one shortest path from
$s$ and no point of $\calP$ has more than three shortest paths from $s$.
Hence, no bisector edge of $\spm(s)$
intersects an obstacle vertex and no three bisector edges intersect at
the same point.

For any polygon $P$, we use $|P|$ to denote the number of vertices of
$P$ and use $\partial P$ to denote the boundary of $P$.


\paragraph{Ray-shooting queries in simple polygons.} Let $P$ be a simple
polygon. With $O(|P|)$ time and space preprocessing, each
ray-shooting query in $P$ (i.e., given a ray in $P$, find the first
point on $\partial P$ hit by the ray) can be answered in $O(\log |P|)$
time~\cite{ref:ChazelleRa94,ref:HershbergerA95}. The result can be
extended to curved simple polygons or
splinegons~\cite{ref:MelissaratosSh92}.

\paragraph{The canonical lists and cycles of planar trees.}
We will often talk about certain planar trees in $\calP$ (e.g.,
$\spt(s)$). Consider a tree $T$ with root $r$. A leaf
$v$ is called a {\em base leaf} if it is the leftmost leaf of a subtree rooted at a child of $r$ (e.g., see Fig.~\ref{fig:cycletree}).
Denote by $\calL(T,v)$ the post-order traversal list
of $T$ starting from such a base leaf $v$,
and we call it a {\em canonical list} of $T$.
The root $r$ must be the last node in $\calL(T,v)$. We remove $r$
from $\calL(T,v)$ and make the remaining list a cycle by connecting
its rear to its front, and let $\calC(T)$ denote the circular list.
Although $T$ may have multiple base leaves,
$\calC(T)$ is unique and we call $\calC(T)$ the {\em canonical cycle} of $T$.
%
\begin{figure}[t]
\begin{minipage}[t]{\linewidth}
\begin{center}
\includegraphics[totalheight=1.6in]{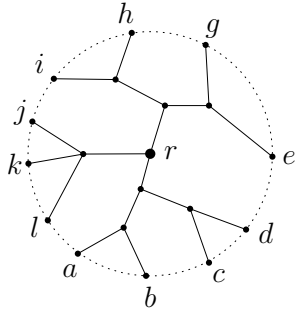}
\caption{\footnotesize
Illustrating a planar tree $T$ with root $r$: $a$ is a base leaf and the list $\calL_l(T,a)$ is $a,b,c,\ldots,l$. }
\label{fig:cycletree}
\end{center}
\end{minipage}
\vspace*{-0.15in}
\end{figure}
We further use $\calL_l(T,v)$ (e.g., see Fig.~\ref{fig:cycletree}) to denote the list of the leaves of $T$
following their relative order in $\calL(T,v)$ 
and use $\calC_l(T)$ to denote the circular list of $\calL_l(T,v)$.
One reason we introduce these notation is the following. Let $e$ be any edge of $T$. 
All nodes of $T$ whose paths to $r$ in $T$ contain $e$ must be consecutive in $\calL(T,v)$ and $\calC(T)$.
Similarly, all leaves of $T$ whose paths to $r$ in $T$ contain $e$ must be consecutive in $\calL_l(T,v)$ and $\calC_l(T)$.

The following observation on shortest paths will be frequently
referred to in the paper.

\begin{observation}\label{obser:10}
\begin{enumerate}
\item
Suppose $\pi_1$ and $\pi_2$ are two shortest paths from $s$ to two
points in $\calP$, respectively; then $\pi_1$ and $\pi_2$ do not
cross each other.
\item
Suppose $\pi_1$ is a shortest path from $s$ to a point in $\calP$ and
$\tau$ is a line segment in $\calP$; then the intersection of
$\pi_1$ and $\tau$ is a sub-segment of $\tau$ (which
may be a single point or empty).
\end{enumerate}
\end{observation}

\section{The Decomposition $\calD$ and the Segment Queries}
\label{sec:segment}


In this section, we introduce a decomposition $\calD$ of $\calP$ and use it to
solve the segment query problem. The decomposition $\calD$
will also be useful for solving the quickest visibility queries.

We first define a set $V$ of points. Let $p$ be an intersection
between a bisector edge of $\spm(s)$ and an obstacle edge.
Since $p$ is on a bisector edge, it is in two cells of $\spm(s)$ and
has two shortest paths from $s$. We make two copies of
$p$ in the way that each copy belongs to only one cell (and thus corresponds to
only one shortest path from $s$). We add the two copies of $p$ to $V$. We
do this for all intersections between bisector edges and obstacle edges.
Consider a triple point $p$, which is in three cells of $\spm(s)$ and
has three shortest paths from $s$. Similarly, we make three copies of
$p$ that belong to the three cells, respectively. We add the three copies
of $p$ to $V$. We do this for all triple points. This finishes the
definition of $V$.

By definition, each point of $V$ has exactly one shortest path from $s$.
Let $\Pi_{V}$ denote the set of shortest paths from $s$ to all points of
$V$. Let $T_{V}$ be the union of all shortest paths of
$\Pi_{V}$. We consider points of $V$ distinct although some of them
are copies of the same physical point. In this way, we can consider
$T_{V}$ as a ``physical'' tree rooted at $s$.

\begin{definition}
Define $\calD$ to be the decomposition of $\calP$ by the edges of $T_{V}$.
\end{definition}

In the following, we assume the shortest path map $\spm(s)$ has
already been computed. We have the following lemma about the
decomposition $\calD$.

\begin{lemma}\label{lem:10}
\begin{enumerate}
\item
The size of the set $V$ is $O(h)$.
\item
The combinatorial size of $\calD$ is $O(n)$.
\item
Each cell of $\calD$ is simply connected.
\item
For any segment $\tau$ in $\calP$, $\tau$ can intersect at most $O(h)$ cells
of $\calD$. Further, for each cell $\Delta$ of $\calD$, the intersection
$\tau$ and $\Delta$ consists of at most two (maximal) sub-segments of $\tau$.
\item
After $O(n)$ time preprocessing, for any segment $\tau'$ in a cell $\Delta$ of
$\calD$, the shortest path from $s$ to $\tau'$ can be computed in
$O(\log |\Delta|)$ time, where $|\Delta|$ is the combinatorial size of
$\Delta$.
\item
For each cell $\Delta$ of $\calD$, $\Delta$ has at most two vertices $r_1$ and $r_2$ (both in $\calV\cup \{s\}$), called ``super-roots'', such that for any point $t\in \Delta$, $\pi(s,t)$ is the concatenation of $\pi(s,r)$ and the shortest path from $r$ to $t$ in $\Delta$, for a super-root $r$ in $\{r_1,r_2\}$.
\item
Given the shortest path map $\spm(s)$, $\calD$ can be computed in $O(n)$ time.
\end{enumerate}
\end{lemma}

We will prove Lemma~\ref{lem:10} later in Section~\ref{sec:decom}.
Below we first give our data structure for answering segment queries
by using Lemma~\ref{lem:10}.

\subsection{The Segment Queries}

As preprocessing, we first compute the decomposition $\calD$. Then, we
build a point location data structure on
$\calD$~\cite{ref:EdelsbrunnerOp86,ref:KirkpatrickOp83}, which can be
done in $O(n)$ time and $O(n)$ space since the size of $\calD$ is $O(n)$ by
Lemma~\ref{lem:10}(2); the data structure can answer each point
location query in $O(\log n)$ time.

In addition, for each cell $\Delta$ of $\calD$, by Lemma~\ref{lem:10}(3),
$\Delta$ is a simple polygon; we build a ray-shooting data structure on
$\Delta$~\cite{ref:ChazelleRa94,ref:HershbergerA95}.
Since the total size of all cells of $\calD$ is
$O(n)$ by Lemma~\ref{lem:10}(2), the total preprocessing time and
space for the ray-shooting queries on all cells of $\calD$ is $O(n)$.

Finally, we do the preprocessing in Lemma~\ref{lem:10}(5). Hence,
given $\spm(s)$, the total preprocessing time and space is $O(n)$.
The following lemma gives our query algorithm.

\begin{lemma}
Given any segment $\tau$ in $\calP$, we can compute a shortest path from $s$
to $\tau$ in $O(h\log \frac{n}{h})$ time.
\end{lemma}
\begin{proof}
Let $a$ and $b$ be the two endpoints of $\tau$, respectively.
Our algorithm works in a ``pedestrian'' way, as follows.

By using a point location query, we find the cell $\Delta_a$ of
$\calD$ that contains $a$. Then, we check whether $\tau$ is contained
in $\Delta_a$. This can be done by using a ray-shooting query as
follows. We
shoot a ray $\rho$ from $a$ towards $b$ and compute the first point $p$ of
$\partial \Delta_a$ hit by the ray. The segment $\tau$ is in $\Delta_a$ if and only
if $b$ is before $p$ on the ray.

If $\tau$ is in $\Delta_a$, then we can
immediately compute the shortest path $\pi(s,\tau)$ from $s$ to $\tau$
in $O(\log |\Delta_a|)$ time by  Lemma~\ref{lem:10}(5).

Otherwise, we compute the shortest path $\pi(s,\overline{ap})$ from
$s$ to the sub-segment $\overline{ap}$ of $\tau$ in $O(\log |\Delta_a|)$
time by
Lemma~\ref{lem:10}(5). Next, based on the edge of $\calD$ containing $p$, we can  determine in constant time the next cell $\Delta$ of $\calP$ that the ray $\rho$ enters. We process the cell $\Delta$ in the similar way as the above
for $\Delta_a$. The algorithm finishes once we process a cell that
contains $b$.

The above computes $\pi(s,\tau')$ for multiple sub-segments
$\tau'$ of $\tau$ such that these sub-segments constitute exactly
$\tau$ and each sub-segment is in a single cell of $\calD$. Clearly, among all shortest paths from $s$ to these sub-segments, the one with the minimum length is the shortest path from $s$ to $\tau$.

To analyze the running time of the above algorithm, let $k$ be
the number of the above sub-segments $\tau'$ of $\tau$. Suppose $\tau'_1,\tau'_2,\ldots,\tau'_k$ are
these sub-segments ordered from $a$ to $b$. For each $1\leq i\leq k$,
let $\Delta_i$ be the cell of $\calD$ that contains $\tau'_i$.
First of all, the point location query for $a$ takes $O(\log n)$ time.
For each $1\leq i\leq k$, determining each sub-segment $\tau'_i$ needs a ray-shooting query in $\Delta_i$, which takes $O(\log |\Delta_i|)$ time;  computing the length
of $\pi(s,\tau_i')$ also takes $O(\log |\Delta_i|)$ time by
Lemma~\ref{lem:10}(5).
Hence, the total time of the algorithm is $O(\log n + \sum_{i=1}^k \log|\Delta_i|)$.

By Lemma~\ref{lem:10}(4), $k=O(h)$. Also, by Lemma~\ref{lem:10}(4), each cell may contain two of the above $k$ sub-segments of $\tau$, and thus it is possible that $\Delta_i$ and $\Delta_j$ refer to the same cell for $i\neq j$.
Let $S$ be the set of the distinct cells of $\Delta_i$ for $i=1,2,\ldots, k$. Since each cell contains at most two of the above $k$ sub-segments of $\tau$, $\sum_{i=1}^k \log|\Delta_i|\leq 2\cdot \sum_{\Delta\in S}\log |\Delta|$. Further, since the cells of $S$ are
distinct, we have $\sum_{\Delta\in S}|\Delta|=O(n)$. 
Due to $|S|\leq k=O(h)$, we have $\sum_{\Delta\in S}\log |\Delta| = O(h\log\frac{n}{h})$.

Therefore, the total time of the algorithm is bounded by $O(h\log \frac{n}{h})$.
\qed
\end{proof}

We summarize our result for segment queries in the following theorem.

\begin{theorem}\label{theo:segment}
Given the shortest path map $\spm(s)$, we can build a data structure
of $O(n)$ size in $O(n)$ time, such that each
segment query can be answered in $O(h\log \frac{n}{h})$ time.
\end{theorem}


\subsection{The Decomposition $\calD$ and Proving
Lemma~\ref{lem:10}}
\label{sec:decom}

In this section we provide the details for $\calD$ and prove
Lemma~\ref{lem:10}.

\begin{figure}[t]
\begin{minipage}[t]{\linewidth}
\begin{center}
\includegraphics[totalheight=1.6in]{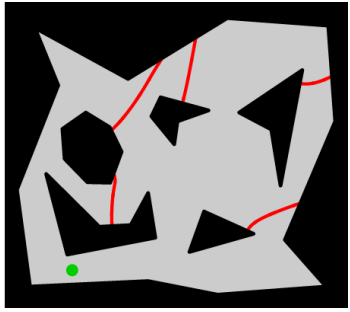}
\caption{\footnotesize
Illustrating the bisector edges of shortest path map (the back area is
the obstacle space): the green point
is the source $s$ and the red curves are the bisector edges. The
figure is generated by the applet in~\cite{ref:HershbergerGe14}} \label{fig:spm}
\end{center}
\end{minipage}
\vspace*{-0.15in}
\end{figure}

Let $\calO$ denote the obstacle space, which is the complement of the
free space of $\calP$. More specifically, $\calO$ consists of the
$h-1$ simple polygonal holes of $\calP$ and the (unbounded) region
outside the outer boundary of $\calP$.
Let $\calB$ denote the union of all bisector edges of $\spm(s)$.
Mitchell \cite{ref:MitchellA91} proved that $\calO\cup \calB$ is
simply connected and $\calP\setminus\calB$ is also
simply connected (e.g., see Fig.~\ref{fig:spm}).
We consider $\calO\cup \calB$ as a planar graph $G$, defined as follows.

The vertex set of $G$ consists of all obstacles of $\calO$ and all triple points
of $\spm(s)$. For any two vertices of $G$, if they are connected by a chain of
bisector edges in $\spm(s)$ such
that the chain does not contain any other vertex of
$G$, then $G$  has an edge connecting the two vertices, and further,
we call the above chain of bisector edges a {\em bisector
super-curve} (e.g., in Fig.~\ref{fig:spm}, each red curve is a
bisector super-curve). We have the following observation about $G$.

\begin{observation}
$G$ is a simple graph, i.e., $G$ does not have a self-loop and no two vertices have more
than one edge.
$G$ has $O(h)$ vertices, edges, and faces.
\end{observation}
\begin{proof}
The first part of the observation can be proved easily from
Mitchell's observation in~\cite{ref:MitchellA91}
that $\calP\setminus\calB$ is simply connected, as follows.

Indeed, assume to the contrary that $G$ has a self-loop at a vertex
$v$. According to our definition, the self-loop corresponds to a
bisector super-curve that connects the vertex $v$ (either a triple
point or an obstacle) to itself.
Let $R$ be region bounded by bisector-super curve and $v$. Hence, $R$
is closed, which contradicts with that $\calP\setminus\calB$ is simply connected.

Similarly, assume to the contrary that two vertices $u$ and $v$ have
two edges. Then, the
two edges correspond to two bisector super-curves. Thus, the region bounded by the two
bisector super-curves and the two vertices is closed, incurring
contradiction again.

To prove the second part of the observation, note that
$G$ is a planar graph.

First, it is known that the number of triple points is
$O(h)$~\cite{ref:Eriksson-BiqueGe15}. Since there are $h$
obstacles in $\calO$, the number of vertices of $G$ is $O(h)$.

Second, the faces of $G$ correspond exactly to the faces of the $(\leq
1)-\spm$ of $\calP$ defined in \cite{ref:Eriksson-BiqueGe15}, whose
total number is proved to be $O(h)$ \cite{ref:Eriksson-BiqueGe15}
(see Lemma 4.3 with $k=1$).
Therefore, the number of faces of $G$ is $O(h)$.

Finally, since both the number of vertices and the number of faces of $G$ are
$O(h)$, the number of edges of $G$ is also $O(h)$.
\qed
\end{proof}

Let $V_1$ be the set of all triple points. It is known that
$|V_1|=O(h)$~\cite{ref:Eriksson-BiqueGe15}.
Let $V_2$ be the set of intersections between obstacle edges
and bisector edges of $\spm(s)$. It is not difficult to see that
each point of $V_2$ corresponds to
an intersection between an obstacle and a bisector super-curve.
Since $G$ has $O(h)$ edges, there are $O(h)$ bisector super-curves.
Thus, $|V_2|=O(h)$.
Recall that $V$ consists of three copies of each point of $V_1$ and
two copies of each vertex of $V_2$.  Since both $|V_1|$ and
$|V_2|$ are $O(h)$, we have $|V|=O(h)$.
This proves Lemma~\ref{lem:10}(1).

Since $|V|=O(h)$, $\Pi_V$ is the set of $O(h)$ shortest paths. Note
that each edge of any path of $\Pi_V$ except the last edge (i.e., the
one connecting a point of $V$) is an edge of the shortest path
tree $\spt(s)$. Hence, the total number of edges of the tree $T_{V}$ is $O(n)$.
Since $\calD$ is the decomposition of $\calP$ by the edges of
$T_{V}$, the combinatorial
size of $\calD$ is $O(n)$. This proves Lemma~\ref{lem:10}(2).

Throughout the paper, let $h^*=|V|$. Hence, $h^*=O(h)$.

To prove the rest of Lemma~\ref{lem:10}, we introduce
another decomposition $\calD'$ as follows.

\begin{definition}
Define $\calD'$ to be the decomposition of $\calP$ by the edges of
$T_{V}\cup \calB$.
\end{definition}

By definition, $\calD$ can be obtained from $\calD'$ by removing all bisector edges of $\calB$.


\begin{lemma}\label{lem:30}
Each cell of $\calD'$ is simply connected.
\end{lemma}
\begin{proof}
Let $Q_0$ be the decomposition of $\calP$ by the edges of $\calB$. Note that $Q_0$
is exactly $\calP\setminus \calB$, which is simply connected
\cite{ref:MitchellA91}.

Let the points of $V$ be $v_1,v_2,\ldots,v_{h^*}$, ordered
arbitrarily. Consider the decomposition $Q_1$ of $Q_0$ by the shortest path
$\pi(s,v_1)$.  Note that $Q_1$ may have more than one connected cell.
Recall that $v_1$ is on a bisector edge
of $\calB$. Since $Q_0$ is simply connected,
$\pi(s,v_1)$ does not cross any bisector edges of
$\spm(s)$, and $\pi(s,v_1)$ itself does not form any cycle,
each cell of  $Q_1$ is simply connected.

Similarly, consider the decomposition $Q_2$ of
$Q_1$ by the shortest path $\pi(s,v_2)$. Again, $\pi(s,v_2)$
does not cross any bisector edge of $\calB$. Further, by
Observation~\ref{obser:10}(1), $\pi(s,v_2)$
and $\pi(s,v_1)$ do not cross each other. Hence, $\pi(s,v_2)$ does not
cross any edge of $Q_1$. Since each cell of $Q_1$ is simply connected,
each cell of $Q_2$ is also simply connected.

We keep considering the rest of the paths $\pi(s,v_i)$ for $i=3,4, \ldots,
h^*$ one by one in the same way as above.
By the similar argument we can obtain that each cell of
$\calD_{h^*}$, which is $\calD'$, is simply connected.
\qed
\end{proof}



It is known that $\calP\setminus \calB$ is simply connected and $\pi(s,t)$ is in
$\calP\setminus \calB$ for any point $t\in \calP$~\cite{ref:MitchellA91}.
To simplify the discussion, together with the copies of the points of
$V$, we consider $\calP'=\calP\setminus \calB$ as a
simple polygon (with some curved edges) by making two copies for each interior point
of every bisector super-curve such that they respectively belong
to the two sides of the curve. In this way, for any point $t\in
\calP'$, it has a unique shortest path $\pi(s,t)$ from $s$ in $\calP'$, which is also a
shortest path in $\calP$. In this way, $\calD'$ becomes a decomposition of $\calP'$ by
the tree $T_V$.

Consider any cell $\Delta'$ of $\calD'$. Recall that $\calV$ is the set of all vertices of $\calP$. We consider the points of
$\calV\cup V\cup \{s\}$ on the boundary $\partial \Delta'$ of $\Delta'$
as vertices of $\Delta'$. Then, the boundary portion between any two adjacent
vertices of $\Delta'$ is an obstacle edge, an edge of $T_V$,
or a bisector super-curve.
Let $p$ be any point of $\Delta'$.  Let $r_{\Delta'}$
be the point of $\Delta'\cap \pi(s,p)$ closest to $s$.
We call $r_{\Delta'}$ the {\em super-root} of $\Delta'$, which is
unique (i.e., independent of $p$) due to the following lemma.

\begin{lemma}\label{lem:40}
\begin{enumerate}
\item
The point $r_{\Delta'}$ is in $\calV\cup \{s\}$, i.e.,
it is either $s$ or an obstacle vertex.
\item
$\pi(s,r_{\Delta'})$ is a sub-path of a shortest path in $\Pi_V$.
\item
For any point $t\in \Delta'$, the concatenation of
$\pi(s,r_{\Delta'})$ and the shortest path from $r_{\Delta'}$ to $t$
in $\Delta'$ is the shortest path $\pi(s,t)$ from $s$ to $t$ in
$\calP'$.
\end{enumerate}
\end{lemma}
\begin{proof}
We prove the lemma by induction in a similar way as in Lemma~\ref{lem:30}.
We use the same terminology as in the proof of Lemma~\ref{lem:30}.
Let the points of $V$ be $v_1,v_2,\ldots,v_{h^*}$, ordered arbitrarily.
Let $Q_0=\calP\setminus\calB$. For each $1\leq i\leq
h^*$, let $Q_i$ denote the decomposition of $Q_{i-1}$ by
$\pi(s,v_i)$. We let $\Pi_0=\emptyset$. For each $1\leq i\leq
h^*$, let $\Pi_i=\Pi_{i-1}\cup \{\pi(s,v_i)\}$. Hence, $\Pi_V=\Pi_{h^*}$.

Initially, consider the decomposition $Q_0$. Note
that there is only one cell $\Delta'$ in $Q_0$.  Clearly,
$r_{\Delta'}=s$ and all three statements hold for $Q_0$ and $\Pi_0$.
We assume the lemma statements hold for $Q_{i-1}$ and $\Pi_{i-1}$.
Our goal is to prove that the lemma statements hold for $Q_i$ and $\Pi_i$.

Let $\Delta'$ be the cell of $Q_{i-1}$ containing $v_i$. By induction,
$\pi(s,v_i)$ is the concatenation of $\pi(s,r_{\Delta'})$ and the
shortest path $\pi(r_{\Delta'},v_i)$ from $r_{\Delta'}$ to $v_i$ in
$\Delta'$. Also by induction, $\pi(s,r_{\Delta'})$ is a sub-path of
$\Pi_{i-1}$. Hence, $\pi(s,v_i)$ does not partition any cell of
$Q_{i-1}$ other than $\Delta'$. In other words, for any cell
$\Delta''$ of $Q_{i-1}$, if $\Delta''\neq \Delta'$, then $\Delta''$ is
still in $Q_i$, and thus the lemma statements still hold on $\Delta''$
and $\Pi_i$.

For the cell $\Delta'$, $\pi(r_{\Delta'},v_i)$
partitions $\Delta'$ into multiple {\em sub-cells}.
Consider any sub-cell $\delta$ of $\Delta'$. Our goal is to show that
the lemma statements hold on $\delta$ and $\Pi_i$. Depending on
whether $\delta$ contains $r_{\Delta'}$, there are two cases.

\paragraph{The case $r_{\Delta'}\in \delta$.}
We first consider the case where $\delta$ contains $r_{\Delta'}$.
Consider any point $p$ in $\delta$. Since $\delta\subseteq \Delta'$, $r_{\Delta'}\in \delta$, and the point of $\Delta'\cap \pi(s,p)$ closest
to $s$ is $r_{\Delta'}$, the point of $\delta\cap \pi(s,p)$ closest
to $s$ is also $r_{\Delta'}$. Hence, $r_{\delta}=r_{\Delta'}$. By
induction, the first and second statements of the lemma hold for
$\delta$ and $\Pi_i$.

For the third statement, consider any point $t\in \delta$. Since $t\in \Delta'$,
$\pi(s,t)$ is a concatenation of $\pi(s,r_{\Delta'})$ and
$\pi(r_{\Delta'},t)$, and the latter path is in $\Delta'$.
To prove the third statement, it sufficient to show that $\pi(r_{\Delta'},t)$ is in $\delta$. Indeed, assume to the contrary that $\pi(r_{\Delta'},t)$ is not in
$\delta$. Then, since $\delta$ is a cell of the decomposition of
$\Delta'$ by $\pi(r_{\Delta'},v_i)$, $\pi(r_{\Delta'},t)$ must cross
$\pi(r_{\Delta'},v_i)$. However, this is not possible due to
Observation~\ref{obser:10}(1).
Hence, $\pi(r_{\Delta'},t)$ must be in $\delta$.

\paragraph{The case $r_{\Delta'}\not\in \delta$.}
Suppose $\delta$ does not contain $r_{\Delta'}$. Let $a$ be the point of
$\pi(r_{\Delta'},v_i)\cap \delta$ closest to $r_{\Delta'}$. We first show that for any
point $p\in \delta$, $a$ is the point of $\pi(s,p)\cap \delta$ closest to $s$.

Indeed, since $p\in \Delta'$, $\pi(s,p)$ contains $r_{\Delta'}$ and $\pi(r_{\Delta'},p)$ is
in $\Delta'$. Since $r_{\Delta'}$ is not in $\delta$, let $b$ be the first point in
$\delta$ we encounter if we traverse on $\pi(r_{\Delta'},p)$ from $r_{\Delta'}$ to $p$.
Clearly, $b$ is not $r_{\Delta'}$ since otherwise $r_{\Delta'}$ would
be in $\delta$. Since $\delta$ is a cell of the decomposition of
$\Delta'$ by $\pi(r_{\Delta'},v_i)$, $b$ must be on $\pi(r_{\Delta'},v_i)$.
In other words, $b\in \delta\cap \pi(r_{\Delta'},v_i)$.

Since $b$ is on both $\pi(r_{\Delta'},v_i)$ and
$\pi(r_{\Delta'},p)$, $b$ is also the first point in $\delta$ we
encounter if we
traverse on $\pi(r_{\Delta'},v_i)$ from $r_{\Delta'}$ to $v_i$. Thus, $b$
is the point of $\pi(r_{\Delta'},v_i)$ closest to $r_{\Delta'}$.
Hence, we obtain $b=a$.

On the other hand, the definition of $b$ implies that $b$ is
the point of $\pi(s,p)\cap \delta$ closest to $s$.

Therefore, $a$ is the point of $\pi(s,p)\cap \delta$ closest to $s$.
This implies that $r_{\delta}=a$.

Note that $a$ is a vertex of $\pi(r_{\Delta'},v_i)$ and $a$ cannot be
$v_i$. Thus, $a$ must be either $s$
or an obstacle vertex (in fact, $a$ cannot be $s$ either due to $a\neq
r_{\Delta'}$), which
proves the first statement of the lemma.

Since $a$ is on $\pi(r_{\Delta'},v_i)$ and thus is on $\pi(s,v_i)$,
$\pi(s,a)$ is a sub-path of $\pi(s,v_i)\in \Pi_i$. This proves the
second statement of the lemma.

For the third statement, consider any point $t\in \delta$. Since $t\in
\Delta'$, by induction,
$\pi(s,t)$ is the concatenation of $\pi(s,r_{\Delta'})$ and
$\pi(r_{\Delta'},t)$, and $\pi(r_{\Delta'},t)$ is in $\Delta'$.
Using the same analysis as above, we can show that $\pi(r_{\Delta'},t)$ must
contain $a$. Further, the portion of $\pi(r_{\Delta'},t)$ between $a$
and $t$ must be in $\delta$, since otherwise $\pi(r_{\Delta'},t)$ would cross
$\pi(r_{\Delta'},v_i)$, incurring contradiction.
Hence, the portion of $\pi(r_{\Delta'},t)$ between $a$ and $t$ is the
shortest path from $a$ to $t$ in $\delta$.
Thus, $\pi(s,t)$ is the concatenation of $\pi(s,a)$ and the shortest
path from $a$ to $t$ in $\delta$. This proves the third statement.

This proves that all lemma statements hold for $\delta$ and $\Pi_i$,
and thus hold for $Q_i$ and $\Pi_i$.

The lemma thus follows.
\qed
\end{proof}


\begin{observation}\label{obser:20}
Each cell $\Delta'$ of $\calD'$ has at most one bisector super-curve on its
boundary.
\end{observation}
\begin{proof}
Assume to the contrary there are two bisector super-curves on the
boundary of $\Delta'$. Then, there must exist an endpoint $p$ of one of
these two bisector super-curves such that the shortest path $\pi(s,p)$
partitions $\Delta'$ into two cells that contain the two bisector
super-curves, respectively. This implies that $\pi(s,p)$ is not in
$\Pi_V$. Since the two endpoints of every bisector super-curve are in
$V$, we obtain $p\in V$ and $\pi(s,p)$ is not in $\Pi_V$, a contradiction.
\qed
\end{proof}

Since $T_V$ is a planar tree, we can define its canonical lists as
discussed in Section~\ref{sec:pre}.
Let $v_1$ be an arbitrary base leaf of $T_V$, which can be found in
$O(n)$ time. Let the leaf list $\calL_l(T_V,v_1)$ be
$v_1,v_2,\ldots, v_{h^*}$, which follow the
counterclockwise order  along $\partial \calP'$.



For each $1\leq i\leq h^*$, let $\alpha_i$ denote the portion of
$\partial\calP'$ counterclockwise from $v_i$ to $v_{i+1}$ (let $v_{h^*+1}$ refer
to $v_1$). Note that $\alpha_i$
is either a bisector super-curve or a chain of obstacle edges.
Suppose we move a point $t$ on $\alpha_i$ from
$v_i$ to $v_{i+1}$. The shortest path $\pi(s,t)$ will continuously change with
the same topology since $\pi(s,t)$ is always in $\calP'$ (which is
simply connected). Let $R_i$ be
the region of $\calP'$ that is ``swept'' by $\pi(s,t)$ during the
above movement of $t$. More specifically, let $p_i$ be the common point on
$\pi(s,v_i)\cap \pi(s,v_{i+1})$ that is farthest to $s$. Then, $R_i$
is bounded by $\pi(p_i,v_i)$, $\pi(p_i,v_{i+1})$, and $\alpha_i$.
For convenience of discussion, we let $R_i$ also contain the common
sub-path $\pi(s,p_i)=\pi(s,v_i)\cap\pi(s,v_{i+1})$ and we call
$\pi(s,p_i)$ the {\em tail} of $R_i$.  We call the region bounded by
$\pi(p_i,v_i)$, $\pi(p_i,v_{i+1})$, and $\alpha_i$ the {\em cell} of
$R_i$.  We consider
$\pi(s,v_i)$, $\pi(s,v_{i+1})$, and $\alpha_i$ as the three portions
of the boundary $\partial R_i$ of $R_i$. The definition implies that
for any point $t$ in $R_i$,
$\pi(s,t)$ is in $R_i$. In fact, if $t$ is in the cell of $R_i$, then
$\pi(s,t)$ is the concatenation of $\pi(s,p_i)$ and the shortest path
from $p_i$ to $t$ in the cell.  Clearly, $\calP'$ is the union of
$R_1,R_2,\ldots, R_{h^*}$. Let $\calR=\{R_1,R_2,\ldots, R_{h^*}\}$.
The next lemma is proved with the help of the regions of $\calR$.
The set $\calR$ will also be quite useful in Section~\ref{sec:first}.
Recall that each edge of $\partial\Delta'$ is either an obstacle edge,
a bisector super-curve, or an edge of $T_V$ (also called a {\em shortest
path edge}).

\begin{lemma}\label{lem:50}
For each cell $\Delta'$ of $\calD'$, there are two shortest paths of $\Pi_V$ that contain all
shortest path edges of $\partial\Delta'$.
\end{lemma}
\begin{proof}
By the definitions of the regions of $\calR$, $\Delta'$ is
contained in the cell of a region $R_i$ of $\calR$. Therefore, each shortest
path edge of $\partial\calD'$ belongs to either $\pi(s,v_i)$ or
$\pi(s,v_{i+1})$. 
\qed
\end{proof}

Observe that the decomposition $\calD$ can be obtain  from $\calD'$ by
removing all bisector super-curves. For any bisector
super-curve $\alpha$, the two cells of $\calD'$ incident to $\alpha$
are merged into one cell of
$\calD$. Due to Observation~\ref{obser:20}, a cell of $\calD'$ can be
merged into at most one cell of $\calD$. Therefore, for each cell $\Delta$
of $\calD$, either $\Delta$ is also in $\calD'$ or $\Delta$ is a {\em
merged cell} merged by exactly two cells of $\calD'$. Since every cell
of $\calD'$ is simply
connected, each cell of $\calD$ is also simply connected. This proves
Lemma~\ref{lem:10}(3).

Consider any line segment $\tau\in \calP$. By
Observation~\ref{obser:10}(2), $\tau$ can cross any shortest
path of $\Pi_V$ at most once. Hence, $\tau$ can cross the shortest paths
of $\Pi_V$ at most $O(h)$ times in total. Whenever $\tau$ crosses the boundary
of a cell of $\calD$, it must cross a shortest path of $\Pi_V$.
Thus, $\tau$ can intersect  $O(h)$ cells of $\calD$. This proves
the first part of Lemma~\ref{lem:10}(4). For the second part, consider any cell
$\Delta$. By Lemma~\ref{lem:50}, if $\Delta$ is not a merged cell, then
$\tau$ can cross the boundary of $\Delta$ at most twice; otherwise,
$\tau$ can cross the boundary of $\Delta$ at
most four times. Therefore, the intersection $\tau\cap \Delta$ consists of
at most two (maximal) sub-segments of $\tau$. This proves the second part of
Lemma~\ref{lem:10}(4).

In the sequel, we prove Lemma~\ref{lem:10}(5). Consider any cell $\Delta$ of $\calD$.
According to our discussion above, $\Delta$ is either in $\calD'$ or a merged cell of two
cells $\Delta_1$ and $\Delta_2$ of $\calD'$. If it is the former case, then we also
call $r_{\Delta}$ the {\em super-root} of $\Delta$; otherwise, we call
$r_{\Delta_1}$ and $r_{\Delta_2}$ the two {\em super-roots} of $\Delta$.
Lemma~\ref{lem:40} leads to the following lemma, which proves Lemma~\ref{lem:10}(6).

\begin{lemma}\label{lem:60}
For any cell $\Delta$ of $\calD$, the following hold.
\begin{enumerate}
\item
Its two super-roots are in $\calV\cup \{s\}$.
\item
For each super-root $r$ of $\Delta$, $\pi(s,r)$ is a sub-path of a
shortest path in $\Pi_V$.
\item
For any point $t\in \Delta$, $\pi(s,t)$ is the concatenation of
$\pi(s,r)$ and the shortest path from $r$ to $t$ in $\Delta$, for a
super-root $r$ of $\Delta$.
\end{enumerate}
\end{lemma}
\begin{proof}
By Lemma~\ref{lem:40}, the proof is straightforward because either
$\Delta$ is a cell of $\calD$ or a merge of two cells of $\calD$.
\qed
\end{proof}

Recall that for any simple polygon $P$ and a fixed source point, each
segment query can be answered in $O(\log
|P|)$ time after $O(|P|)$ time preprocessing~\cite{ref:ArkinSh16}.
As preprocessing, for each cell $\Delta$ of $\calD$, since it is a simple
polygon, we compute the above segment query data structure
with respect to each super-root of $\Delta$. This takes $O(n)$ time and
space in total by Lemma~\ref{lem:10}(2).

Consider any segment $\tau'$ in a cell $\Delta$ of $\calD$. By
Lemma~\ref{lem:60}, $\pi(s,\tau')$ is the
concatenation of $\pi(s,r)$ from $s$ to a super-root
$r$ of $\Delta$ and the
shortest path $\pi(r,\tau')$ from $r$ to $\tau'$ in $\Delta$. As $r$
is in $\calV\cup\{s\}$ by Lemma~\ref{lem:60}(1), $\pi(s,r)$ is
available from $\spm(s)$,
and $\pi(r,\tau')$ can be found in $O(\log |\Delta|)$ time.
Hence, our query algorithm works as follows. For each super-root
$r$ of $\Delta$, we compute $\pi(s,r)$ and $\pi(r,\tau')$ to obtain a
``candidate'' shortest path from $s$ to $\tau'$. Then, we return the
shorter one of the at most two candidates paths as the solution. The
total time is $O(\log |\Delta|)$.  This proves Lemma~\ref{lem:10}(5).

\paragraph{Remark.} One may wonder why we do not use $\calD'$ instead
of $\calD$ to answer the segment queries. The reason is that the
boundaries of cells of $\calD'$ contain bisector super-curves and
the query segment $\tau$ may intersect a bisector super-curve multiple
times, and thus a similar observation as
Lemma~\ref{lem:10}(4) cannot be guaranteed on $\calD'$.
\paragraph{}

Finally, we prove Lemma~\ref{lem:10}(7) in the following lemma.

\begin{lemma}
Given $\spm(s)$, the decomposition $\calD$ can
be computed in $O(n)$ time.
\end{lemma}
\begin{proof}
Let $\calD_1$ be the decomposition of $\spm(s)$ by the edges of
$\spt(s)$. As discussed before, we can easily obtain $\spt(s)$ from $\spm(s)$
and thus obtain $\calD_1$ in $O(n)$ time. Further, for each point
$v\in V$, we add to $\calD_1$ the last edge of the shortest path $\pi(s,v)$, which
is also the edge connecting $v$ to the root of the cell of $\spm(s)$
containing $v$. Let $\calD_2$ be the resulting
decomposition, which can be obtained in $O(n)$ time.
Note that each edge of $T_V$ is also an edge of $\calD_2$.

Since $\calD$ is a decomposition of $\calP$ by the
edges of $T_V$, $\calD$ can be obtained from $\calD_2$ by removing those edges
that are not in $\calD$. To this end, we first remove all bisector edges
from $\calD_2$. Then, we remove the edges of $\spt(s)$ that are not in
$T_V$. This can be done by first
marking all edges of $T_V$ in $\calD_2$ and then removing all unmarked
edges of $\spt(s)$ from $\calD_2$. Below we only discuss how to mark all
edges of $T_V$ in $O(n)$ time since the latter step is trivial.

For each vertex $v$ of $V$, we mark the edges of $\pi(s,v)$
in $\calD_2$ as follows. We start from $v$ and traverse along
$\pi(s,v)$ from $v$ to $s$, marking every edge that has not been
marked yet; we stop the traversal either when we encounter $s$ or we
encounter an edge that has been marked. In this way, every edge of
$T_V$ is marked exactly once. Since $T_V$ has $O(n)$ edges,
the above marking algorithm runs in $O(n)$ time.

Thus, the decomposition $\calD$ can be computed in $O(n)$ time.
\qed
\end{proof}

\section{The Quickest Visibility Queries: The Preliminary Result}
\label{sec:first}

In this section, we give our preliminary result on quickest visibility
queries, which sets the stage for our improved result in
Section~\ref{sec:second}.

For any subset $A$ of $\calP$, a point $p\in A$ is called a {\em closest point} of $A$ (with respect to $s$) if $d(s,A)=d(s,p)$.

Given any query point $q$ in $\calP$, our goal is to find a shortest
path from $s$ to $\vis(q)$. Let $q^*$ be a closest point of $\vis(q)$.
To answer the query, it is sufficient to determine $q^*$. Thus we will focus on
finding $q^*$. Note that if $q$ is visible to $s$, then $q^*=s$. We can determine whether $s$ is visible to $q$ in $O(\log n)$ time by checking whether $q$ is in the cell of $\spm(s)$ whose root is $s$. In the following, we assume $s$ is not visible to $q$.


\begin{figure}[t]
\begin{minipage}[t]{\linewidth}
\begin{center}
\includegraphics[totalheight=1.0in]{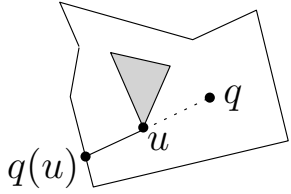}
\caption{\footnotesize
Illustrating a window $\overline{uq(u)}$ of $q$.}
\label{fig:window}
\end{center}
\end{minipage}
\vspace*{-0.15in}
\end{figure}

We define the {\em windows} of $q$ and $\vis(q)$, which were
used for studying the visibility polygons, e.g.,
\cite{ref:BoseEf02,ref:ChenWe15}. Consider an obstacle vertex $u$
that is visible to $q$ such that the two incident obstacle edges of $u$ are on
the same side of the line through $q$ and $u$ (e.g., see Fig.~\ref{fig:window}).
Let $q(u)$ denote the first point on $\partial\calP$ hit by the ray from $u$ along the direction from $q$ to $u$. Then $\overline{uq(u)}$ is called a
{\em window} of $q$; we say that the window is {\em defined by} $u$. Further, we call
$\overline{qq(u)}$ the {\em extended window} of $\overline{uq(u)}$.

Each window of $q$ is an edge of $\vis(q)$, and thus the number of
windows of $q$ is $O(K)$, where $K=|\vis(q)|$.
Further, there must be a closest point $q^*$ that is
on a window of $q$~\cite{ref:ArkinSh16}. Hence, as in~\cite{ref:ArkinSh16},
a straightforward algorithm to compute $q^*$ is to compute
shortest paths from $s$ to all windows of $s$ and the path of
minimum length determines $q^*$.
To compute shortest paths from $s$ to all windows, if we apply our
segment queries on all windows using
Theorem~\ref{theo:segment}, then the total time would be
$O(K\cdot h\cdot \log\frac{n}{h})$.
In the rest of this section, we present an algorithm that can compute
$q^*$ in $O((K+h)\log h\log n)$ time, without having to compute
shortest paths to all windows. The key idea is to prune some (portions of) windows
such that $q^*$ is still in the remaining windows and the shortest paths from $s$ to all
remaining windows can be computed efficiently.

\subsection{The Algorithm Overview}

As the first step, we compute $\vis(q)$, which can be done in $O(K\log
n)$ time after $O(n+h^2\log h)$ time and $O(n+h^2)$ space
preprocessing~\cite{ref:ChenVi15}.
Then, we can find all windows and extended-windows in $O(K)$ time.
For ease of exposition, we make a general position assumption for $q$
that $q$ is not collinear with any two obstacle vertices. The
assumption implies that $q$ is in the interior of $\calP$ and no two windows are collinear.

Let $u_0$ be the root of the cell of $\spm(s)$ containing $q$ (if $q$
is on the boundary of multiple cells, then we take an arbitrary such
cell). Hence, $\pi(s,u_0)\cup\overline{u_0q}$ is a shortest path $\pi(s,q)$ from
$s$ to $q$. Note that $u_0$ must define a window
$\overline{u_0q(u_0)}$ of $q$~\cite{ref:MitchellA91}. Let
$\overline{u_0q({u_0})}, \overline{u_1q({u_1})}, \ldots,
\overline{u_{k}q(u_{k})}$ be all windows of $q$ ordered {\em clockwise}
around $q$. Clearly, $k=O(K)$.
For each $0\leq i\leq k$, let $q_i=q(u_i)$.

Note that the window $\overline{u_0q_0}$ is special in the sense that
$u_0$ is in $\pi(s,q)$. So we first apply our algorithm in
Theorem~\ref{theo:segment} on $\overline{u_0q_0}$ to compute
the closest point $q^*_0$ of $\overline{u_0q_0}$. Clearly, if
$q^*\in \overline{u_0q_0}$, then $q^*=q^*_0$.  In the
following, we assume  $q^*\not\in \overline{u_0q_0}$.
Let $Q=\{q,q_1,q_2,\ldots,q_k\}$. Note that $Q$ does not contain $q_0$ but
$q$.
If $q^*\in Q$, then we can find $q^*$ by computing $d(s,p)$ for
all $p\in Q$, which can be done in $O(k\log n)$ time using $\spm(s)$.
In the following, we assume $q^*\not\in Q$.
Note that the above assumption that $q^*\not\in \overline{u_0q_0}\cup
Q$ is only for arguing the correctness of our following
algorithm, which actually proceeds without knowing whether the
assumption is true or not.

For each $0\leq i\leq k$, let $w_i=\overline{qq_i}$,
i.e., the extended window of $\overline{u_iq_i}$.
Let $W=\{w_i\ |\ 1\leq i\leq k\}$. 
For convenience of discussion, we assume that each $w_i$ of $W$ does not
contain its two endpoints $q$ and $q_i$ (but the endpoints of $w_i$ still refer to $q$ and $q_i$).
Since $q^*\not\in \overline{u_0q_0}\cup Q$,
$q^*$ must be on an extended window of $W$.
Clearly, $q^*$ is also a closest point of $W$.
Since no two windows of $q$ are collinear, no extended-window of $W$
contains another.
We assign each window $w_i\in W$ a direction
from $q$ to $q_i$, so that we can talk about its left or right side.

Suppose $q^*$ is on $w_i\in W$. Since $w_i$ is an open segment, by the
definition of $q^*$, the shortest path
$\pi(s,q^*)$ must reach $q^*$ from either the left side or the right
side of $w_i$. Formally, we say that $\pi(s,q^*)$ reaches $q^*$ from
the left side (resp., right side) of $w_i$ if there is a small
neighborhood of $q^*$ such that all points of $\pi(s,q^*)$ in the
neighborhood are on the left side (resp., right side) of $w_i$.
Let $w_i^l$ (resp., $w_i^r$) denote the set of points $p$ on $w_i$
whose shortest path from $s$ to $p$ is from the left (resp., right)
side of $w_i$. Hence, $q^*$ is either on $w_i^l$ or on $w_i^r$.

Our algorithm will find two points $q^*_l$ and $q^*_r$ such that if
$q^*$ is on $w_i^l$ for some $i\in [1,k]$, then $q^*=q^*_l$, and
otherwise (i.e.,$q^*$ is in $w_i^r$ for some $i\in [1,k]$),
$q^*=q^*_r$.

In the following, we will only present our algorithm for finding
$q^*_l$ since the case for $q^*_r$ is symmetric. In the following
discussion, we assume $q^*$ is on $w_i^l$ for some $i\in [1,k]$.
Note that this assumption is only for arguing the
correctness of our algorithm, which actually proceeds
without knowing whether the assumption is true.

The rest of this section is organized as follows. In Section~\ref{sec:obser}, we discuss some observations, based on which we describe our pruning algorithm in Section~\ref{sec:prune} to prune some
(portions of) segments of $W$ such that $q^*$ ($=q_l^*$) is still in the
remaining segments of $W$. In Section~\ref{sec:computeq}, we will
finally compute $q^*_l$ (which will be $q^*$) on the remaining segments of $W$.
Some implementation details of the algorithm are given in Sections~\ref{sec:bundle} and \ref{sec:imple}. Section~\ref{sec:wrapup} summarizes the overall algorithm.

As will be clear later, our algorithm uses extended windows instead of windows because extended windows can help us with the pruning.

\subsection{Observations}
\label{sec:obser}

%
%



For any point $t\in \calP$ with $s\neq t$, and its shortest path $\pi(s,t)$, we use
$t^+$ to denote a point on $\pi(s,t)$ infinitely close to $t$ (but
$t^+\neq t$).
If $t$ is on $w_i^l$ for some $i\in [1,k]$, then $t^+$ must be on the left side of
$w_i$.

For any segment $w$ of $W$, we say that $w$ or a sub-segment of $w$
can be {\em pruned} if it does not contain $q^*$.
Our pruning algorithm, albeit somewhat involved, is based on the following simple observation.

\begin{observation}\label{obser:basic}
For any point $t\in w_i^l$ for some $i\in [1,k]$, if $\pi(s,t^+)$
intersects any segment
$w\in W$ or an endpoint of it, then $t$ can be pruned (i.e., $t$ cannot be $q^*$).
\end{observation}
\begin{proof}
Let $t'$ be a point on $\pi(s,t^+)$ that is a point on
any segment $w\in W$ or an endpoint of it. Clearly, $t'\in \vis(s)$ and $d(s,t')<d(s,t)$.
Thus, $t$ cannot be $q^*$.
\qed
\end{proof}

\begin{figure}[t]
\begin{minipage}[t]{\linewidth}
\begin{center}
\includegraphics[totalheight=1.8in]{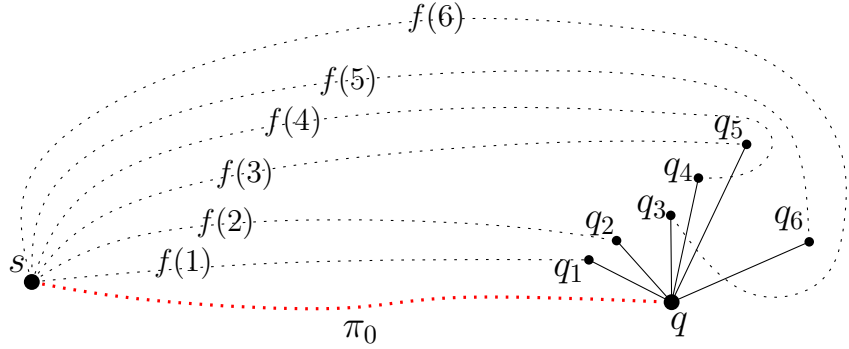}
\caption{\footnotesize
Illustrating the map $f(\cdot)$: $f(1)=1$, $f(2)=2$, $f(3)=5$, $f(4)=4$, $f(5)=6$, and $f(6)=3$. Note that the paths could be ``below'' $\pi_0$, but for ease of exposition, we ``flip'' them above $\pi_0$, and this flip operation does not change the topology of these paths.}
\label{fig:map}
\end{center}
\end{minipage}
\vspace*{-0.15in}
\end{figure}

Consider the shortest paths $\pi(s,q_i)$ for $i=1,2,\ldots,k$.
To simplify the notation, let $\pi_i=\pi(s,q_i)$ for each $i\in
[1,k]$. In particular, let $\pi_0=\pi(s,q)$ (not $\pi(s,q_0)$). Recall
that $Q=\{q,q_1,\ldots,q_k\}$. The
union of all paths $\pi_i$ for $0\leq i\leq k$ forms a planar tree, denoted
by $T_Q$, with root at $s$. Consider the canonical cycle $\calC(T_Q)$
as defined in Section~\ref{sec:pre}.
Let $\calC_Q$ be the circular list of the points of $Q$
following their relative order in $\calC(T_Q)$. We further break
$\calC_Q$ into a list $\calL_Q$ at $q$, such that $\calL_Q$ starts
from $q$ and all other points of $\calL_Q$ follow the counterclockwise
order in $\calC_Q$. Assume $\calL_Q$ is
$\{q,q_{f(1)},q_{f(2)},\ldots,q_{f(k)}\}$, i.e., the $(i+1)$-th point
of the list is $q_{f(i)}$; e.g., see Fig.~\ref{fig:map}. So $f(\cdot)$ essentially maps each point
of $Q\setminus\{q\}$
from its position in $\calL_Q$ to its position in the list
$\{q_{1},q_2,\ldots,q_{k}\}$. Hence, $f(1)\ldots,f(k)$ is a
permutation of $1,\ldots,k$, and $f(i)\neq f(j)$ if $i\neq
j$.  The reason we introduce the list
$\calL_Q$ is that intuitively, for any $1\leq i<j\leq k$,
the path $\pi_{f(j)}$ is {\em counterclockwise} from $\pi_{f(i)}$ with
respect to $\pi_0$ around $s$. For convenience, we let $f(0)=0$.


Later in Section~\ref{sec:imple} we will give the implementation
details for the following lemma.
\begin{lemma}\label{lem:map}
Given $\spm(s)$, after $O(n)$ time preprocessing, we can compute the
list $\calL_Q$ and thus determine the map $f(\cdot)$ in $O(k\log n)$
time.
\end{lemma}

\begin{observation}\label{obser:50}
For any $i\in [1,k]$, $\pi_0$ does not contain $q_i$ and $\pi_i$ does not contain $q$.
\end{observation}
\begin{proof}
Assume to the contrary that $\pi_0$ contains $q_i$ for some $i\in [1,k]$.
Since $q$ is in $\pi_0$, by Observation~\ref{obser:10}(2),
$\pi_0=\pi_i\cup \overline{q_iq}$. Recall that $\overline{qq_0}\in \pi_0$.
This implies that either $\overline{qq_0}$ contains $q_i$ or
$\overline{qq_i}$ contains $q_0$, which further implies the two
windows $\overline{u_0q_0}$ and $\overline{u_iq_i}$ are collinear.
This incurs contradiction since no two windows are collinear.
Hence, $\pi_0$ does not contain $q_i$.

Assume to the contrary that $\pi_i$ contains $q$. Then, since both $q$
and $q_i$ are in $\pi_i$, by Observation~\ref{obser:10}(2),
$\overline{qq_i}$ is in $\pi_i$. Hence,
$\pi_i=\pi_0\cup\overline{qq_i}$.
Recall that $u_0$ is the root of the cell of
$\spm(s)$ containing $q$, and $\pi_0=\pi(s,u_0)\cup \overline{u_0q}$.
Since $q$ is in the interior of $\calP$,
$\overline{u_0q}$ and $\overline{qq_i}$ must be collinear
since otherwise there would be a shorter path from $u_0$ to $q_i$
without containing $\overline{qq_i}$.
Recall that $u_i\in \overline{qq_i}$. Since $\overline{u_0q}$ and $\overline{qq_i}$ are
collinear, the three points $q$, $u_0$, and $u_i$ are collinear. But this contradicts
with our general position assumption that $q$ is not collinear with any two obstacle
vertices. 
\qed
\end{proof}


\begin{lemma}\label{lem:90}
Suppose $\pi_j$ contains $q_i$ with $i\neq j$ and $i,j\in
[1,k]$. If $i<j$, then $w_j$ can be pruned; otherwise, $w_i$ can
be pruned.
\end{lemma}
\begin{proof}
We first discuss the case $i<j$. Consider the region $D$ bounded by the
closed curve that is the union of $w_i$, $w_j$, and the subpath of $\pi_j$ between $q_i$
and $q_j$ (e.g., see Fig.~\ref{fig:regionD}(a)).
By Observation~\ref{obser:10}(1), $\pi_j$ does not cross
$\pi_0$. Since $i<j$, $w_j$ is clockwise from $w_i$ with respect to
$w_0$ (which is the last edge of $\pi_0$). Hence, $D$ must be locally on the left side of $w_j$.

Consider any point $t\in w^l_j$. We show that $t$ cannot be
$q^*$. Recall that $w_j$ is an
open segment, so $t$ is not $q$ or $q_j$.
Since $t\in w^l_i$, the point $t^+$ must be in $D$. By the definition of
$D$, $s$ is not in the interior of $D$. Hence, $\pi(s,t^+)$ must
intersect the boundary of $D$. Since $\pi(s,t^+)$ cannot cross the
subpath of $\pi_j$ between $q_i$ and $q_j$, $\pi(s,t^+)$ must intersect
$w_i$, $w_j$, or a point of $\{q,q_i,q_j\}$. By Observation~\ref{obser:basic}, $t$ cannot be
$q^*$.

\begin{figure}[t]
\begin{minipage}[t]{\linewidth}
\begin{center}
\includegraphics[totalheight=1.2in]{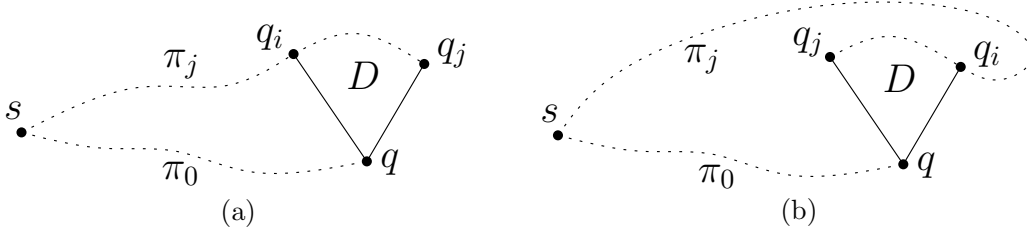}
\caption{\footnotesize
Illustrating the proof of Lemma~\ref{lem:90}: (a) the case $i<j$; (b) the case $i>j$.}
\label{fig:regionD}
\end{center}
\end{minipage}
\vspace*{-0.15in}
\end{figure}

The above shows that $t$ cannot be $q^*$.  Thus, $w_j$ can be
pruned.

For the case $i>j$, the argument is similar  (e.g., see Fig.~\ref{fig:regionD}(b)).
Since $i>j$, $D$ must be locally on the left side of $w_i$. For
any point $t\in w_i^l$, using the similar argument as above, we can
show that $t$ cannot be $q^*$. Thus, $w_i$ can be pruned.
\qed
\end{proof}

Lemma~\ref{lem:100} provides an algorithm to remove all extended-windows of $W$
that can be pruned by Lemma~\ref{lem:90}.

\begin{lemma}\label{lem:100}
Given $\spm(s)$ and with $O(n)$ time preprocessing, we can find in
$O(k\log n)$ time all segments of $W$ that can be pruned by
Lemma~\ref{lem:90}.
\end{lemma}
\begin{proof}
The task is to determine those indices $i$ and $j$ such that $q_i$ is contained
in $\pi_j$ for $i\neq j$ in $[1,k]$, after which we
can determine whether $w_i$ or $w_j$ should be pruned by Lemma~\ref{lem:90}.
Recall that $f(1),f(2),\ldots,f(k)$ is a permutation of the indices of
$\{1,2,\ldots,k\}$. Therefore, equivalently we can determine those indices $f(i)$ and $f(j)$ such that $q_{f(i)}$ is contained in $\pi_{f(j)}$ for $f(i)\neq f(j)$ in $[1,k]$.
We actually do not need to explicitly find all such pairs, as shown
below.

A key observation is that if $q_{f(i)}$ is contained in a path
$\pi_{f(j)}$ with $f(j)\neq f(i)$, then it must be that $j<i$ and
$q_{f(i)}$ is contained in $\pi_{f(m)}$ for any $m\in [j, i]$. Indeed,
if $q_{f(i)}$ is contained in a path $\pi_{f(j)}$ with $f(j)\neq f(i)$,
then the subpath of $\pi_{f(j)}$ from $s$ to $q_{f(i)}$ is
$\pi_{f(i)}$. According to the definition of the map $f(\cdot)$, i.e.,
the list $\{q_{f(1)},q_{f(2)},\ldots,q_{f(k)}\}$, $q_{f(i)}$
must be after $q_{f(j)}$ in the list, i.e., $j<i$. Further, for any
$m\in [j, i]$, $q_{f(i)}$ is in the path of the tree $T_Q$ from
$q_{f(m)}$ to the root $s$, which is the shortest path $\pi_{f(m)}$.

Based on the above observation, our algorithm works as follows. We
consider the points $q_{f(i)}$ in the order of $i=1,2,\ldots, k$.
Suppose we are about to process $q_{f(i)}$. The algorithm maintains a
stack $S$ of indices in $[1,i-1]$ in increasing order (from bottom to
top of $S$) such that for each $j\in [1,i-1]$, if $j\not\in S$, then
$w_{f(j)}$ has been pruned. Initially we set
$S=\emptyset$ before we process $q_{f(1)}$. In general, our algorithm
processes $q_{f(i)}$ for any $i\geq 1$ as follows.

If $S=\emptyset$, then we push $i$ on top of $S$ and proceed to
process $q_{f(i+1)}$. Otherwise, we first check whether $q_{f(i)}$ is
contained in $\pi_{f(m)}$, where $m$ is the top index on $S$.

\begin{enumerate}
\item
If $q_{f(i)}\not\in \pi_{f(m)}$, then $q_{f(i)}$ is not in any path
$\pi_{f(j)}$ with $j<m$ by the above observation. We push $i$ on top of $S$ and then proceed on processing $q_{f(i+1)}$.

\item
If $q_{f(i)}\in \pi_{f(m)}$, then depending on whether $f(i)<f(m)$, there are two cases.

\begin{enumerate}
\item
If $f(i)<f(m)$, then by Lemma~\ref{lem:90}, we prune $w_{f(m)}$ and
pop $m$ from $S$. Then, we repeat the same algorithm as above (i.e.,
first check whether $S=\emptyset$, and if not,  check whether
$q_{f(i)}$ is contained in $\pi_{f(m)}$, where $m$ is the new top
index of $S$).

\item
If $f(i)>f(m)$, then by Lemma~\ref{lem:90}, we prune $w_{f(i)}$ and
proceed on processing $q_{f(i+1)}$.
\end{enumerate}
\end{enumerate}

The algorithm finishes once $q_{f(k)}$ has been processed. It is not
difficult to see that if we can check
whether $q_{f(i)}$ is in $\pi_{f(m)}$ in $O(c)$ time, then the
algorithm runs in $O(k\cdot c)$ time since each index of $[1,k]$ can
be pushed or popped from $S$ at most once. In the following, we show
that $c=O(\log n)$ after $O(n)$ time preprocessing, and this will prove
the lemma.


First of all, if both $q_{f(i)}$ and $q_{f(m)}$ are in the
same cell $\sigma$ of $\spm(s)$, then $q_{f(i)}\in \pi_{f(m)}$ if and only if
$q_{f(i)}\in \overline{r(\sigma)q_{f(m)}}$, where $r(\sigma)$ is the root of $\sigma$.
Otherwise, if $q_{f(i)}$ is not in any edge of the shortest path tree
$\spt(s)$, then $q_{f(i)}$ cannot be in $\pi_{f(m)}$. Otherwise,
suppose $q_{f(i)}$ is on an edge $e$
of $\spt(s)$. We can find the edge $e$ in $O(\log n)$ time by a point location
query on the decomposition of $\spm(s)$ by the edges of $\spt(s)$.
Let $v$ be an endpoint of $e$, and thus $v$ is a node of $\spt(s)$. Let
$r$ be the root of the cell of $\spm(s)$ containing $q_{f(m)}$. Then,
$q_{f(i)}$ is in $\pi_{f(m)}$ if and only if $v$ is an
ancestor of $r$ in $\spt(s)$. Note that $v$ is an
ancestor of $r$ if and only if the lowest common ancestor of $v$ and
$r$ is $v$. We can build a data structure on $\spt(s)$ in $O(n)$ time such that given
any two nodes of the tree, the lowest common ancestor can be found in constant
time~\cite{ref:BenderTh00,ref:HarelFa84}.

Hence, we can determine whether $q_{f(i)}\in \pi_{f(m)}$ in $O(\log n)$ time
after $O(n)$ time preprocessing.

The lemma thus follows.
\qed
\end{proof}

We apply the algorithm in Lemma~\ref{lem:100} to prune the segments of
$W$. But to simplify the notation, we assume that none of the
segments of $W$ is pruned since otherwise we could re-index all
segments of $W$. So now $W$ has the following property.

\begin{observation}\label{obser:60}
For any $i\in [1,k]$, $q_i$ is not contained in any $\pi_j$ with
$j\in [0,k]$ and $j\neq i$.
\end{observation}
\begin{proof}
Suppose to the contrary that $q_i$ is contained in $\pi_j$ for some
$j\in [0,k]$ and $i\neq j$. On the one hand, due to Observation~\ref{obser:50},
$j\neq 0$. On the other hand, if $j\in [1,k]$, then by
Lemma~\ref{lem:90} either $w_i$ or $w_j$ would have already been removed from $W$. \qed
\end{proof}

For each $i\in [1,k]$, 
since $\pi_0$ does not cross $\pi_i$, $\pi_0\cup \pi_i\cup w_i$ forms
a closed curve that separates the plane into two regions, one locally on the left of $w_i$
and the other locally on the right $w_i$. We let $D_i$
denote the region locally on the left side of $w_i$ including
$\pi_0\cup \pi_i\cup w_i$ as its boundary (it is possible that $D_i$
is unbounded). If $\pi_0\cap \pi_i$ is
a sub-path including at least one edge, then it is also considered to be in $D_i$.
We have the following observation for $D_i$.

\begin{observation}
If $q^*\in w_i^l$, then $\pi(s,q^*)$ must be in $D_i$.
\end{observation}
\begin{proof}
Let $t=q^*$ that is on $w_i^l$. Then, 
$t^+$ is in the interior of $D_i$. By Observation~\ref{obser:basic},
$\pi(s,t^+)$ cannot intersect $w_i$. Also,
$\pi(s,t^+)$ cannot cross either $\pi_0$ or $\pi_i$, and $s$ is on the
boundary of $D_i$. Hence, $\pi(s,t^+)$ must be inside $D_i$. Thus,
$\pi(s,q^*)$ is in $D_i$.
\qed
\end{proof}



\begin{figure}[t]
\begin{minipage}[t]{\linewidth}
\begin{center}
\includegraphics[totalheight=1.2in]{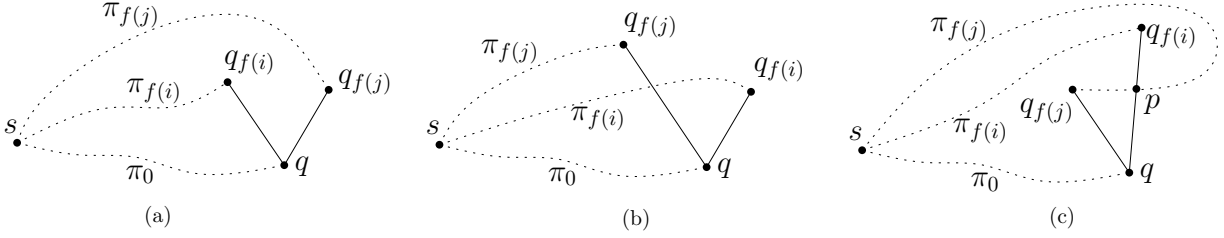}
\caption{\footnotesize
Illustrating Lemma~\ref{lem:key}.}
\label{fig:twopaths}
\end{center}
\end{minipage}
\vspace*{-0.15in}
\end{figure}

Our pruning algorithm mainly relies on the following lemma, whose proof
in turn boils down to Observation~\ref{obser:basic}.

\begin{lemma}\label{lem:key}
Suppose $i$ and $j$ are two indices with $1\leq i<j\leq k$.
\begin{enumerate}
\item
If $f(i)<f(j)$, then $\pi_{f(i)}$ does not
cross ${w_{f(j)}}$ and  $\pi_{f(j)}$ does not cross
${w_{f(i)}}$, and further, $D_{f(i)}$ is contained in $D_{f(j)}$ (e.g., see Fig.~\ref{fig:twopaths}(a)).
\item
If $f(i)>f(j)$, then either $\pi_{f(i)}$ crosses ${w_{f(j)}}$ or  $\pi_{f(j)}$ crosses $w_{f(i)}$.  Further, in the former case (e.g., see Fig.~\ref{fig:twopaths}(b)), $w_{f(i)}$ can be
pruned, and in the latter case (e.g., see Fig.~\ref{fig:twopaths}(c)), the sub-segment $\overline{qp}$ of $w_{f(i)}$ can be pruned, where $p$ is the point at which $\pi_{f(j)}$ crosses $w_{f(i)}$.
\end{enumerate}
\end{lemma}
\begin{proof}
Suppose $f(i)<f(j)$.
We first show that $q_{f(j)}$ cannot be in the interior of the region $D_{f(i)}$.

Assume to the contrary that $q_{f(j)}$ is in the interior of $D_{f(i)}$.
Let $p_{f(j)}$ be a point on ${w_{f(j)}}$ arbitrarily
close to $q$ (but $p_{f(j)}\neq q$). Since $f(i)<f(j)$, $w_{f(j)}$ is clockwise
from $w_{f(i)}$ with respect to $w_0$. Since $q$ is not in $\pi_{f(i)}$ by Observation~\ref{obser:50}, $p_{f(j)}$ is not in
$D_{f(i)}$. Since $q_{f(j)}$ is in the interior of $D_{f(i)}$, $\pi_{f(i)}$ must
cross $w_{f(j)}$ at a point $p$ with $p\neq q_{f(j)}$ (e.g., see Fig.~\ref{fig:pathcross}).

\begin{figure}[h]
\begin{minipage}[t]{\linewidth}
\begin{center}
\includegraphics[totalheight=1.0in]{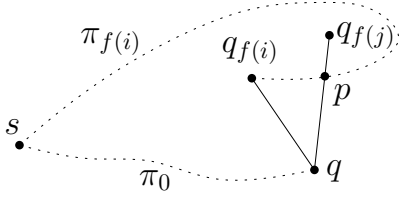}
\caption{\footnotesize
Illustrating the scenario where $q_{f(j)}$ is in the interior of $D_{f(i)}$.}
\label{fig:pathcross}
\end{center}
\end{minipage}
\vspace*{-0.15in}
\end{figure}

Depending on whether $p\in \pi_{f(j)}$, there are two cases.

\begin{enumerate}
\item
If $p\in \pi_{f(j)}$, then since $p\in w_{f(j)}$, we obtain
$\pi_{f(j)}=\pi(s,p)\cup \overline{pq_{f(i)}}$ by Observation~\ref{obser:10}(2). Since
$q_{f(j)}$ is in the interior of $D_{f(i)}$, we further obtain
that $\pi_{f(i)}$ is counterclockwise from $\pi_{f(j)}$  with respect
to $\pi_0$. Thus, we have $i>j$, a contradiction.

\item
If $p\not\in \pi_{f(j)}$, then since $i<j$ and
$\pi_{f(j)}$ is counterclockwise from $\pi_{f(i)}$ with respect to
$\pi_0$, $\pi_{f(j)}$ must cross an interior point $p'$ of $\overline{qp}$
before reaching $q_{f(j)}$. This implies that
$\pi_{f(j)}=\pi(s,p')\cup \overline{p'q_{f(j)}}$ by Observation~\ref{obser:10}(2), and thus, $\pi_{f(j)}$
contains $p$ since $p\in \overline{p'q_{f(i)}}$. Hence, we again
obtain contradiction.
\end{enumerate}

This proves that $q_{f(j)}$ cannot be in the interior of the
region $D_{f(i)}$.

By Observations~\ref{obser:50} and \ref{obser:60}, $q_{f(j)}$ cannot be in $\pi_0$ or $\pi_{f(i)}$. Since no segment of $W$ contains another, $q_{f(j)}$ cannot be in $w_{f(i)}$. Hence, $q_{f(j)}$ cannot be on the
boundary of $D_{f(i)}$. Therefore, $q_{f(j)}$ is outside $D_{f(i)}$.
Next we show that $\pi_{f(i)}$ does not cross ${w_{f(j)}}$.

Indeed,
since both  $q_{f(j)}$ and $p_{f(j)}$ are outside $D_{f(i)}$,
in order for $\pi_{f(i)}$ to cross ${w_{f(j)}}$,
$\pi_{f(i)}$ must cross ${w_{f(j)}}$ at least twice, which is not possible by Observation~\ref{obser:10}(2).
Similarly, in order for $\pi_{f(j)}$ to cross ${w_{f(i)}}$, it would
have to cross ${w_{f(i)}}$ at least twice, which is not possible.

This proves that $\pi_{f(i)}$ does not cross ${w_{f(j)}}$ and
$\pi_{f(j)}$ does not cross ${w_{f(i)}}$.
Since $w_{f(j)}$ is clockwise from $w_{f(i)}$ and $\pi_{f(j)}$ does
not cross $w_{f(i)}$, $w_{f(i)}$ is contained in $D_{f(j)}$.
Further, since $\pi_{f(j)}$ is counterclockwise from $\pi_{f(i)}$ and $\pi_{f(i)}$ does not cross $w_{f(j)}$, $D_{f(i)}$ must be contained in $D_{f(j)}$.

This proves the first part of the lemma.

For the second part of the lemma, we assume  $f(i)>f(j)$. By the same analysis as above, $q_{f(i)}$ cannot be on the boundary of $D_{f(j)}$.
Depending on whether $q_{f(i)}$ is in the interior of $D_{f(j)}$
or outside it, there are two cases.

\begin{enumerate}
\item
If $q_{f(i)}$ is outside $D_{f(j)}$, then since $i<j$,
$\pi_{f(j)}$ is counterclockwise from $\pi_{f(i)}$ with respect to
$\pi_0$. Further, since $\pi_{f(i)}$ and $\pi_{f(j)}$ do not
cross each other and $\pi_{f(i)}$ does not contain $q$ (by Observation~\ref{obser:50}), $\pi_{f(i)}$ must cross $w_{f(j)}$. Let $p$ be the point of  $w_{f(j)}$ where $\pi_{f(i)}$ crosses.
Let $D$ be the open region bounded by
$w_{f(i)}$, $\overline{qp}$, and the subpath $\pi'$ of
$\pi_{f(i)}$ between $p$ and $q_{f(i)}$.

Consider any point $t$ on $w^l_{f(i)}$ (if any).
The point $t^+$ must be in the interior of $D$.
Clearly, $s$ is not in $D$. Hence, $\pi(s,t^+)$ must cross the boundary of $D$. Since $\pi(s,t^+)$ cannot cross $\pi'$, it must cross either $\overline{pq}$ or $w_{f(i)}$. By Observation~\ref{obser:basic}, $t$ can be pruned.
Thus, $w_{f(i)}$ can be pruned.

\item
If $q_{f(i)}$ is in the interior of $D_{f(j)}$, let $p_{f(i)}$ be a point on $w_{f(i)}$ infinitely close to $q$. Since $f(i)>f(j)$, by the same analysis as before, $p_{f(i)}$ is not in $D_{f(j)}$.
Since $q_{f(i)}$ is in the interior of $D_{f(j)}$, $\overline{q_{f(i)}p_{f(i)}}$
must intersect the boundary of $D_{f(j)}$ at a point $p$.
Since $\overline{q_{f(i)}p_{f(i)}}$ does not intersect $\pi_0$ or $w_{f(j)}$,
$p$ is on $\pi_{f(j)}$. This proves that $\pi_{f(j)}$ crosses $w_{f(i)}$.

Consider the region $D$ bounded by $\overline{qp}$, $w_{f(j)}$, and the
subpath of $\pi_{f(j)}$ between $p$ and $q_{f(j)}$.

Consider any point $t$ on $\overline{qp}\cap w^l_{f(i)}$.
By the similar argument as above, we can show that $t$ can be pruned.
Thus, $\overline{qp}$ can be pruned.
\end{enumerate}
The lemma thus follows.
\qed
\end{proof}

For any $1\leq i<j\leq k$,  we say $\pi_i$ and $\pi_j$ are {\em
consistent} if $f(i)<f(j)$. By
Lemma~\ref{lem:key}, if $\pi_i$ and $\pi_j$ are not consistent, then we can do
some pruning, based on which we present our pruning algorithm in Section~\ref{sec:prune}. Figure~\ref{fig:afterprune} gives an example showing the remaining parts of the segments of $W$ after the pruning algorithm.

\begin{figure}[t]
\begin{minipage}[t]{\linewidth}
\begin{center}
\includegraphics[totalheight=2.3in]{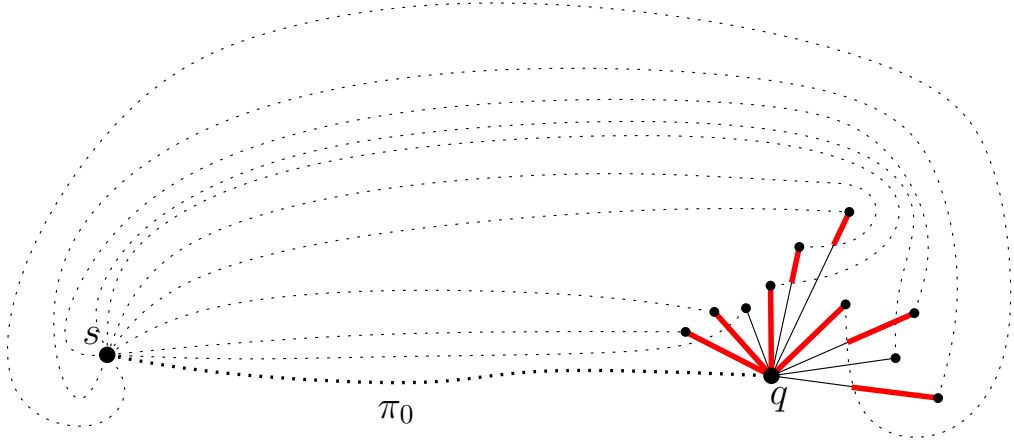}
\caption{\footnotesize
The thick (red) segments are the remaining parts of the segments of $W$ after the pruning algorithms (so that $q_l^*$ must be on the left side of a red segment).
Note that the paths could be ``below'' $\pi_0$, but for ease of exposition, we ``flip'' them above $\pi_0$, and this flip operation actually does not change the topology of these paths.}
\label{fig:afterprune}
\end{center}
\end{minipage}
\vspace*{-0.15in}
\end{figure}

\subsection{A Pruning Algorithm for Pruning the Segments of $W$}
\label{sec:prune}

We process the paths $\pi_{f(1)},\pi_{f(2)},\ldots,\pi_{f(k)}$ in this order.
Assume that $\pi_{f(i-1)}$ has just been processed and we are about to
process $\pi_{f(i)}$. Our algorithm maintains a sequence of {\em
bundles}, denoted by $\bbB=\{B_1,B_2,\ldots B_{g}\}$. Each {\em
bundle} $B\in \bbB$ is defined recursively as follows. Essentially $B$
is a list of sorted indices of a subset of $\{1,2,\ldots,i-1\}$, but the
indices are grouped in a special and systematic way.

There are two types of bundles: {\em atomic} and {\em composite}. If
$B$ has only one index, then it is an atomic bundle. Otherwise, $B$ is
a composite bundle consisting of a sequence of at least
two bundles $B'_1,\ldots,B'_{g'}$ (with $g'\geq 2$) such that the last
bundle $B'_{g'}$ must be atomic (others can be either atomic or
composite), and we call the index contained in $B'_{g'}$ the {\em wrap
index} of $B$. We consider the bundles $B'_1,\ldots,B'_{g'}$ as the
{\em children bundles} of $B$.

Let $f_{\min}(B)$ and $f_{\max}(B)$ denote the smallest and largest $f(j)$ of
all indices $j$ of $B$, respectively. If $B$ is composite,
then $B$ further has the following three {\em bundle-properties}. (1) The indices of $B$
are distinct and sorted increasingly by their order in $B$. (2) For
any $1\leq b<g'-1$, $f_{\max}(B'_b)<f_{\min}(B'_{b+1})$. (3) If $j$ is the wrap
index of $B$, then $f_{\min}(B)=f(j)$ and $\pi_{f(j)}$ crosses $w_{f(j')}$
for every $j'\in B\setminus\{j\}$ (intuitively, $\pi_{f(j)}$ ``wraps'' the
point $q_{f(j')}$, and this is why we call $j$ a ``wrap''
index). Refer to Fig.~\ref{fig:bundle} for an example.

\begin{figure}[t]
\begin{minipage}[t]{\linewidth}
\begin{center}
\includegraphics[totalheight=2.3in]{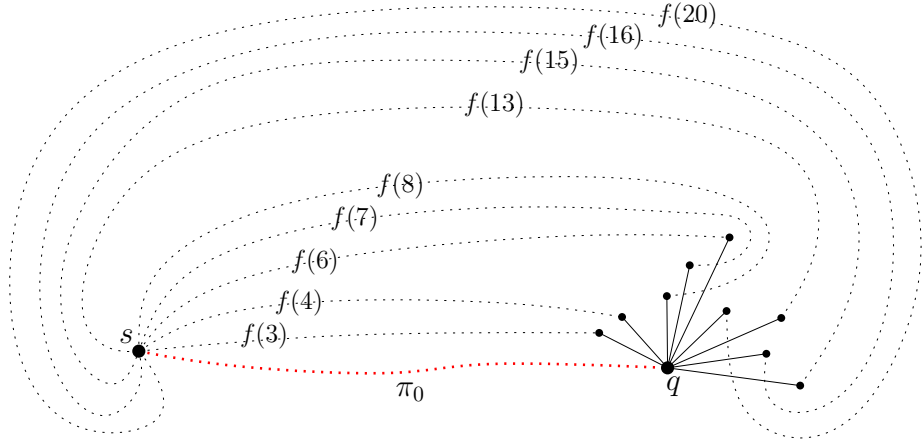}
\caption{\footnotesize
Illustrating the shortest paths corresponding to the indices in the current bundle sequence $\bbB=\{\underline{\{3\}}, \underline{\{4\}}, \underline{\{\{\{6\},\{7\}\},\{8\}\}},\underline{\{\{13\},\{\{15\},\{16\}\},\{20\}\}}         \}$, where each underline indicates a bundle of $\bbB$. For example, the last bundle is a composite bundle consisting of three children bundles with $20$ as its wrap index. In the figure, the indices of the paths are labeled. Note that the paths could be ``below'' $\pi_0$, but for ease of exposition, we ``flip'' them above $\pi_0$, and this flip operation actually does not change the topology of these paths.}
\label{fig:bundle}
\end{center}
\end{minipage}
\vspace*{-0.15in}
\end{figure}

For convenience, if the context is clear, we also consider a bundle
$B$ as a set of sorted indices. So if an index $j$ is in $B$, we can write ``$j\in B$''.

\paragraph{Remark.} We use the word ``bundle'' because each index $j$ of
$B$ refers to the shortest path $\pi_{f(j)}$. Therefore, $B$ is a ``bundle'' of shortest paths.
\paragraph{}

In addition, the bundle sequence $\bbB=\{B_1,B_2,\ldots,B_g\}$ maintained by our
algorithm has the following two {\em $\bbB$-properties}. (1) The indices in all
bundles are distinct in $[1,i-1]$ and are sorted increasingly by their order in the
sequence. (2) For any $1\leq b<g$, $f_{\max}(B_b)<f_{\min}(B_{b+1})$.

\begin{observation}
\begin{enumerate}
\item
For any $1\leq b<b'\leq g$ and any indices $j\in B_b$ and $j'\in B_{b'}$ (both $B_b$ and $B_{b'}$ are from $\bbB$), the two shortest paths
$\pi_{f(j)}$ and $\pi_{f(j')}$ are consistent  (e.g., see Fig.~\ref{fig:bundle}).

\item
For any composite bundle $B=\{B'_1,\ldots,B'_{g'}\}$,
for any $1\leq b<b'\leq g'-1$ and any indices $j\in B'_b$ and $j'\in
B'_{b'}$, the two shortest paths $\pi_{f(j)}$ and $\pi_{f(j')}$ are consistent (e.g., see Fig.~\ref{fig:bundle}).
\end{enumerate}
\end{observation}
\begin{proof}
We only prove the first part since the second part is similar.

Since $b<b'$, it holds that $j<j'$.
Clearly, $f(j)\leq f_{\max}(B_b)$ and $f_{\min}(B_{b'})\leq f(j')$. Since $b<b'$,
we have $f_{\max}(B_b)<f_{\min}(B_{b'})$. Therefore, we obtain
$f(j)<f(j')$. Thus, $\pi_{f(j)}$ and $\pi_{f(j')}$ are consistent.
\qed
\end{proof}

In the following, we discuss our algorithm for processing the shortest
path $\pi_{f(i)}$, during which $\bbB$ will be updated. Initially when $i=1$, we
simply set $\bbB$ to contain the only atomic bundle $B=\{1\}$ and this finishes our
processing for $\pi_{f(1)}$. In general when $i>1$, we do the following.

We first find the index $\beta$ such that $f_{\max}(B_{\beta})<f(i)<f_{\max}(B_{\beta+1})$.
Later in Section~\ref{sec:bundle}
we will give a data structure to maintain the bundle sequence $\bbB$
such that $\beta$ can be found in $O(\log n)$ time.

If $\beta=g$ (so $B_{\beta+1}$ does not exist in this case), then we add a new atomic bundle $B_{g+1}=\{i\}$ to the rear of
$\bbB$ and we are done with processing $\pi_{f(i)}$. Note that the two
$\bbB$-properties are maintained.

Otherwise, we check whether $f_{\min}(B_{\beta+1})<f(i)$. We have the following lemma.

\begin{lemma}\label{lem:120}
If $f_{\min}(B_{\beta+1})<f(i)$, then the extended-window $w_{f(i)}$ can be
pruned.
\end{lemma}
\begin{proof}
Assume that $f_{\min}(B_{\beta+1})<f(i)$. Since $f(i)<f_{\max}(B_{\beta+1})$, we have $f_{\min}(B_{\beta+1})<f(i)<f_{\max}(B_{\beta+1})$, which also implies that
$B_{\beta+1}$ is a composite bundle. Let $r$ be the wrap index of
$B_{\beta+1}$. Due to $f(r)=f_{\min}(B)$, it follows that
$f(r)<f(i)$. Since every index of $\bbB$ is smaller than $i$, $r<i$. By Lemma~\ref{lem:key}, $\pi_{f(r)}$ does not cross $w_{f(i)}$.

\begin{figure}[t]
\begin{minipage}[t]{\linewidth}
\begin{center}
\includegraphics[totalheight=1.2in]{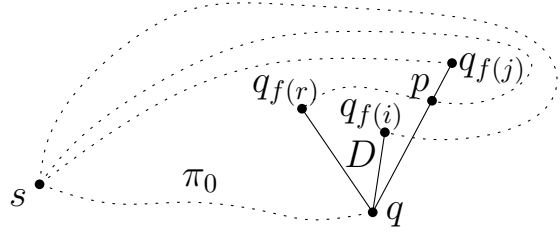}
\caption{\footnotesize
Illustrating the proof of Lemma~\ref{lem:120}.}
\label{fig:bundleprune}
\end{center}
\end{minipage}
\vspace*{-0.15in}
\end{figure}

Consider the index $j\in B$ with $f(j)=f_{\max}(B)$. Hence, $f(j)>f(i)$. By
the third bundle-property, $\pi_{f(r)}$ crosses $w_{f(j)}$, say, at a point $p$ (e.g., see Fig.~\ref{fig:bundleprune}). Consider the region $D$ bounded by $w_{f(r)}$, $\overline{pq}$, and the subpath of $\pi_{f(r)}$ between $p$ and $q_{f(r)}$. Since $r<i$ and
$f(r)<f(i)<f(j)$, $q_{f(i)}$ must be in $D$ since otherwise
$\pi_{f(r)}$ would cross $w_{f(i)}$, contradicting with Lemma~\ref{lem:key}(1). Also, by
Observation~\ref{obser:60}, $q_{f(i)}$ is not on $\pi_{f(r)}$.
Therefore, $q_{f(i)}$ is in the interior of $D$. This implies
that the shortest path from $s$ to any point $t$ of $w_{f(i)}$ must
intersect $w_{f(r)}$, $w_{f(j)}$, or their endpoints. Therefore, no point of
$w_{f(i)}$ can be $q^*$. Thus, $w_{f(i)}$ can be pruned.
\qed
\end{proof}

By Lemma~\ref{lem:120}, if $f_{\min}(B_{\beta+1})<f(i)$, we simply ignore
$\pi_{f(i)}$ and finish the processing of $\pi_{f(i)}$.

In the following, we assume $f(i)<f_{\min}(B_{\beta+1})$ (note that
$f(i)=f_{\min}(B_{\beta+1})$ is not possible since $i\not\in \bbB$).
Next, we are going to find all such indices $j$ of
$\bbB$ that $\pi_{f(j)}$ crosses $w_{f(i)}$. To this end,
the following two lemmas are crucial.

\begin{lemma}\label{lem:130}
\begin{enumerate}
\item
For any index $j$ in $B_b$ for any $b\in [1, \beta]$,
$\pi_{f(j)}$ does not cross $w_{f(i)}$.

\item

For any index $j$ in $B_b$ for any $b\in [\beta+1,g]$, if $\pi_{f(j)}$
crosses $w_{f(i)}$, then $w_{f(j)}$ can be pruned; otherwise,
$\pi_{f(i)}$ must cross $w_{f(j)}$.


\item
If $j$ is in $B_{b}$ for some $b\in [\beta+2,g]$ and $\pi_{f(j)}$ crosses $w_{f(i)}$, then $\pi_{f(j')}$ crosses $w_{f(i)}$ for any $j'\in B_{b'}$ and any $b'\in [\beta+1,b-1]$.

\item
If $j$ is in $B_{b}$ for some $b\in [\beta+1,g-1]$ and $\pi_{f(j)}$ does not cross $w_{f(i)}$, then $\pi_{f(j')}$ does not cross $w_{f(i)}$ for any $j'\in B_{b'}$ and any $b'\in [b+1,g]$.
\end{enumerate}
\end{lemma}
\begin{proof}
We prove the four parts of the lemma separately.
\begin{enumerate}
\item
If $j$ is in a bundle $B$ of
$\{B_1,B_2,\ldots,B_{\beta}\}$. Note that $j<i$. Since $f(j)\leq f_{\max}(B)$
and $f_{\max}(B)\leq f_{\max}(B_{\beta})<f(i)$, we obtain $f(j)<f(i)$. Consequently, by
Lemma~\ref{lem:key}(1), $\pi_{f(j)}$ does not cross ${w_{f(i)}}$.

\item
If $j$ is in a bundle $B$ of $\{B_{\beta+1},B_{\beta+2},\ldots,B_{g}\}$, then
$f(j)>f(i)$. Since $j<i$, according to Lemma~\ref{lem:key}(2), either $\pi_{f(j)}$ crosses
${w_{f(i)}}$ or $\pi_{f(i)}$ crosses ${w_{f(j)}}$.
If $\pi_{f(j)}$ crosses ${w_{f(i)}}$, by Lemma~\ref{lem:key}(2),
$w_{f(j)}$ can be pruned. Otherwise, $\pi_{f(i)}$ must cross ${w_{f(j)}}$.

\item
Let $j$ and $j'$ be the indices as in the lemma statement.
Our goal is to show that $\pi_{f(j')}$  crosses $w_{f(i)}$.

\begin{figure}[t]
\begin{minipage}[t]{\linewidth}
\begin{center}
\includegraphics[totalheight=1.0in]{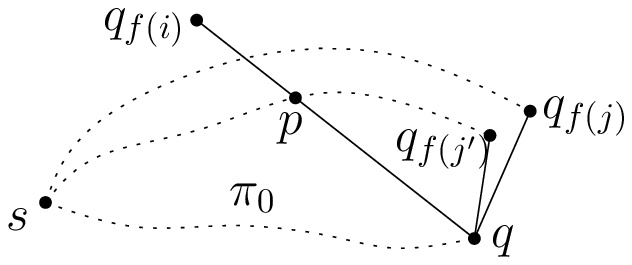}
\caption{\footnotesize
Illustrating the proof of Lemma~\ref{lem:130}.}
\label{fig:contained}
\end{center}
\end{minipage}
\vspace*{-0.15in}
\end{figure}

Clearly, $j'<j$ and $f(j')<f(j)$. By Lemma~\ref{lem:key}(1),
$D_{f(j')}$ is contained in $D_{f(j)}$ (e.g., see Fig.~\ref{fig:contained}).
Since $f(i)<f(j')$ and
$f(i)<f(j)$, if we move from $q$ to $q_{f(i)}$ along $w_{f(i)}$, we will enter
the interior of both $D_{f(j)}$ and $D_{f(j')}$. If we keep moving,
note that we cannot encounter any point in either $w_{f(j')}$ or
$w_{f(j)}$. Since $\pi_{f(j)}$ crosses $w_{f(i)}$, if we move as above
on $w_{f(i)}$, we will encounter a point on  $\pi_{f(j)}$, which is
part of the boundary of $D_{f(j)}$. Since $D_{f(j')}$ is contained in
$D_{f(j)}$, the above moving will also encounter a point $p$ on
$D_{f(j')}$  (e.g., see Fig.~\ref{fig:contained}). Due to Observation~\ref{obser:60}, $p$ cannot be
$q_{f(i)}$. Hence, $\pi_{f(j')}$ must cross $w_{f(i)}$ at $p$.

\item
This part is equivalent to the above third part.
\end{enumerate}
\qed
\end{proof}

For any bundle $B$ in $\{B_{\beta+1},B_{\beta+2},\ldots,B_{g}\}$, if
$B$ has two indices $j$ and $j'$ such that $w_{f(i)}$ crosses
$\pi_{f(j)}$ but does not cross $\pi_{f(j')}$, then we say that $B$ is
a {\em mixed} bundle, which is necessarily a
composite bundle.

\begin{lemma}\label{lem:140}
For any mixed bundle $B=\{B'_1,B'_2,\ldots, B'_{g'}\}$, the following holds.
\begin{enumerate}
\item
The path $\pi_{f(r)}$ must cross $w_{f(i)}$, where $r$ is the wrap index of $B$, i.e., $B'_{g'}=\{r\}$.

\item
If an index $j$ is in $B'_{b}$ for some $b\in [2,g'-1]$ and $\pi_{f(j)}$ crosses $w_{f(i)}$, then $\pi_{f(j')}$ crosses $w_{f(i)}$ for any $j'\in B'_{b'}$ and any $b'\in [1,b-1]$.

\item
If an index $j$ is in $B'_{b}$ for some $b\in [1,g'-2]$ and
$\pi_{f(j)}$ does not cross $w_{f(i)}$, then $\pi_{f(j')}$ does not
cross $w_{f(i)}$ for any $j'\in B'_{b'}$ and any $b'\in [b+1,g'-1]$.

\item
If a bundle $B'$ of $B$ has two indices $j$ and $j'$ such that $w_{f(i)}$ crosses $\pi_{f(j)}$ but does not cross $\pi_{f(j')}$, then we also say that $B'$ is a {\em mixed} bundle. This lemma applies to $B'$ recursively.
\end{enumerate}
\end{lemma}
\begin{proof}
\begin{enumerate}
\item
Suppose $j$ is an index of $B$ such that $\pi_{f(j)}$ crosses ${w_{f(i)}}$.
If $j=r$, then we are done with the proof. In the
following, we assume $j\neq r$. Hence, $f(j)>f(r)$.

\begin{figure}[t]
\begin{minipage}[t]{\linewidth}
\begin{center}
\includegraphics[totalheight=1.3in]{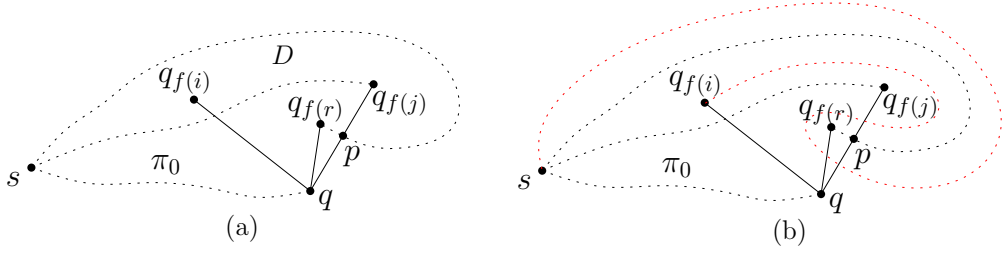}
\caption{\footnotesize
Illustrating the proof of Lemma~\ref{lem:140}(1): the path $\pi_{f(i)}$ is marked with red color in (b).}
\label{fig:wrapcross}
\end{center}
\end{minipage}
\vspace*{-0.15in}
\end{figure}

Assume to the contrary that $\pi_{f(r)}$ does not cross $w_{f(i)}$.
Since $r$ is the wrap index,
$\pi_{f(r)}$ crosses $w_{f(j)}$, say, at a point $p$ (e.g., see Fig.~\ref{fig:wrapcross}(a)). Consider the region $D$
bounded by $\pi_{f(j)}$, $\overline{pq_{f(j)}}$, and the subpath of $\pi_{f(r)}$
between $s$ and $p$, such that $D$ is on the right side of the
directed segment $\overline{pq_{f(j)}}$ from $p$ to
$q_{f(j)}$. Since $f(i)<f(r)<f(j)$ and $w_{f(i)}$ crosses
$\pi_{f(j)}$ but does not cross $\pi_{f(r)}$, $q_{f(i)}$ must be in the
region $D$. Since $i>j$ and $i>r$, if we go from $q_{f(i)}$ to $s$
along $\pi_{f(i)}$, we will get out of $D$ by crossing
$\overline{pq_{f(j)}}$, after which we get into the interior of the region $D_{f(j)}$ since $\pi_{f(i)}$ cannot cross $\pi_{f(r)}$  (e.g., see Fig.~\ref{fig:wrapcross}(b)). If we keep moving
towards $s$ along $\pi_{f(i)}$, before reaching $s$ we will need to
get out of the interior of $D_{f(j)}$ through $w_{f(j)}$ again.
However, due to Observation~\ref{obser:10}(2), since $\pi_{f(i)}$ already crosses $w_{f(j)}$ somewhere on $\overline{pq_{f(j)}}$,
it cannot intersect $w_{f(j)}$ again. Thus, we obtain contradiction.

\item
This part follows the similar proof as the third part of
Lemma~\ref{lem:100} and we omit the details.

\item
This part is equivalent to the second part of the lemma.

\item
Using the same analysis, we can prove that the same lemma
applies to $B'$ recursively.
\end{enumerate}
\qed
\end{proof}

In light of the preceding two lemmas, in the following we will find the indices $j$ of
$\bbB$ such that $\pi_{f(j)}$ crosses $w_{f(i)}$ and then prune
$w_{f(j)}$ by Lemma~\ref{lem:130}(2) (i.e., remove $j$
from $\bbB$); we say that such an index $j$ is {\em prunable}.

Before describing our algorithm, we first discuss an operation that
will be used in the algorithm. Consider a composite bundle
$B=\{B_1',B_2',\ldots,B'_{g'}\}$ of $\bbB$.
Let $r$ be a wrap index of $B$, i.e., $B'_{g'}=\{r\}$. Suppose $w_{f(i)}$ crosses
$\pi_{f(r)}$. Our algorithm will remove $r$ from $B$ and thus from
$\bbB$. This is done by a {\em wrap-index-removal} operation.
Further, suppose $B$ is the $j$-th bundle of $\bbB$, i.e., $B=B_j$.
After $r$ is removed, the operation will implicitly insert the bundles
$B_1',B_2',\ldots,B'_{g'-1}$ into the position of $B$ in the bundle
list $\bbB$, i.e., after the operation, $\bbB$ becomes
$B_1,\ldots,B_{j-1},B'_1,\ldots,B'_{g'-1},B_{j+1},\ldots,B_g$.
Note that this new bundle list still has the two $\bbB$-properties. Indeed,
$f_{\max}(B_{j-1})<f_{\min}(B)=f(r)<f_{\min}(B_1')$ and $f_{\max}(B'_{g'-1})\leq
f_{\max}(B)<f_{\min}(B_{j+1})$. 
Later in Section~\ref{sec:bundle}
we will give a data structure to maintain the bundles of $\bbB$ so that
each wrap-index-removal operation can be implemented in $O(\log n)$
time.

Another operation that is often used in the algorithm is the
following. Given any $i,j\in [1,k]$, we want to determine
whether $w_{f(i)}$ crosses $\pi_{f(j)}$. We call it the {\em shortest path
segment intersection} (or {\em SP-segment-intersection}) query.  Later
in Section~\ref{sec:imple} we will present an algorithm that can
answer
each such query in $O(\log h\log  n)$ time, after $O(n\log h)$ time
and space preprocessing.

We are ready to describe our algorithm for removing all prunable indices
from $\bbB$.  By Lemma~\ref{lem:130}(1), each bundle $B_b$ of $\bbB$ for $1\leq
b\leq \beta$ does not contain any prunable index. For each bundle $B$
of $B_{\beta+1},B_{\beta+2},\ldots, B_g$ in order, we call a procedure
{\em prune($B$)} until the procedure returns ``false''.

If all indices of $B$ are prunable, then {\em prune($B$)}
will return ``true'' and the entire bundle $B$ will be removed from $\bbB$.
Otherwise, the procedure will return false.
Further, if $B$ is a mixed bundle, then all prunable indices of $B$ will be removed (and the procedure returns false).

The procedure $prune(B)$ works as follows (see
Algorithm~\ref{algo:10} for the pseudocode). It is a recursive
procedure, which is not surprising since the bundles are defined
recursively. As a base case, if $B$ is an atomic bundle $\{j\}$, then
we call an SP-segment-intersection query to check whether
$\pi_{f(j)}$ crosses $w_{f(i)}$. If yes, we remove the bundle $B$ and
return true; otherwise, we return false. If $B$ is a composite bundle $\{B'_1,B'_2,\ldots,B'_{g'}\}$ with $r$
as the wrap index (i.e., $B'_{g'}=\{r\}$), then we first call an SP-segment-intersection to
check whether $\pi_{f(r)}$ crosses $w_{f(i)}$. If not, by
Lemma~\ref{lem:140}(1), $B$ does not have any prunable index and thus we
simply return false. If yes, then we call a wrap-index-removal operation to
remove $B'_{g'}$. Afterwards, for each $b'=1,2,\ldots,g'-1$ in order,
we call $prune(B'_{b'})$ recursively. If $prune(B'_{b'})$ returns false, then we
return false (without calling $prune(B'_{b'+1})$). If it returns true, we
remove $B'_{b'}$ (in fact all children bundles of $B'_{b'}$ have been
removed by $prune(B'_{b'})$). If $b'=g'-1$, then we return true
(since all children bundles of $B$ have been removed); otherwise,
we proceed on calling $prune(B'_{b'+1})$.

\begin{algorithm}[h]
\caption{The procedure $prune(B)$}
\label{algo:10}
\KwIn{A bundle $B$}
\KwOut{remove all prunable indices of $B$} \BlankLine
\eIf{$B$ is an atomic bundle $\{j\}$} 
{
  \eIf(\tcc*[f]{call an SP-segment-intersection query}){$\pi_{f(j)}$ crosses $w_{f(i)}$}
  {
     remove B\;
     return true\;
  }
  {
     return false\;
  }
}
{
 Let $B=\{B'_1,B'_2,\ldots,B'_{g'}\}$ and $B'_{g'}=\{r\}$\;
 \eIf(\tcc*[f]{call an SP-segment-intersection query}){$\pi_{f(r)}$ does not cross $w_{f(i)}$}
 {
   return false\;
 }
 {
   remove $B'_{g'}$; \tcc*[f]{perform a wrap-index-removal
   operation}

   \For{$b'\leftarrow 1$ \KwTo $g'-1$}
   {
      \eIf{$prune(B'_{b'}) = $ false}
	  {
	      return false\;
	  }
	  {
	      remove $B_{b'}'$;
	   }
   }
   return true\;
 }
}
\end{algorithm}


%

If {\em prune($B_{b}$)} returns true for every $b$ with $\beta+1\leq b\leq g$, then we add a new atomic bundle $\{i\}$ at the end of $\bbB$, which now becomes
$\{B_1,B_2,\ldots,B_{\beta},\{i\}\}$. This also finishes our
preprocessing for $\pi_{f(i)}$. Otherwise, {\em prune($B_{b}$)} returns false for some $b$ with $\beta+1\leq b\leq g$. In this case, as a final step, we create a new
composite bundle $B$, consisting of all bundles of $\bbB$
after $B_{\beta}$ (not including $B_{\beta}$)
and the atomic bundle $\{i\}$ as the last child bundle of $B$. This is done
by a {\em bundle-creation} operation. We will show in
Section~\ref{sec:bundle} that this operation can be implemented in $O(\log n)$ time. Afterwards, the new bundle sequence $\bbB$ becomes
$\{B_1,B_2,\ldots,B_{\beta},B\}$.
The following lemma shows that the new bundle $B$ is a ``valid''
composite bundle and the updated $\bbB$ maintains the two $\bbB$-properties.

\begin{lemma}
The new bundle $B$ has the three bundle properties and the updated $\bbB$ has the two
$\bbB$-properties.
\end{lemma}
\begin{proof}
Let $B=\{B_1',B_2',\ldots,B_{g'}'\}$, where $B'_{g'}=\{i\}$.
We show that $B$ has the three properties of composite bundles as follows.

\begin{enumerate}
\item
Indeed, recall that every index of the original $\bbB$ is smaller than
$i$. Note that
although some indices have been removed from $\bbB$, we never change
any relative order of two indices of $\bbB$. Further, $i$ is the last index of $B$. Therefore, the indices of
$B$ are sorted increasingly by their order in $B$. Hence, $B$ has the first property.

\item
To show the second property, again the bundles
$B_1',B_2',\ldots,B_{g'-1}'$, which are from the original $\bbB$,
never change their relative orders. By the recursive definition of
bundles, it holds that $f_{\max}(B'_{b'})<f_{\min}(B'_{b'+1})$ for any $1\leq b'<g'-1$.
Thus, the second property also holds on $B$.

\item
For the third property, recall that $f(i)<f_{\min}(B_{\beta+1})$. Since
each $B'_{b'}$ with $1\leq b'\leq g'-1$ is a ``descendent'' bundle of
$B_{b}\in \bbB$ (we consider $B_{b}$ a descendent bundle of itself)
for some $b\in [\beta+1,g]$, it holds that $f_{\min}(B_{\beta+1})\leq
f_{\min}(B_b)$. Since $f(i)<f_{\min}(B_{\beta+1})$, $f(i)< f_{\min}(B_b)$.
Therefore, $f_{\min}(B)=f(i)$. Further, for each $j\in B\setminus \{i\}$,
since $j$ is not prunable (otherwise $j$ would have already been pruned), $\pi_{f(j)}$ does not cross $w_{f(i)}$ (by
Lemma~\ref{lem:130}(2)). By
Lemma~\ref{lem:130}(2), $\pi_{f(i)}$ must cross $w_{f(j)}$. Hence, the
third property holds on $B$.
\end{enumerate}

To see that the updated bundle sequence $\bbB$ maintains the two
$\bbB$-properties, by using the similar analysis as above, the first property
holds. For the second property, we have proved above that
$f_{\min}(B)=f(i)$. Further, recall that $f_{\max}(B_{\beta})<f(i)$. Therefore,
we obtain $f_{\max}(B_{\beta})<f_{\min}(B)$. Consequently, the second property
also holds on $\bbB$.
\qed
\end{proof}

To analyze the running time of the above algorithm,
let $m$ be the number of indices that have been removed from $\bbB$. Then, the
algorithm makes at most $m+1$ SP-segment-intersection queries. To see
this, once the query discovers an index $j$ that is not prunable, the
algorithm will stop without making any more such queries. On the other
hand, each wrap-index-removal operation removes an index, and thus the
number of such operations is at most $m$. Further, observe that for each bundle
$B$, whenever we make a recursive call on a child bundle of $B$,
the wrap index of $B$ is guaranteed to be removed. Therefore, the
number of total recursive calls is at most $m$ as well. Hence, the running time
of the algorithm is $O((m+1)\log h\log n)$.

This finishes our algorithm for processing the path $\pi_{f(i)}$.
The total time for processing $\pi_{f(i)}$ is  $O((m+1)\log h\log n)$.
Since once an index is removed from $\bbB$, it will never be inserted into $\bbB$ again, the sum of all such $m$ in the entire algorithm for processing all paths $\pi_{f(i)}$ for $i=1,2,\ldots, k$ is at most $k$. Hence, the total time of the entire algorithm is $O(k\log h\log n)$.

Again, Fig.~\ref{fig:afterprune} gives an example showing the remaining parts of the segments of $W$ after the pruning algorithm.

\subsection{The Data Structure for Maintaining the Bundles}
\label{sec:bundle}

In this section, we give a data structure for maintaining the bundle
sequence $\bbB$ such that our algorithm runs in the time as claimed above.
In particular, we show that during our algorithm for processing $\pi_{f(i)}$
each of the following operations can be performed in $O(\log k)$ ($=O(\log n)$) time: inserting a new bundle $\{i\}$ at the end of $\bbB$, the bundle-creation operation,
the wrap-index-removal operation, 
and finding the index $\beta$.

We first present our data structure and then discuss the operations.

\subsubsection{The Data Structure}

Let $\bbB=\{B_1,B_2,\ldots,B_g\}$. It is not difficult to see that the bundles of $\bbB$ naturally form a tree structure. So we use a {\em bundle tree} $T_{\bbB}$ to represent it, as follows. The tree $T_{\bbB}$ has a root $\gamma$, whose children from left to right are exactly the bundles $B_1,B_2,\ldots,B_g$ in this order.
For each such bundle $B$, if $B$ is atomic, then $B$ is a leaf of $T_{\bbB}$ and the index of $B$ is stored at the leaf. Otherwise, suppose $B=\{B_1',B_2',\ldots,B_{g'}'\}$. Then, we store the wrap index of $B$ at the node $B$ and $B$ has $g'-1$ children from left to right corresponding to $B_1',B_2',\ldots,B_{g'-1}'$ in this order. If one of these bundles is composite, then its subtree is defined recursively. Refer to Fig.~\ref{fig:bundletree} for an example.

For each node $\mu$ of $T_{\bbB}$, let $T_{\bbB}(\mu)$ denote the subtree rooted at $\mu$.
It is easy to see that if $\mu$ is a leaf, then $T_{\bbB}(\mu)$ represents an atomic bundle; otherwise, $T_{\bbB}(\mu)$ represents a composite bundle.
Each node of the tree except the root stores an index. Further, the post-order traversal of each subtree $T_{\bbB}(\mu)$ gives exactly the sequence of indices in the bundle represented by $T_{\bbB}(\mu)$.

\begin{figure}[t]
\begin{minipage}[t]{\linewidth}
\begin{center}
\includegraphics[totalheight=1.3in]{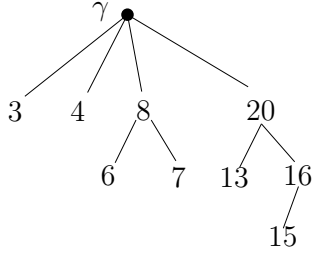}
\caption{\footnotesize
Illustrating the bundle tree $T_{\bbB}$ for the bundle sequence in Fig.~\ref{fig:bundle}.}
\label{fig:bundletree}
\end{center}
\end{minipage}
\vspace*{-0.15in}
\end{figure}

We implement the bundle tree $T_{\bbB}$ as follows. In general,
consider any internal node $\mu$. We let $\mu$ have two pointers $front$ and
$rear$ pointing to the leftmost and rightmost children of $\mu$,
respectively.
In this way, from $\mu$, we can access its leftmost and
rightmost children in $O(1)$ time.
All children of $\mu$ are organized by a doubly
linked list: Each child of $\mu$ maintains a {\em left} (resp., {\em
right}) pointer pointing to its left (resp. right) sibling, so that we can remove a node in constant time; the left
(resp., right) pointer of the leftmost (resp., rightmost) child is
empty. In this way, from the leftmost child of $\mu$, we can visit
all children of $\mu$ in order from left to right in linear time.

In order to compute the index $\beta$ in $O(\log k)$ time,
we use another balanced binary search tree $T_f$ to maintain the ranges
$[f_{\min}(B),f_{\max}(B)]$ of the bundles $B$ of $\bbB$.
The tree $T_{f}$ has $g$ leaves corresponding to the bundles of $\bbB$
from left to right. For each leave $v\in T_{f}$, let $B_v$ denote the bundle of $\bbB$
corresponding to $v$; we associate with $v$ the range $[f_{\min}(B_v),f_{\max}(B_v)]$.
By the second property of $\bbB$, the ranges of the leaves from left to right are
sorted by either the minimum values or the maximum values of the
ranges. Clearly, the height of $T_f$ is $O(\log k)$. In addition,
each leave $v$ is associated with a {\em cross pointer}
pointing to the node of $T_{\bbB}$ corresponding to the bundle $B_v$,
so that once we
have the access to $v$ in $T_f$ we can locate $B_v$ in $T_{\bbB}$ in
constant time. Finally, each internal node $v$ of $T_f$ maintains the
minimum range value of the leftmost leave in the subtree of $T_f$ rooted at
$v$, which is used for searching.


This completes our data structure for maintaining the bundles of
$\bbB$, which consists of two trees $T_{\bbB}$ and $T_f$.
In the following, we show how to use our data structure to implement
the operations on $\bbB$ needed in our algorithm for processing $\pi_{f(i)}$.

\subsubsection{Performing Operations}

First of all, finding the index $\beta$ can be easily done in $O(\log
k)$ time by searching the tree $T_f$. Further, by using the cross pointer, we can
immediately access the node $\mu$ of $T_{\bbB}$ whose subtree
$T_{\bbB}(\mu)$ represents $B_{\beta}$.

If $\beta=g$, then our algorithm adds $B=\{i\}$ at the end of $\bbB$. To
implement it,  we first
insert $B$ to $T_f$ as the rightmost leaf with the range $[f(i),f(i)]$,
which can be done in $O(\log k)$ time.
Then, we add the atomic bundle $B$ to the rear of $\bbB$ by adding a
leaf to $T_{\bbB}$ as the rightmost child of the root $\gamma$. The
tree $T_{\bbB}$ can be updated
in constant time with the help of the rear pointer of $\gamma$.

If $\beta\neq g$, then we check whether $f_{\min}(B_{\beta+1})<f(i)$ (note
that we can find the leaf for $B_{\beta+1}$ in $T_f$ in
$O(\log k)$ time). If $f_{\min}(B_{\beta+1})<f(i)$, then we are done for processing $\pi_{f(i)}$.
In the following, we assume $f_{\min}(B_{\beta+1})>f(i)$.

Our algorithm first calls the procedure $prune(B_{\beta+1})$. To
implement it, note that $B_{\beta+1}$ is represented by the subtree
$T_{\bbB}(\mu')$, where $\mu'$ is the right sibling of $\mu$.
Since we already have the access to $\mu$,
by using the right pointer of $\mu$, we can access $\mu'$ in
constant time.  The procedure
$prune(B_{\beta+1})$ begins with checking whether $B_{\beta+1}$ is
atomic, which can be done in constant time by checking whether $\mu'$
is a leaf.

If yes, then the procedure stops after an SP-segment-intersection query.
Further, if $B_{\beta+1}$ needs to be removed, then we simply remove
the leaf $\mu'$, which can be done in constant time
(recall that the children of any node of $T_{\bbB}$
are organized by a doubly linked list).
Further, we also remove the corresponding leaf from $T_f$ in $O(\log k)$ time.

If $B_{\beta+1}$ is not atomic, let $B_{\beta+1}=\{B_1',B_2',\ldots,B'_{g'}\}$.
We can obtain the wrap index of $B_{\beta+1}$ in constant time since
it stored at the node $\mu'$.
To implement wrap-index-removal operation, essentially, we need to replace the
node $\mu'$ by its children. This can be done in constant time by
using the left, right, front, and rear pointers of $\mu'$. Depending
on whether $\mu'$ is the leftmost or rightmost child of $\gamma$, we
may also need to update the front or rear pointer of $\gamma$, which
can also be easily done in constant time. We omit these details.

Next, our algorithm calls the procedure $prune(B_1')$. We can access
the node of $T_{\bbB}$ whose subtree represents $B_1'$ in constant time
after the above wrap-index-removal operation (i.e., by following the
front pointer of $\mu'$).
The algorithm then works recursively. Note that
$B'_1$ now becomes a bundle of $\bbB$. Hence, the above algorithm
description on $B_{\beta+1}$ applies to $B_1'$ recursively.


The algorithm stops when either we are at the end of $\bbB$
or the procedure $prune(B')$ returns false for a bundle $B'$ in the current $\bbB$.
In the former case, we add $\{i\}$ to the rear of the current list
$\bbB$ in the same way as before.  In the latter case, we preform a
bundle creation operation by creating a composite bundle $B$ including
all bundles of the current $\bbB$ after $B_{\beta}$ as well as $\{i\}$ in the
rear of $B$. We implement this bundle creation operation as follows.

Note that we have the access of the node $\mu_1$ whose subtree
represents $B'$ after $prune(B')$ returns false. Let $\mu_2$ be the
rightmost child of $\gamma$, which can be accessed in constant time from the root $\gamma$.
Next, in constant time, we construct a subtree $T$ representing the
bundle $B$ and use $T$ to replace the subtrees of $\gamma$ from
$\mu_1$ to $\mu_2$ (e.g., see Fig.~\ref{fig:buncreate}), as follows.
First, we create a new node $\mu_3$ storing the single index $i$.
Second, we set the front pointer of
$\mu_3$ to $\mu_1$ and set the rear pointer of $\mu_3$ to $\mu_2$.
Third, if $\mu_1$ has a left sibling, denoted by $\mu_4$, then
we set the left pointer of $\mu_3$ to $\mu_4$ and set the right
pointer of $\mu_4$ to $\mu_3$; otherwise, we set the front pointer of
$\gamma$ to $\mu_3$.
Fourth, we set the rear pointer of $\gamma$ to $\mu_3$.
Fifth, we set the left pointer of $\mu_1$ to empty.

\begin{figure}[t]
\begin{minipage}[t]{\linewidth}
\begin{center}
\includegraphics[totalheight=1.3in]{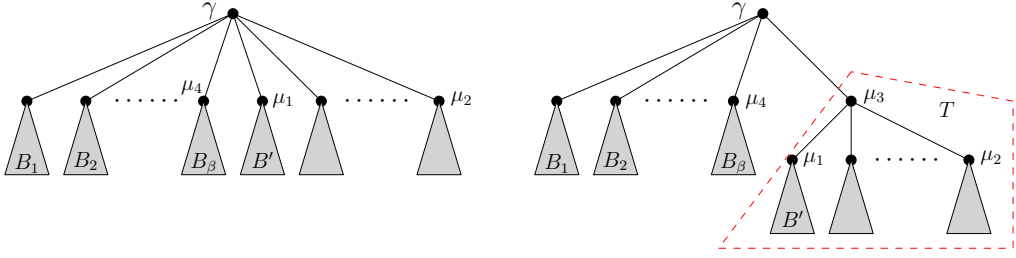}
\caption{\footnotesize
Illustrating the bundle creation operation. Left: the bundle tree before the operation. Right: the bundle tree after the operation (the subtree $T$ represents the bundle $B$).}
\label{fig:buncreate}
\end{center}
\end{minipage}
\vspace*{-0.15in}
\end{figure}

Finally we update the tree $T_f$ as follows.
Recall that the algorithm stops when either we are at the end of $\bbB$
or $prune(B')$ returns false for a bundle $B'$ in the current $\bbB$.
In the former case, we let $B=\{i\}$, and in the latter case, we let
$B$ denote the new bundle created by the bundle creation procedure.
In either case, we update as  $T_f$ as follows. Note that the
original $\bbB$ is $\{B_1,B_2,\ldots,B_g\}$ and the updated $\bbB$ is
$\{B_1,B_2,\ldots,B_{\beta},B\}$.
Essentially, the bundles
$\{B_{\beta+1},B_{\beta+2},\ldots,B_{g}\}$ have been replaced by $B$.
So we first remove the leaves corresponding to the  bundles
$\{B_{\beta+1},B_{\beta+2},\ldots,B_{g}\}$ from $T_f$.
Since they are consecutive in $T_f$, the
remove can be done in $O(\log k)$ time. Next, we insert $B$ into $T_f$
as the rightmost leave. In the former case (i.e., $B=\{i\}$),
$f_{\min}(B)=f_{\max}(B)=f(i)$. In the latter case, $f_{\min}(B)=f(i)$ and $f_{\max}(B)=f_{\max}(B_g)$,
which can be obtained in $O(\log k)$ time from the original
$T_f$. Hence, in either case the total time for updating $T_f$ is $O(\log
k)$. In addition, we set the cross pointer of the new leaf to the node
$\mu''$ of $T_{\bbB}$ whose subtree represents $B$, which is done in constant
time since we have the access of $\mu''$ after $T_{\bbB}$ is updated
(e.g., $\mu''$ is $\mu_3$ in the case of Fig.~\ref{fig:buncreate}).




\subsection{Computing the Closest Point $q^*$}
\label{sec:computeq}

Recall that we have assumed that $q^*$ is on $w_i^l$ for some $i\in [1,k]$, i.e., $q^*=q^*_l$.
According to our pruning algorithm for computing the bundle sequence
$\bbB$, the point $q^*$ 
must be on $w^l_{f(j)}$ for some index $j\in \bbB$. In this section,
we will compute $q^*$ by using the bundle sequence $\bbB$.
For example, in Fig~\ref{fig:afterprune}, our goal is to compute $q^*$ on the left sides of those (red) thick segments.

Recall that we have defined in Section~\ref{sec:decom} that
$R_i$ is the region of $\calP$ bounded by $\pi(s,v_i)$, $\pi(s,v_{i+1})$, and
$\alpha_i$, where $\alpha_i$ is either a bisector super-curve whose
endpoints are $v_i$ and $v_{i+1}$ or a chain of obstacle edges. Also
recall that $R_i$ consists of a tail and a cell.



Let $\tau$ be any segment in $\calP$ such that $R_i$ contains
$\pi(s,\tau)$.
With the help
of the decomposition $\calD$ proposed in Section~\ref{sec:segment},
we propose a {\em region-processing} algorithm to compute
$\pi(s,\tau)$ in the following lemma.

\begin{lemma}\label{lem:regionpro}
Suppose $\tau$ is a segment of $\calP$ such that $R_i$ contains
$\pi(s,\tau)$ and $R_i$ is known. Then $\pi(s,\tau)$ can be computed in
$O(\log h\log n)$ time, after $O(n\log h)$ time and space preprocessing.
\end{lemma}
\begin{proof}
We first present our region-processing algorithm for computing
$\pi(s,\tau)$, and then argue its correctness. Finally, we will analyze the
running time of the algorithm.

\paragraph{The algorithm.}
For each of $\pi(s,v_i)$, $\pi(s,v_{i+1})$, and $\alpha_i$,
we check whether it crosses $\tau$. Note that this step is not
necessary for $\alpha_i$ if $\alpha_i$ is a chain of obstacle
edges since $\tau$ cannot cross any obstacle edge.
By Observation~\ref{obser:10}(2), $\tau$ intersect  $\pi(s,v_i)$
(resp., $\pi(s,v_{i+1})$) at most once.

To avoid the tedious case analysis, by Observation~\ref{lem:10}(2), we
assume that if $\tau$ intersects $\pi(s,v_i)$ or $\pi(s,v_{i+1})$,
then the intersection is a single point (i.e., not a general sub-segment of
$\tau$).
Let $a$ (resp., $b$) be the intersection between $\tau$ and
$\pi(s,v_i)$ (resp., $\pi(s,v_{i+1})$); if there is no intersection,
we simply let $a$  (resp., $b$) refer to $\emptyset$.
In general, if
$\alpha_i$ is a bisector super-curve, $\tau$ may intersect $\alpha_i$
multiple times, and we let $c$ be an arbitrary such intersection;
similarly, if there is no intersection let $c$ refer to $\emptyset$.

If $a=b$ and $a\neq\emptyset$, then $a$ is a point on the tail of $R_i$. By
Observation~\ref{lem:10}(2), $\tau$ can only intersect the tail once.
By the definition of $R_i$, for any point $t$ in the cell of $R_i$,
$d(s,a)\leq d(s,t)$. This implies that $\pi(s,a)$ is $\pi(s,\tau)$. So we can finish the
algorithm in this case.

Otherwise (i.e., $a\neq b$ or $a=b=\emptyset$),
if at least one element of $\{a,b,c\}$ is not $\emptyset$,
then for each point $p$ of $\{a,b,c\}$ and $p\neq\emptyset$, we do the following.
Observe that $p$ is not on the tail of $R_i$.
By the definition of the decomposition $\calD$, regardless of whether $p$ is
on $\pi(s,v_i)$, $\pi(s,v_{i+1})$, or $\alpha_i$,
there is a cell $\Delta_p$ of $\calD$ such that $\Delta_p$
contains $p$ and $\Delta_p$ is in $R_i$.
By Lemma~\ref{lem:10}(4), $\Delta_p\cap \tau$ consists of at most two
maximal sub-segments $\tau_1$ and $\tau_2$. Since $\Delta_p$ is a simple polygon, we can
build a ray-shooting data structure on each of the inside and the
outside of
$\Delta_p$. Then, we can compute $\tau_1$ and $\tau_2$ in $O(\log n)$
time by using ray-shooting queries.
Next, we compute  $\pi(s,\tau_1)$ and $\pi(s,\tau_2)$ in $O(\log n)$
time by Lemma~\ref{lem:10}(5).
In this way, we obtain at most six candidate paths (for the at most three non-empty points of $\{a,b,c\}$) and return
the shortest one as $\pi(s,\tau)$.

The remaining case is when every element of $\{a,b,c\}$ is
$\emptyset$, i.e., $\tau$ does not cross any of
the three parts of $\partial R_i$. In this case,
$\tau$ is contained in a single cell $\Delta$ of $\calD$.
We can determine $\Delta$ by locating the cell of $\calD$ that contains an arbitrary endpoint of
$\tau$. Then, we compute $\pi(s,\tau)$ by Lemma~\ref{lem:10}(5).

\paragraph{The correctness.}
Recall that $R_i$ contains $\pi(s,\tau)$. Let $t$
a closest point of $\tau$ (i.e., $\pi(s,\tau)=\pi(s,t)$). Thus, $R_i$ contains $t$.
If $t$ is on the tail of $R_i$, then our algorithm correctly computes
$\pi(s,\tau)$ as discussed above.
Otherwise, if $\tau$ is in $R_i$, then $\tau$ must be in a single cell of $\calD$.
Clearly, our algorithm correctly computes $\pi(s,\tau)$ in this case.
If $\tau$ is not in $R_i$, then since $R_i$ contains $t$,
$\tau$ must cross the boundary of $R_i$.
Suppose we move from $t$ along $\tau$ until we cross the boundary of
$R_i$ at a point $p$. Let $\Delta_p$ be the cell of $\calD$ that is in
$R_i$ and contains $p$. Be definition, $\Delta_p$ also contains $t$. If $p$
is on $\pi(s,v_i)$ (resp., $\pi(s,v_{i+1})$), then since $\tau$ intersects
$\pi(s,v_i)$ (resp., $\pi(s,v_{i+1})$) at a single point, our
algorithm correctly computes $\pi(s,\tau)$.
If $p$ is on $\alpha_i$, then all intersections between $\tau$
and $\alpha_i$ are in $\Delta_p$ since $\alpha_i$ is contained in
$\Delta_p$.
Hence, our algorithm also correctly computes $\pi(s,\tau)$.

\paragraph{The time analysis.}
The algorithm needs at most six calls of
Lemma~\ref{lem:10}(5), which take $O(\log n)$ time.
It also has at most two SP-segment-intersection queries for computing
the intersections of $\tau$ with $\pi(s,v_i)$ and $\pi(s,v_{i+1})$.
Again, we will show that each such query can be answered in
$O(\log h\log n)$ time with $O(n\log h)$ time and space preprocessing.

In addition, if $\alpha_i$ is a bisector super-curve,
our algorithm also needs to compute an intersection between $\tau$ and $\alpha_i$.
This can be done in $O(\log n)$ time after linear time preprocessing
on $\alpha_i$ using the
ray-shooting data structure on curved simple
polygons or splinegons~\cite{ref:MelissaratosSh92} (indeed, each
bisector edge of $\alpha_i$ is convex, and thus it is straightforward
to make $\alpha_i$ a splinegon~\cite{ref:MelissaratosSh92}, e.g., by the standard technique as detailed in the proof of Lemma \ref{lem:200}).
Thus, the total preprocessing time on all such curves
$\alpha_i$ for $i=1,2,\ldots,h^*$ is $O(n)$.

Also, we have mentioned before that we need a constant number of
ray-shooting queries on the cells $\Delta_p$ to determine the at most
two sub-segments of $\Delta_p\cap \tau$. The query time is $O(\log n)$
and the total preprocessing time on all cells of $\calD$ is $O(n)$.

Hence, our region-processing algorithm runs in $O(\log h\log n)$ time,
and the total preprocessing time and space is $O(n\log h)$.
\qed
\end{proof}


Recall that $\calR=\{R_1,R_2,\ldots,R_{h^*}\}$.
Due to our general position assumption that $q$ is not collinear with
any two obstacle vertices, none of $\{q,q_1,\ldots,q_{k}\}$ is an obstacle vertex. Then,
for each $k'\in [0,k]$, there is a unique region $R_i$ of $\calR$ whose cell
contains $q_{f(k')}$, such that the shortest path $\pi_{f(k')}$ is
contained in $R_{i}$, and we let $z(k')$ refer to the index $i$ of $R_i$.
Computing $z(0),z(1),\ldots,z(k)$ can be done in $O(k\log n)$ time by point
location queries on the cells of the regions of $\calR$.

For any two indices $k_1$ and $k_2$ in $[1,h^*]$, if $k_1\leq k_2$,
then let $[k_1,k_2]_R$ denote the set of all integers $k'\in
[k_1,k_2]$; otherwise, let $[k_1,k_2]_R$ denote the set of all
integers $k'\in [k_1,h^*]\cup [1,k_2]$. Recall that the regions
$R_1,R_2,\ldots,R_{h^*}$ are counterclockwise around $s$.
We actually use $[k_1,k_2]_R$ to refer to the set of indices of the
regions of $\calR$ from $R_{k_1}$ to $R_{k_2}$ counterclockwise around
$s$.

Next we compute $q^*$ on $w^l_{f(j)}$ for $j\in
\bbB$, by using our region-processing algorithm in
Lemma~\ref{lem:regionpro}.
Consider the bundles of $\bbB=\{B_1,B_2,\ldots, B_g\}$. For each
$b$ with $1\leq b\leq g$, we call a procedure $path(B_b,z(i))$, where
$i$ is the last index of $B_{b-1}$ if $b\geq 2$ and $i=0$ otherwise.
Note that given the access of $B_b$, we can obtain $i$ in constant
time by using our data structure in Section~\ref{sec:bundle}. Also
note that $i<j$ for any index $j\in B_b$.
The procedure $path(B_b,z(i))$ works as follows.

Depending on whether $B_b$ is atomic or composite, there are two cases.

\paragraph{The atomic case.}
If $B_b$ is atomic, let $j$ be the only index of $B_b$. According to the
bundle-properties, $i<j$ and $f(i)<f(j)$. So $\pi_{f(j)}$ and
$\pi_{f(i)}$ are consistent. By Lemma~\ref{lem:key}(1),
$D_i$ is contained in $D_j$. Let $D$ be $D_j$ minus the interior of $D_i$.
We have the following observation.

\begin{figure}[t]
\begin{minipage}[t]{\linewidth}
\begin{center}
\includegraphics[totalheight=0.9in]{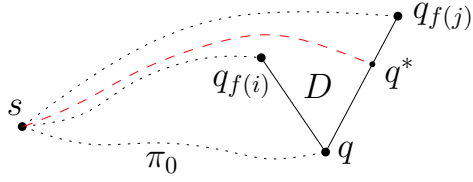}
\caption{\footnotesize
Illustrating Observation~\ref{obser:100}.}
\label{fig:atomicpath}
\end{center}
\end{minipage}
\vspace*{-0.15in}
\end{figure}

\begin{observation}\label{obser:100}
If $q^*$ is on $w^l_{f(j)}$, then $\pi(s,q^*)$ must be in
$D$ (e.g., see Fig.~\ref{fig:atomicpath}).
\end{observation}
\begin{proof}
Suppose $q^*$ is on $w^l_{f(j)}$. Let $t=q^*$. By definition, the point $t^+$ is in the
interior of $D$. Since $t=q^*$, $\pi(s,t^+)$ does not intersect any point of $w_{f(i)}$
or $w_{f(j)}$ and it does not contain $q$ either.
Also, $\pi(s,t^+)$ cannot cross either $\pi_{f(i)}$ or $\pi_{f(j)}$.
Hence, $\pi(s,t)$ must be in $D$.
\qed
\end{proof}


Observation~\ref{obser:100} leads to the following lemma.

\begin{lemma}\label{lem:170}
If $q^*$ is on $w^l_{f(j)}$, then $\pi(s,q^*)$ is in
$R_{k'}$ for some index $k'\in [z(i),z(j)]_R$, and further,
any shortest path $\pi(s,w_{f(j)})$ from $s$ to $w_{f(j)}$ is $\pi(s,q^*)$.
\end{lemma}
\begin{proof}
Suppose $q^*$ is on $w^l_{f(j)}$.
Since $q^*$ is also a closest point of $w_{f(j)}$, $\pi(s,w_{f(j)})$ must be $\pi(s,q^*)$.

Note that $\pi(s,q^*)$ must be contained in a region of $\calR$.
By Observation~\ref{obser:100}, $\pi(s,q^*)$ is in $D$. Hence,
$\pi_{f(j)}$ is counterclockwise from $\pi(s,q^*)$ with respect to
$\pi_{f(i)}$ around $s$. Since $\pi_{f(j)}$ is in $R_{z(j)}$,
and $\pi_{f(i)}$ is in $R_{z(i)}$, there is a region $R_{k'}\in \calR$
that contains $\pi(s,q^*)$ such that $R_{z(j)}$ is counterclockwise from $R_{k'}$ with respect
to $R_{z(i)}$ around $s$, which implies that $k'\in[z(i),z(j)]_{R}$.
\qed
\end{proof}

For each $k'\in [z(i),z(j)]_R$, we apply our region-processing
algorithm on $R_{k'}$ and $w_{f(j)}$ to obtain a path, and we keep the shortest
path $\pi$ among all such paths;
let $q_{f(j)}^l$ be the endpoint of $\pi$ on $w_{f(j)}$.
According to Lemma~\ref{lem:170}, if $q^*$ is on $w^l_{f(j)}$, then
$q^*$ must be $q_{f(j)}^l$.

For the purpose of analyzing the total running time of our algorithm, as will be seen later, for each $k'\in [z(i),z(j)]_R$ with $k'\neq z(i)$ and $k'\neq z(j)$,
the region-processing algorithm will not be called on $R_{k'}$
again in the entire algorithm for computing $q^*_l$.
On the other hand, we charge the two algorithm calls
on $R_{k'}$ for $k'=z(i)$ and $k'= z(j)$ to the index $j$ of $\bbB$. In this way, the total number of calls to the region-processing
procedure in the entire algorithm is $O(h^*+k)$ since the total number of indices of
$\bbB$ is at most $k$ and the total number of regions $R_{k'}$ is $h^*$.

\paragraph{The composite case.}
If $B_b$ is composite, the algorithm is more complicated. Let $j$ be
the wrap index of $B_b$. Observation~\ref{obser:100} and
Lemma~\ref{lem:170} still hold on $j$. However, since now the region $D$ also
contains a
portion of $w_{f(j')}$ for each $j'\in B_b$ and $j'\neq j$ (e.g., see
Fig.~\ref{fig:compath}), $D$ may also
contain the shortest path from $s$ to $w_{f(j')}$. In order to
avoid calling the region-processing procedure on the same region of
$\calR$ too many times, we use the following approach to process
$w_{f(j)}$.

\begin{figure}[t]
\begin{minipage}[t]{\linewidth}
\begin{center}
\includegraphics[totalheight=1.0in]{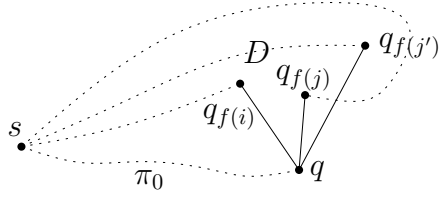}
\caption{\footnotesize
$j$ is the wrap index of $B_b$ and $j'$ is another index of $B_b$ with $j'\neq j$; $\pi_{f(j')}$ is in the region $D$.}
\label{fig:compath}
\end{center}
\end{minipage}
\vspace*{-0.15in}
\end{figure}

For any two different
indices of $k'$ and $k''$ in a range $[k_1,k_2]_R$ of indices of the regions of
$\calR$, we say that $k''$ is {\em ccw-larger} than $k'$ if $[k',k'']_R$ is a
subset of $[k_1,k_2]_R$ (e.g., if $k_1<k_2$, then $k'<k''$).

Define $z_{ij}$ to be the ccw-largest index in $[z(i),z(j)]$
such that $w_{f(j)}$ crosses $\partial R_{z_{ij}}$ (if no such index
exists, then let $z_{ij}=z(i)$). We first compute $z_{ij}$ (to be discussed later). Then, we call
the region-processing procedure on $R_{k'}$ for all $k'\in
[z(i),z_{ij}]$ and return the shortest path $\pi$ that is found; let
$q_{f(j)}^l$ be
the endpoint of $\pi$ on $w_{f(j)}$. By the following lemma, if
$q^*$ is on $w_{f(j)}^l$, then $q_{f(j)}^l$ is $q^*$.

\begin{lemma}
If $q^*$ is on $w^l_{f(j)}$, then $\pi(s,q^*)$ is in
$R_{k'}$ for some index $k'\in [z(i),z_{ij}]_R$, and further, any
shortest path $\pi(s,w_{f(j)})$ from $s$ to $w_{f(j)}$ is
$\pi(s,q^*)$.
\end{lemma}
\begin{proof}
By Lemma~\ref{lem:170}, the lemma statement holds for some $k'\in
[z(i),z(j)]_R$. In the following we show that $k'$ is in
$[z(i),z_{ij}]_R$.

Assume to the contrary that $k'$ is not in $[z(i),z_{ij}]_R$.
Then, $k'$ is ccw-larger than $z_{ij}$ and $w_{f(j)}$ does not
cross $\partial R_{k'}$. This implies that $w_{f(j)}$ and $q$ are in $R_{k'}$.
Since $i<j$, $\pi_{f(j)}$ is counterclockwise from $\pi(f(i))$ with
respect to $\pi_0=\pi(s,q)$. This implies that $z(i)\in [z(0),z(j)]_R$.
But $w_{f(j)}\in R_{k'}$ implies that $z(0)=z(j)=k'$. Thus, we have
$z(i)=z(j)=k'$.
Since $z_{ij}\in [z(i),z(j)]_R$, we obtain $z_{ij}=k'$. But this
contradicts with that $k'$ is not in $[z(i),z_{ij}]_R$.

The lemma thus follows. \qed
\end{proof}

The following lemma makes sure that when we process $w_{f(j')}$ for
any other index $j'$ of $B_b$ with $j'\neq j$, we do not need to
consider the regions $R_{k'}$ for $k'\in [z(i),z_{ij}-1]$ if $z_{ij}\neq z(i)$.

\begin{lemma}\label{lem:190}
Suppose $z_{ij}\neq z(i)$. If $q^*$ is on $w^l_{f(j')}$ for some $j'\in
B_b$ and $j'\neq j$, then $\pi(s,q^*)$ is in $R_{k'}$ for some
$k'\in [z_{ij},z(j')]_R$.
\end{lemma}
\begin{proof}
Consider any such $j'$ as in the lemma statement.
Since $j$ is the wrap index of $B_b$, $\pi_{f(j)}$ crosses
$w_{f(j')}$ at a point $p$ (e.g., see Fig.~\ref{fig:pruneregion}).
By Lemma~\ref{lem:key}(2), the
portion $\overline{qp}$ of $w_{f(j')}$ can be pruned, i.e., $q^*$
cannot be on $\overline{qp}$. Let $D^1$ be the region bounded by
$\overline{qp}$, $w_{f(j)}$, and the subpath $\pi(p,q_{f(j)})$ of $\pi_{f(j)}$
between $p$ and $q_{f(j)}$. Note that $D^1\subseteq D_{f(j')}$.

Since $q^*\in w^l_{f(j')}$, $\pi(s,q^*)$ must be in $D_{f(j')}$. We claim
that $\pi(s,q^*)$ is in $D^2=D_{f(j')}\setminus D^1$ (e.g., see
Fig.~\ref{fig:pruneregion}). To see this, $D^2$ is one
of the two sub-regions of $D_{f(j')}$ partitioned by $w_{f(j)}\cup
\pi(p,q_{f(j)})$.
Since $q^*$ is not on $\overline{qp}$, $q^*$ must be in the interior
of $\overline{pq_{f(j')}}$, which is in $D^2$. Hence, to prove that
$\pi(s,q^*)$  is in $D^2$, it is sufficient to show that $\pi(s,q^*)$
does not cross either $w_{f(j)}$ or $\pi(p,q_{f(j)})$. Indeed, $\pi(s,q^*)$ does
not cross $\pi(p,q_{f(j)})$. On the other hand, $\pi(s,q^*)$ does not intersect
$w_{f(j)}$ since otherwise $q^*$ would not be a closest point of
$\vis(q)$.  This shows that $\pi(s,q^*)$ is in $D^2$.

\begin{figure}[t]
\begin{minipage}[t]{\linewidth}
\begin{center}
\includegraphics[totalheight=1.3in]{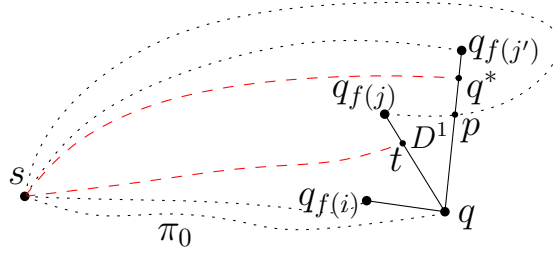}
\caption{\footnotesize
$j$ is the wrap index of $B_b$ and $j'$ is another index of $B_b$ with $j'\neq j$; $\pi_{f(j')}$ is in the region $D$.}
\label{fig:pruneregion}
\end{center}
\end{minipage}
\vspace*{-0.15in}
\end{figure}

Since $z(i)\neq z_{ij}$, $z_{ij}$ is ccw-larger than $z(i)$. By the
definition of $z_{ij}$, $w_{f(j)}$ crosses $\partial R_{z_{ij}}$, say,
at a point $t$ (e.g., see Fig.~\ref{fig:pruneregion}). Hence, the region $R_{z_{ij}}$ contains a shortest path $\pi(s,t)$ from $s$ to $t$.
Further, since $z_{ij}\in [z(i),z(j')]_R$, $\pi(s,t)$ is also in $D^2$.
Since both $s$ and $t$ are on the boundary of $D^2$, $\pi(s,t)$
partitions $D^2$ into two sub-regions and one of them, denoted by
$D^3$, contains $q^*$. Since $\pi(s,q^*)$ does not cross $\pi(s,t)$,
$\pi(s,q^*)$ is in $D^3$, which implies that $\pi(s,q^*)$ must be in
some region $R_{k'}$ with $k'\in [z_{ij},z(j')]_R$.

This proves the lemma.
\qed
\end{proof}

In order to compute the index $z_{ij}$, we will use a {\em
$\calR$-region range} query. Namely, given the index range
$[z(i),z(j)]_R$ as well as $w_{f(j)}$, the query can be used to compute
$z_{ij}$.
In Section~\ref{sec:imple} we will give a data structure that can answer each such
query in $O(\log h\log n)$ time (after $O(n\log h)$ time and space preprocessing).

After $w_{f(j)}$ is processed as above, $q_{f(j)}^l$ is computed.
By Lemma~\ref{lem:190}, to
process $w_{f(j')}$ for other indices $j'$ of $B_b\setminus\{j\}$, we only need to
consider the indices of the regions of $\calR$
after $z_{ij}$. Let $B'_1,B_2',\ldots,B'_{g'-1}$ be
the bundles in $B_b$ other than the last one. For each $1\leq b'\leq
g'-1$, if $b'=1$, we call $path(B'_{b'},z_{ij})$ recursively;
otherwise, we call $path(B'_{b'},z(i'))$ recursively, where $i'$ is the last
index of $B'_{b'-1}$.

\paragraph{Remark.} For the procedure $path(B'_{1},z_{ij})$, the above algorithm still works by replacing $z(i)$ by $z_{ij}$. To argue the correctness, the region $D$ in Observation~\ref{obser:100} and Lemma~\ref{lem:170} should be defined
to be the region $D^3$ in the proof of Lemma~\ref{lem:190} (with respect to $j'$); then all observations above (after replacing $z(i)$ by $z_{ij}$) still hold for
$path(B'_{1},z_{ij})$.
\paragraph{}

After $w_{f(j)}$ is processed for each $j\in \bbB$, $q_{f(j)}^l$ is
computed for every $j\in \bbB$; among these at most $k$ points,
we return the point $q'$ whose value $d(s,q')$ is the
smallest as $q_l^*$, which is $q^*$ based on our above analysis (and
also due to our assumption that $q^*$ is on $w_i^l$ for some $i\in [1,k]$).
The total number of calls on the
region-processing procedures is $O(k+h^*)$. The total number of
$\calR$-region range queries is $O(k)$ since each
such query is for a composite bundle and there
are at most $k$ bundles in total. Hence, the total time of the algorithm is
$O((h+k)\log h\log n)$. Recall that $k\leq K$.

\subsection{The Algorithm Implementation}
\label{sec:imple}

In this section, we discuss some implementation details left out
above. Specifically, we will give our algorithm for computing the map $f(\cdot)$,
and give our data structures for answering the
SP-segment-intersections queries and the $\calR$-region range
queries.

\subsubsection{Computing the Map $f(\cdot)$}
Recall the definitions of $Q$, $\calC_Q$, and
$\calL_Q$ in Section~\ref{sec:obser}. Computing the map $f(\cdot)$ is to compute the list
$\calL_Q=\{q,q_{f(1)},\ldots,q_{f(k)}\}$. Intuitively, we want to
order the paths $\pi_1,\ldots,\pi_k$ counterclockwise around $s$ with
respect to $\pi_0$. Our goal is to prove Lemma~\ref{lem:map}.

We begin with our preprocessing algorithm.
Let $\Sigma(s)$ denote the decomposition of $\spm(s)$ by the edges of
$\spt(s)$, which can be constructed in $O(n)$ time after $\spm(s)$ is
given. For each cell $\sigma$ of $\Sigma(s)$, we
pick an arbitrary point in the interior of $\sigma$ as the {\em
representative point} of $\sigma$.  Let $X$ denote the set of all such
representative points. Let $T_X$ be the tree that is the union of the shortest paths
from $s$ to all points of $X$, and let $s$ be the root of $T_X$.
Clearly, $T_X$ has $O(n)$ nodes and can be computed in $O(n)$ time once we have
$\Sigma(s)$. The points of $X$ are exactly the leaves of $T_X$.
We find a base leave $p^*$ of $T_X$ in $O(n)$ time.
Then, we compute in $O(n)$ time the list $\calL_l(T_X,p^*)$ of all
leaves and the cycle $\calL_l(T_X)$.
To simplify the notation, let $\calL_X=\calL_l(T_X,p^*)$ and let
$\calC_X=\calL_l(T_X)$.
This finishes our preprocessing, which takes $O(n)$ time.

In the sequel, we discuss our algorithm for computing the list
$\calL_Q$ in $O(k\log n)$ time. It is sufficient to compute the
circular list $\calC_Q$ since  we can obtain $\calL_Q$ from $\calC_Q$
in $O(k)$ time by breaking the cycle at $q$.

Let $q_0=q$ (temporarily only for the discussion in this subsection).
Recall that for each point $q_i\in Q$ with $0\leq i\leq k$,
$u_i$ is the root of the cell of $\spm(s)$ that contains $q_i$ and
determines the shortest path $\pi_i$, and note that $\overline{q_iu_i}$
is in a cell of $\Sigma(s)$, denoted by $\sigma_i$ (which
can be determined in $O(\log n)$ time by a point location in $\Sigma(s)$).
If all cells
$\sigma_0,\sigma_1,\ldots,\sigma_k$ are distinct, then the order of the
points of $Q$ following the relative order of the representative points of the cells
$\sigma_0,\sigma_1,\ldots,\sigma_k$ in $\calC_X$ is exactly $\calC_Q$,
which can be computed in $O(k\log n)$ time with help of the circular list
$\calC_X$.

If $\sigma_0,\sigma_1,\ldots,\sigma_k$ are not distinct, then we
first compute the circular list of the cells by the above
algorithm. To simplify the notation, let
$\sigma_0,\sigma_1,\ldots,\sigma_k$ be the circular list.
Then, two cells are the same only if they are adjacent in the
list. Hence, we can determine in $O(k)$ time the cycle of unique
cells $\sigma_0',\sigma_1',\ldots,\sigma_{k'}'$ for $k'<k$, and
further, for each cell $\sigma_i'$, the set $Q(\sigma_i')$ of points
of $Q$ in $\sigma_i'$ can also be determined. Consider a cell $\sigma_i'$
and let $u_i'$ be the root.
Let $T(\sigma_i')$ be the union of the segments $\overline{u_i'q'}$
for all $q'\in Q(\sigma_i')$, and we consider $T(\sigma_i')$ as a tree
rooted at $u_i'$. Since $u_i'$ is an obstacle vertex, $u_i'$ is a node
in $T_X$. If $u_i'$ is not $s$, then let $p$ be the parent of
$u_i'$ in $T_X$; otherwise let $p$ be the child of $s$ in
$T_X$ that is an ancestor of the base leave $p^*$ (we
compute that particular child of $s$ in the
preprocessing). 
Starting from the counterclockwise first child of $u_i'$ in
$T(\sigma_i')$ with respect
to $\overline{u_i'p}$, and let $\calL(\sigma_i')$ be the
list of the children of $u_i'$ in $T(\sigma_i')$ ordered
counterclockwise. It can be verified that the concatenation of
$\calL(\sigma_0'),\calL(\sigma_1'),\ldots,\calL(\sigma'_{k'})$ is
exactly the circular list $\calC_Q$. Following the above description, the circular list  $\calC_Q$ can be
computed in $O(k\log n)$ time.

This proves Lemma~\ref{lem:map}.

\subsubsection{The SP-segment-intersection Queries}
\label{sec:SPsegment}

In this section, we present our data structure for answering the SP-segment-intersection queries. Specifically, given any $i,j\in [1,k]$, we
want to determine whether $w_{f(i)}$ crosses $\pi_{f(j)}$, and if yes,
compute an intersection.  Here we consider a more general problem.
Given a point $t$ and a segment $\tau$ in $\calP$, we want to
compute an intersection between $\tau$ and the shortest path
$\pi(s,t)$ (or report none if they do not intersect). In the case
where $t$ has multiple shortest paths (and thus $\pi(s,t)$ is not unique),
the root $r$ of a cell of $\spm(s)$ should also be provided so that
$\pi(s,t)$ refers to the one that contains $\overline{rt}$. But to
simplify the discussion, we assume $t$ always has a unique shortest
path (the other case can be solved by our algorithm too). 

We will show that with $O(n\log h)$ time and space
preprocessing (with a given $\spm(s)$), each such query can be
answered in $O(\log h\log n)$ time. When $h=O(1)$, the
result is optimal.



Recall the definitions of $V$, $\Pi$, $T_{V}$, and the list
$\calL_l(T_V,v_1)=\{v_1,v_2,\cdots,v_{h^*}\}$ in
Section~\ref{sec:segment}.
In the following, we build up our data structure incrementally: We will first show how to answer queries when $t$ is in
$V$, then show how to answer queries when $t$ a vertex of $T_{V}$,
and finally discuss the general case where $t$ can be any point in
$\calP$.

We build a complete binary search tree $T_1$ as follows.
The leaves of $T_1$ from left to right correspond to the points
$v_1,v_2,\ldots,v_{h^*}$ of $V$ in this order. In the following we will consider the points of $V$ and the leaves of $T_1$ interchangeably. Note that each point of $V$ is also a leaf in the tree $T_V$.
Consider any node $u$ of $T_1$. We maintain a path
$P(u)$ of edges of $T_V$, defined as follows.
Let $T_1(u)$ be the subtree of $T_1$ rooted at $u$ and let $S(u)$ be
the set of the leaves of $T_1(u)$.  If $u$ is the root,
then $P(u)$ is the common sub-path (i.e., the intersection) of the
shortest paths $\pi(s,p)$ for all $p\in S(u)$ (note that $\pi(s,p)$ is also the path of $T_V$ from $p$ to the root $s$). Otherwise, $P(u)$ is the
portion of the common sub-path of $\pi(s,p)$ for all $p\in S(u)$
that is not stored in $P(u')$ for any ancestor $u'$ of $u$.
In this way, for each leave $v_i$, the edges of $P(u)$ of all nodes $u$ in the path
of $T_1$ from $v_i$ to the root are pairwise disjoint and
comprise exactly $\pi(s,v_i)$.
Further, for each node $u$ of $T_1$, since $P(u)$ is a path of edges,
we build a ray-shooting data structure on $P(u)$ by standard
techniques as detailed in the following lemma.

\begin{lemma}\label{lem:200}
For the path $P(u)$ of each node $u$ of $T_1$ with $m=|P(u)|$,
we can build a data structure of $O(m)$
size in $O(m)$ time such that given any ray $\rho$ in the plane, we
can compute in $O(\log m)$ time the first intersection (if any) between $\rho$ and $P(u)$.
\end{lemma}
\begin{proof}
This can be easily done by using the ray-shooting data structure for
simple polygons \cite{ref:ChazelleRa94,ref:HershbergerA95}. We provide
the details below.

Let $R$ be a big rectangle in the plane that contains all edges of
$P(u)$. Let $p$ be the topmost point of $P(u)$. We shoot a ray from
$p$ upwards until it hits $\partial R$ at a point $p'$. Then, we can
consider $P(u)$, $\overline{pp'}$, and $R$ bounds a simple polygon
$P$. We build a ray-shooting data structure in $P$ in $O(m)$ size and
space~\cite{ref:ChazelleRa94,ref:HershbergerA95}.

Consider any ray-shooting query for $P(u)$. Given a ray $\rho$, we
compute the first point $a$ of $\partial P$ hit by $\rho$ in $O(\log m)$
time by using the ray-shooting data structure on $P$. If $a$ is on
$P(u)$, then we are done and return $a$ as the answer. If $a$ is on
$\partial R$, then we are also done and report that there is no intersection
between $\rho$ and $P(u)$. If $a$ is on $\overline{pp'}$, then we keep
shooting the ray after $a$ and using the ray-shooting data structure
again to compute the next point $a'\in \partial P$ hit by the ray. Similarly as
above, if $a'$ is on $P(u)$, then we are done and return $a'$. If $a'$
is on $\partial R$, then we report that there is no intersection. Note that
$a'$ cannot be on $\overline{pp'}$. Hence, we can answer the
ray-shooting query on $P(u)$ in $O(\log m)$ time by making at most two
ray-shooting queries on $P$.
\qed
\end{proof}

We call the information associated with each node $u$ of $T_1$ the
{\em auxiliary data structure} at $u$.

\begin{lemma}\label{lem:T1}
The size of $T_1$ is $O(n\log h)$ and $T_1$ can be built in $O(n\log h)$
time.
\end{lemma}
\begin{proof}
Recall that the number of edges of $T_V$ is $O(n)$. In the following,
we first show that each edge $e$ of $T_V$ is stored in $P(u)$ of at most two nodes $u$ in each level of $T_1$.

Assume to the contrary that there are three such nodes $u$ in
the same level of $T_1$ that all store the same edge $e$ of $T_V$ in $P(u)$.
Let the three nodes be $u_1,u_2,u_3$ from left to right.
If $u_1,u_2,u_3$ are consecutive, then two of them, say,
$u_1$ and $u_2$, must share the same parent $u$.
Since $e$ is in both $P(u_1)$ and $P(u_2)$, by definition, $e$ should
be in $P(u')$ for an ancestor $u'$ of $u$ (including $u$ itself).
Thus, $e$ should not be in either $P(u_1)$ or $P(u_2)$,
incurring contradiction.

In the following we assume $u_1,u_2,u_3$ are not consecutive. If two
of them share the same parent, then we can apply the same argument as
above. 
Otherwise, we show below that the
sibling $u'$ of $u_2$ (i.e., $u$ and $u'$ share
the same parent) has $P(u')$ including $e$. Consequently, the above proof applies.


Let $V_e$ be the set of points of $V$ whose paths from $s$ in $T_V$
contain the edge $e$. Note that $V_e$ consists of exactly the leaves in the subtree of $T_V$ separated by $e$. By the definition of $\calL_l(T_v,v_1)$,
the points of $V_e$ are consecutive in $\calL_l(T_v,v_1)=\{v_1,v_2,\ldots,v_{h^*}\}$. According to the definition of
$T_1$, the leaves of $T_1$ corresponding to the points of $V_e$ are
consecutive in $T_1$. Since $e$ is in both $P(u_1)$ and $P(u_3)$, all
leaves of the subtrees of $T_1(u_1)$ and $T_2(u_3)$ are in $V_e$.
Since $u_2$ is between $u_1$ and $u_3$, $u'$ is also between $u_1$ and
$u_2$. Thus, all leaves of
$T_1(u')$ must also be in $V_e$, implying that $e$ is in the common sub-path
of $\pi(s,p)$ for all $p\in S(u')$.  Since $e$ is in $P(u_2)$, $e$ is not in
$P(u'')$ for any proper ancestor $u''$ of $u_2$. Because $u'$ and $u_2$
share the same parent, we obtain that $e$ is also in $P(u')$.



This proves that each edge $e$ of $T_V$ is stored in at most two nodes
in each level of $T_1$. Since $T_1$ has $O(\log h^*)$ levels and $h^*=O(\log h)$, each edge
$e$ is stored in $O(\log h)$ nodes. Hence, the size of $T_1$
is $O(n\log h)$.

In the following, we construct the tree $T_1$ in $O(n\log h)$ time.
The key is to compute $P(u)$ for each node $u$ of $T_1$, after
which constructing the ray-shooting data structure on $P(u)$ can be
done in linear time by Lemma~\ref{lem:200}.

For each edge $e$ of $T_V$, we compute the range
$[l_e,r_e]\subseteq [1,h^*]$ that consists of all indices $i$ such that
$e$ is contained in the path from $v_i$ to $s$ in $T_V$. This can be done in $O(n)$ time as follows.
For each vertex $v$ of $T_V$, we define the range $[l_v,r_v]$ as the
set of all indices $i$ such that $v$ is contained in the path from $v_i$ to $s$ in $T_V$.
We first compute the ranges for all vertices of $T_V$. This can be
easily done a post-order traversal of $T_V$ starting from the leaf $v_1$.
Specifically, during the traversal for each vertex $v$, if $v$ is a leaf
containing $v_i\in V$, we set $l_{v}=i$ and
$r_{v}=i$; otherwise, all children of $v$ have been visited and we
set $l_v$ (resp., $r_v$) to be the smallest (resp., largest) $l_{v'}$
of all children $v'$ of $v$.
After the traversal, the ranges for all vertices of $T_V$ are computed.
Then, for each edge $e$ of $\Pi$, it is not difficult to
see that the range of $e$ is the same as that of $v$, where $v$ is the endpoint of $e$ such that the path from $s$ to $v$ in $T_V$ contains $e$.

Next we compute $P(u)$ for all nodes $u$ of $T_1$ as follows. We
consider the edges of $T_V$ following the post-order traversal from
$v_1$. For each edge $e$, by using the range $[l_e,r_e]$, we find
those nodes $u$ of $T_1$ whose $P(u)$ contains $e$. This can be done
in the similar way as the standard insertion operation in segment
trees~\cite{ref:deBergCo08}. Specifically, for each node $u$ of $T_1$, let
$[l_u,r_u]$ be the range consists of all indices $i$ such that $v_i$
is $S(u)$. Starting from the root of $T_1$, for each
node $u$, if $[l_u,r_u]\subseteq [l_e,r_e]$, then we insert $e$ to
$P(u)$; otherwise, for each child $u'$ of $u$, if $[l_e,r_e]\cap
[l_{u'},r_{u'}]\neq \emptyset$, then we proceed on $u'$ recursively.
As the standard insertion operations on segment trees, each edge $e$
is processed in $O(\log h)$ time since the height of
$T_1$ is $O(\log h)$. Hence, the total time of the algorithm is
$O(n\log h)$.
Note that since we consider the edges of $T_V$ by following the
post-order traversal from $v_1$, whenever we insert an edge $e$ to
$P(u)$, $e$ is always the edge adjacent to the first edge of the
current $P(u)$
and $e$ is then appended to $P(u)$ as the new first edge. After the
algorithm finishes, the sub-path $P(u)$ is readily available by following the
edges in the order they have been inserted and the first edge is the
one closest to $s$.

This proves the lemma.
\qed
\end{proof}

We show how to answer SP-segment-intersection queries by using
the tree $T_1$. We begin with a special case where the query point
$t$ is in $V$, say $t=v_i$ for some $i\in [1,h^*]$.
Our goal is to compute an intersection between $\tau$ and $\pi(s,v_i)$.
To answer the query, we follow the path of $T_1$ from the root to
the leaf $v_i$. For each node $u$ in the path, we use a ray-shooting
query to compute an intersection between $P(u)$ and $\tau$. If we find
an intersection, then we report
the intersection and stop the algorithm; otherwise, we proceed on the
next node. The correctness of the
algorithm is based on the fact that the union of $P(u)$ of all nodes
$u$ in the above path is exactly $\pi(s,v_i)$. The query time is $O(\log
h\log n)$ since each ray-shooting query takes $O(\log n)$ time and the height of $T_1$ is $O(\log h)$.

We then consider a more general case where the query point $t$ is a vertex $v$ of
$T_V$ ($v$ is not necessarily in $V$). To answer the query, we first pick an arbitrary
leave $v_i$ in the subtree of $T_V$ rooted at $v$ (for this, in the
preprocessing step we need to associate with
$v'$ an arbitrary leaf in its subtree for each node $v'$ of $T_V$).
Clearly, $v$ must be in the path $\pi(s,v_i)$. We follow
the path of $T_1$ from the root to the leaf $v_i$. For each node $u$ in the
path, we compute an intersection between $P(u)$ and $\tau$ by using a ray-shooting query.
If there is an intersection $p$, we check whether
$p$ is in the sub-path of $\pi(s,v_i)$ between $s$ and $v$ (see below
for more details about this). If yes, then we report $p$ and stop the algorithm.
Otherwise, since $\tau$ can only cross $\pi(s,v_i)$ once, there
cannot be any intersection between $\tau$ and $\pi(s,v)$; thus, in
this case we simply return none.
If there is no intersection between $\tau$ and $P(u)$, then
we proceed on the next node in the path. If we do not
find any intersection after we reach $v_i$, then we report none.

It remains to discuss how to determine whether $p$ is between $s$ and
$v$. The point $p$ is on an edge $e$ of $\pi(s,v_i)$, which is also in
$T_V$. Let $v'$ be the
endpoint of $e$ that is farther to $s$ in $T_V$. Observe that $p$ is between $s$ and
$v$ if and only if $v'$ is between $s$ and $v$. To determine the
latter, observe that $v'$ is between $s$ and $v$ if and only if $v'$ is
after $v$ in the canonical list $\calL(T_V,v_1)$, which can be
determined in $O(\log n)$ time (e.g., by binary search)
after $\calL(T_V,v_1)$ is computed in the preprocessing.

Hence, the total time for answering the query is $O(\log h\log n)$.

In the following, by making use of the above result,
we consider the most general case where $t$ can be any point in
$\calP$. We first present the result for the simple polygon case.

\begin{lemma}\label{lem:simpleSP}
For any simple polygon $P$ of $m$ vertices and a source point $s$ in
$P$, after $O(m)$ time
preprocessing, we can answer each SP-segment-intersection query in
$O(\log m)$ time.
\end{lemma}
\begin{proof}
Given any query segment $\tau$ and a point $t$ in $P$, the query asks
for the intersection between $\tau$ and the shortest path $\pi(s,t)$
from $s$ to $t$ in $P$ (or report none if there is no intersection).

In the preprocessing, we compute the shortest path tree $\spt(s)$ and shortest path
map $\spm(s)$ from $s$ in $P$, which can be done in $O(m)$
time~\cite{ref:GuibasLi87}. We then build a point location data
structure on $\spm(s)$ in $O(n)$
time~\cite{ref:EdelsbrunnerOp86,ref:KirkpatrickOp83}. Further, we
compute
the canonical cycle $\calC(\spt(s))$ in $O(m)$ time.

Let $r_t$ be the root of the cell of $\spm(s)$ containing $t$ such
that $\pi(s,t)$ contains $\overline{r_tt}$.
We first check whether $\overline{r_tt}$ intersects $\tau$. If yes, we
return the intersection. Otherwise, we proceed to compute the
intersection between $\tau$ and
the shortest path $\pi(s,r_t)$ from $s$ to $r_t$.

Let $a$ and $b$ be the two endpoints of $\tau$, respectively.
We first check whether $a$ is on $\pi(s,r_t)$, as follows.
If $a\in \pi(s,r_t)$, then $a$ must be on an edge $e$ of
$\pi(s,r_t)\subseteq \spt(s)$, and further, $r_t$ must be a descendent
of $v_e$, where
$v_e$ is the endpoint of $e$ farther to $s$ in $\pi(s,r_t)$. Therefore, to check whether
$a$ is on $\pi(s,r_t)$, we can use the following approach. First, we
determine whether $a$ is on an edge of $\spt(s)$, which can be done in
$O(\log m)$ time by  a point location query on the decomposition of
$\spm(s)$ by the edges of $\spt(s)$. If $a$ is not on an edge of $\spt(s)$, then we
know that $a$ cannot be in $\pi(s,r_t)$. Otherwise, we proceed on
determining whether $r_t$ is a descendent of $v_e$.
To this end, observe that  $r_t$ is a descendent of $v_e$ if
and only if the lowest common ancestor of $v_e$ and $r_t$ in $\spt(s)$ is
$v_e$, which can be computed in $O(1)$ time after $O(m)$ time
preprocessing on $\spt(s)$ \cite{ref:BenderTh00,ref:HarelFa84}.

Hence, we can check whether $a$ is in $\pi(s,r_t)$ in $O(\log m)$
time. Similarly we can check whether $b$ is in $\pi(s,r_t)$ in $O(\log
m)$ time. If either $a$ or $b$ is on $\pi(s,r_t)$, then we stop the
algorithm and return it  as an intersection of $\tau$ and
$\pi(s,t)$. Below, we assume neither $a$ nor $b$ is in $\pi(s,r_t)$.
Thus, our goal is to compute the intersection between $\pi(s,r_t)$
and the interior of $\tau$.

\begin{figure}[t]
\begin{minipage}[t]{\linewidth}
\begin{center}
\includegraphics[totalheight=1.8in]{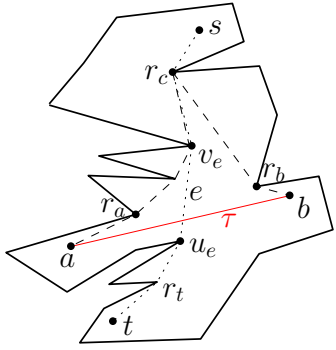}
\caption{\footnotesize
Illustrating an example where $\pi(s,r_t)$ intersects the interior of
$\tau$.}
\label{fig:simpoly}
\end{center}
\end{minipage}
\vspace*{-0.15in}
\end{figure}

Let $r_a$ be the root of the cell of $\spm(s)$ containing $a$. Define $r_b$
similarly. Let $r_c$ be the lowest common ancestor or $r_a$ and $r_b$
in $\spm(s)$  (e.g., see
Fig.~\ref{fig:simpoly}), which can be found in constant time by a lowest common
ancestor query.
Let $F$ denote the funnel that is the region of $P$
bounded by $\pi(r_c,a)$, $\pi(r_c,b)$, and $\overline{ab}$. Note that
both $\pi(r_c,a)$ and
$\pi(r_c,b)$ are convex with the convexity towards the interior of $F$.
We assume that if we traverse from $r_c$ counterclockwise around $\partial
F$ we will be on $\pi(r_c,a)$ before arriving at $\tau$ (otherwise we
exchange the notation $a$ and $b$). Observe that $\pi(s,r_t)$
intersects the interior of $\tau$ if and only if there is an edge $e$ of
$\pi(s,r_t)$ such that $e$ intersects the interior of $\tau$ and one
endpoint of $e$ is in $F$ and the other one is outside $F$ (e.g., see
Fig.~\ref{fig:simpoly}). Let $v_e$ be the endpoint of $e$ in $F$ and
$u_e$ be the endpoint of $e$ outside $F$. Observe that such an edge $e$
exists if and only if $r_t$ is
between $r_a$ and $r_b$ counterclockwise in the circular list
$\calC(\spt(s))$, which can be determined in
$O(\log m)$ time by binary search on the list.

Further, if such an edge $e=\overline{u_ev_e}$ exists, then we further
compute the intersection $e\cap \tau$. To determine the edge $e$, we first find the
vertex $v_e$ as follows. 
We find the lowest common ancestor
of $r_t$ and $r_a$, denoted by $v_1$. If $v_1$ is not $r_c$, then
$v_1$ must be on $\pi(r_c,r_a)$ and $v_e$ is $v_1$. Otherwise, the
lowest common ancestor of $r_t$ and $r_b$ is $v_e$. After $v_e$ is
found, $e$ is the first edge in the shortest path $\pi(v_e,r_t)$ from
$v_e$ to $v_t$, which can be found in $O(\log m)$ time using a
two-point shortest path query on the vertex pair $(v_e,r_t)$ with
$O(m)$ time preprocessing~\cite{ref:GuibasOp89,ref:HershbergerA91}.
\qed
\end{proof}

Combining all our results above, the following lemma gives our final result.

\begin{lemma}\label{lem:SP}
Given $\spm(s)$, we can build a data structure of $O(n\log
h)$ size in $O(n\log h)$ time that can answer each
SP-segment-intersection query in $O(\log h\log n)$ time.
\end{lemma}
\begin{proof}
In the preprocessing, we build the tree $T_1$, which takes $O(n\log
h)$ time and space. For each cell $\Delta$ of the decomposition $\calD$, since it is a simple polygon, we build the data structure in Lemma~\ref{lem:simpleSP} with respect to each super-root of $\Delta$; this takes $O(n)$ time and space in total.

Given $\tau$ and $t$, our query algorithm works as follows. We first
determine the cell $\Delta$ of $\calD$ that contains $t$. We also
determine the super-root $r$ of $\Delta$ such that
$\pi(s,t)=\pi(s,r)\cup \pi(r,t)$. All this can be done in $O(\log n)$
time.
Note that $r$ is a vertex in $T_V$. Hence, we can compute
an intersection between $\tau$ and $\pi(s,r)$ in
$O(\log h\log n)$ time using the tree $T_1$. If there is an intersection,
we return it and stop the algorithm. Otherwise, we compute
an intersection between $\tau$ and $\pi(r,t)$ in the cell $\Delta$.
To this end, we first compute the at most two sub-segments of $\tau\cap
\Delta$ by using the ray-shooting queries inside and outside $\Delta$.
For this, in the preprocessing,
for each cell $\Delta$ of $\calD$, we compute ray-shooting data
structures on both the inside and outside of $\Delta$ (e.g., by the similar techniques as in Lemma~\ref{lem:200}). Computing these
ray-shooing data structure on all cells of $\calD$ takes $O(n)$ time.
Then, for each sub-segment $\tau'$ of $\tau\cap \Delta$, we compute
the intersection (if any) between $\tau'$ and $\pi(r,t)$
in $O(\log n)$ time by Lemma~\ref{lem:simpleSP}.
Hence, the overall query algorithm runs in $O(\log h\log n)$ time.

The lemma thus follows.
\qed
\end{proof}


\subsubsection{The $\calR$-Region Range Queries}

In the following, we give our data structure for answering the
$\calR$-region queries. Specifically, given a range $[i,j]_R$ of indices of the
regions of $\calR$ and an extended-window $\tau \in W$, the query asks
for the ccw-largest index $r\in [i,j]_R$ such that $\tau$ crosses
the region boundary $\partial R_r$ (or report none if such an index
does not exist). We
actually consider a more general query where $\tau$ can be any segment
in $\calP$ (not necessarily in $W$). Our goal is to show the following result.

\begin{lemma}\label{lem:range}
Given $\spm(s)$, we can build a data structure in $O(n\log h)$ time and
space such that each $\calR$-region range query can be answered in $O(\log h\log n)$ time.
\end{lemma}

Recall that for each region $R_r\in \calR$, its boundary $\partial R_r$ consists of
three portions: $\pi(s,v_r)$, $\pi(s,v_{r+1})$, and $\alpha_r$.

Recall that $\calL_l(T_v,v_1)=\{v_1,v_2,\ldots,v_{h^*}\}$.
We build a complete binary search tree $T_2$ as follows. Like $T_1$ in Section~\ref{sec:SPsegment}, the leaves of
$T_2$ from left to right correspond to $v_1,v_2,\ldots,v_{h^*}$. For each
node $u$ of $T_2$, we construct the same auxiliary data structure $P(u)$ as in $T_1$.
In addition, we build another auxiliary data structure $U(u)$ for each
internal node $u$ of $T_2$ as follows.

We use $T_2(u)$ to denote the
subtree of $T_2$ rooted at $u$ and use $S(u)$ to denote the set of the
leaves of $T_2(u)$. As in $T_1$ in Section~\ref{sec:SPsegment}, each
point of $V$ corresponds to a leaf
of $S(u)$ and is also a leaf of $T_V$. Let $p_u$ be the
point of the path $P(u)$ in $T_V$ that is farthest from $s$. In the
case where $P(u)$ is empty, let $p_u$ be $p_{u'}$ for the parent
$u'$ of $u$ if $u\neq s$ and $p_u=s$ otherwise. Note that $p_u$ is a
node of $T_V$. Let $U$ be the union of the paths of $T_V$ from $p_u$ to
all leaves of $S(u)$ in $T_V$, excluding the sub-path from $s$ to $p_u$ in $T_V$.
It is not difficult to see that $U$
is actually a subtree of $T_V$.
Recall that the points of $S(u)$ are consecutive in the list $\calL_l(T_v,v_1)=\{v_1,v_2,\ldots,v_{h^*}\}$.
Let $S(u)$ be $v_a,v_{a+1},\ldots,v_b$ with $1\leq a\leq b\leq h^*$.
If $a<b$ (i.e., $u$ is not a leaf), for each $c\in [a,b-1]$, recall that $\alpha_c$ belongs to $\partial R_c$ and $\alpha_c$ is either a bisector super-curve or a chain of obstacle edges, and we add $\alpha_c$ to $U$ if $\alpha_c$ is a bisector super-curve.
The resulting $U$ is $U(u)$.
Note that $U(u)$ is connected since every point of $U(u)$ has a path
on $U(u)$ connecting to the point $p_u$. We consider $U(u)$ as a
subdivision of the plane by all edges of $U(u)$, without considering the obstacles of $\calP$.

We claim that each cell (excluding the outer unbounded one) of $U(u)$ is simply connected. Indeed, if $U(u)$ does not contain any bisector super-curve $\alpha_c$, then $U(u)$ is a connected subtree of $T_V$ and thus there is only one cell, which is the outer unbounded one. If $U(u)$ contains a bisector super-curve $\alpha_c$ for some $c\in [a,b-1]$, then $\alpha_c$ along with $\pi(p_u,v_c)$ (which is also the path from $p_u$ to $v_c$ in $T_V$ and is in $U(u)$) and $\pi(p_u,v_{c+1})$ forms a closed cell $C$ of $U(u)$. Note that $C$ is also a cell in the decomposition $\calD'$.
Also, for any closed cell $C'$ of
$U(u)$ (i.e., $C'$ is not the outer unbounded one), $C'$ must be formed by a bisector super-curve in $U(u)$ as discussed above. Therefore, each closed cell of $U(u)$ is simply connected.

For each closed cell $C$ of $U(u)$, we build a ray-shooting data
structure. Although $C$ has a bisector super-curve, which
consists of hyperbolic curves instead of line segments, Melissarators
and Souvaine~\cite{ref:MelissaratosSh92} showed that we
can still build a ray-shooting data structure for $C$ in linear time
and space such that each query can be answered in logarithmic
time\footnote{In fact, since each bisector edge of $\spm(s)$ is a convex curve, $C$ is
naturally a splinegon~\cite{ref:MelissaratosSh92}.}.

For the outer cell $C$ of $U(u)$, we can use the similar approach as
Lemma~\ref{lem:200} to preprocess it in linear time such that each ray-shooting query on $C$ can be answered in logarithmic time.


In addition, recall that $\alpha_{h^*}$ connects $v_{h^*}$ and $v_1$. If $\alpha_{h^*}$ is a bisector super-curve, then we build a ray-shooting data structure for $\alpha_{h^*}$~\cite{ref:MelissaratosSh92}.

This finishes the description of our data structure $T_2$.

\begin{lemma}
The space of $T_2$ is $O(n\log h)$ and $T_2$ can be built in $O(n\log
h)$ time.
\end{lemma}
\begin{proof}
First of all, the auxiliary data structures $P(u)$ on all nodes $u$
of $T_2$ can be built in $O(n\log h)$ time and space as in Lemma~\ref{lem:T1}. In the
following, we focus on the second auxiliary data structure $U(u)$.
To analyze the total space, we first show that each edge $e$ of $T_V$
can be in $U(u)$ for at most two nodes $u$ in each level of $T_2$.

Indeed,
assume to the contrary that there are three such nodes. Since the points of $V$ whose paths from $s$ in $T_V$ that contain $e$ must be consecutive in the list
$\{v_1,v_2,\ldots,v_{h^*}\}$ (and thus in the consecutive leaves of $T_2$), by the similar analysis as in Lemma~\ref{lem:T1}, we can find two nodes $u_1$ and $u_2$ sharing the same parent such that $e$ is contained in both $U(u_1)$ and $U(u_2)$. But this implies that $e$ must be stored in $P(u)$ for a proper ancestor $u$ of $u_1$ (or $u_2$). This further implies that $e$ cannot be stored in either $U(u_1)$ or $U(u_2)$.

Hence, each edge $e$ of $T_V$ can be in $U(u)$ for at most two nodes $u$ in the same level of $T_2$. Consequently, each edge of $T_V$ is contained in $U(u)$ for at most $O(\log h)$ nodes  $u$ of $T_2$, as the height of $T_2$ is $O(\log h)$.

Next we show that for each bisector super-curve $\alpha_c$, it is
stored in $U(u)$ for at most $O(\log h)$ nodes $u$ of $T_2$. Recall that
the two endpoints of $\alpha_c$ are two leaves $v_c$ and $v_{c+1}$ of
$T_2$. Notice that $\alpha_c$ is in $U(u)$ if and only if $[c,c+1]\subseteq
[l_u,r_u]$, where $l_u$ (resp., $r_u$) is the index of the leftmost
(resp., rightmost) leaf in the subtree $T_2(u)$.
Clearly, $[c,c+1]\subseteq [l_u,r_u]$ if and only if $u$ is in the
path from the root to the lowest common ancestor of $v_c$ and
$v_{c+1}$, and there are $O(\log h)$ such nodes $u$.

Since the total size of all bisector super-curves is $O(n)$, the space
of $U(u)$ in $T_2$ used to store the bisector super-curves is $O(n\log h)$.

Combining the above discussions, the size of $T_2$ is $O(n\log h)$.

For each node $u$ of $T_2$, constructing $U(u)$ can be done in
linear time in the size of $U(u)$ as follows. Let $v_a,v_{a+1},\ldots,v_b$ be the leaves in $T_2(u)$. We consider the paths from
$p_u$ to these leaves in $T_V$ one by one in a bottom-up manner. Initially we let $U(u)$ contain the only path $\pi(p_u,v_a)$. In general, suppose $\pi(p_u,v_{c-1})$ has been considered (initially, $c-1=a$). Then we process $\pi(p_u,v_{c})$ as follows.
We traverse on $\pi(p_u,v_c)$ from $v_c$ to $p_u$ in $T_V$ until we meet an obstacle vertex that is on the current $U(u)$, and then add all traversed edges of $\pi(p_u,v_c)$ to $U(u)$. We continue the algorithm as above until $\pi(p_u,v_b)$ is processed. Finally, for each $c\in [a,b-1]$ (if $a<b$), if $\alpha_c$ is a bisector super-curve, then we add $\alpha_c$ to $U(u)$.
The above algorithm constructs $U(u)$ in linear time.

Then, we construct the ray-shooting data
structures for the cells of $U(u)$, which can also be done in linear time in
the size of $U(u)$.

Since the total size of $U(u)$ of all nodes $u$ of $T_2$ is $O(n\log h)$, the
total time for constructing the second auxiliary data structures is
$O(n\log h)$.
Therefore, $T_2$ can be computed in $O(n\log h)$ time. \qed
\end{proof}

By using the tree $T_2$, the following lemma gives our query algorithm, which proves Lemma~\ref{lem:range}.

\begin{lemma}
Each $\calR$-region range query can be answered in $O(\log h\log n)$ time.
\end{lemma}
\begin{proof}
Given a range $[i,j]_R$ of indices of the
regions of $\calR$ and a segment $\tau \in \calP$, we want to compute
the ccw-largest index $r\in [i,j]_R$ such that $\tau$ crosses
the boundary $\partial R_r$ (if no such index $r$ exists, then we return none).
Let $r^*$ be the sought index.


Recall that both $i\leq j$ and $i>j$ are possible. We first consider
the case where $i\leq j$. In this case, $[i,j]_R$ consists of
$\{i,{i+1},\ldots,j\}$. 
We begin with finding the lowest common ancestor of the two leaves $v_i$ and
$v_j$ in $T_2$, denoted by $w$.
Our algorithm consists of four procedures.

\paragraph{The first procedure.}
The first procedure considers the nodes in the path of $T_2$ from the
root to $w$. For each node $u$ in the path, we check whether $\tau$
crosses $P(u)$ by a ray-shooting query. If yes, then $\tau$ crosses the shortest
path $\pi(s,v_j)$ and thus crosses $\partial R_j$.
Hence, we can simply return $r^*=j$ and stop the
algorithm. Otherwise, we proceed on the next node until $w$ is
considered.

After $w$ is considered, if $r^*$ is not found, then we go to the second procedure.

\paragraph{The second procedure.}
The second procedure considers the nodes in the path of $T_2$ from
$u_j$ up to $w$ in a bottom-up fashion. For each node $u$, there are three cases.

\begin{enumerate}
\item
If $u=w$, we stop the second procedure and go to the third procedure.

\item
If $u=u_j$, then we check whether
whether $\tau$ intersects $P(u)$ by calling a ray-shooting query. If
there is an intersection, we return $r^*=j$. Otherwise, we proceed on the parent of $u$.

\item
Suppose $u$ is neither $u_j$ nor $w$.

If $u_j$ is in the left sub-tree of $u$, then we check whether $\tau$ intersects $P(u)$
by a ray-shooting query. If there is an intersection, then we return
$r^*=j$. Otherwise, we proceed on the parent of $u$.

If $u_j$ is in the right sub-tree of $u$, then we first check whether $\tau$ intersects $P(u)$. If yes, then we return $r^*=j$. Otherwise, let $u'$ be the left child of $u$ (if $u$ does not have a left child, then we proceed on the parent of $u$). We proceed as follows.

We check whether $\tau$ intersects $P(u')$. If yes, we return
$r^*$ as the rightmost index of the leaves in the subtree $T_2(u')$.
Otherwise, we check whether $\tau$ intersects $U(u')$ by first locating
the cell $C$ of $U(u')$ containing an endpoint of $\tau$ and then calling a
ray-shooting query on $C$. If not, we
proceed on the parent of $u$ (not $u'$). Otherwise, we set $u=u'$ and
go to the fourth procedure.

\end{enumerate}

\paragraph{The third procedure.}
In this procedure, we consider the vertices on the path of $T_2$ from
the left
child of $w$ down to $u_i$, which is symmetric to the second procedure.
For each node $u$, there are two cases.
\begin{enumerate}
\item
If $u\neq u_i$, we first check whether $\tau$ intersects
$P(u)$ by a ray-shooting query. If yes, we return the index of
the rightmost leaf of $T_2(u)$ as $r^*$.
Otherwise, if $u_i$ is at the right subtree of $u$,
then we proceed on the right child of $u$.

If $u_i$ is at the left subtree of $u$, let $u'$ be the
right child of $u$ (if $u$ does not have a right child, then we proceed on the left child of $u$). We first check whether $\tau$ intersects $P(u')$.
If yes, we return the index of the rightmost leaf of $T_2(u')$ as $r^*$.
Otherwise, we check whether $\tau$ intersects $U(u')$. If not, we
proceed on the left child of $u$. Otherwise, we set $u=u'$ and go to the fourth procedure.

\item
If $u=u_i$, then we check whether $\tau$ intersects $P(u)$. If yes, we
return $r^*=i$. Otherwise, we return none, i.e., $\tau$ does not
intersect $\partial R_r$ for any $r\in [i,j]_R$.
\end{enumerate}

\paragraph{The fourth procedure.}
In the fourth procedure, we have a vertex $u$ of $T_2$ such that
$\tau$ does not intersect $P(u)$ but intersects $U(u)$.
Starting from $u$, the procedure works as follows.
If $u$ is a leaf, then we simply return the index of the leaf as
$r^*$. Otherwise, let $u'$ be
the right child of $u$. If $\tau$ intersects $P(u')$, then we return
$r^*$ as the index of the rightmost leaf of $T_2(u')$. Otherwise, we check
whether $\tau$ intersects $U(u')$. If yes, we set $u$ to $u'$ and
proceed as above. Otherwise, we set $u$ to the left
child of $u$ and proceed as above.
\paragraph{}

For the running time of the algorithm, observe that the
algorithm only visits $O(\log h)$ vertices of $T_2$ and makes $O(\log
h)$ ray-shooting queries as the height of $T_2$ is $O(\log h)$.
Each ray-shooting query is either on $P(u)$ or $U(u)$ for some node
$u$ of $T_2$, which runs in $O(\log n)$ time.
Hence, the total time of the algorithm is $O(\log h\log n)$.

The above gives the query algorithm for the case $i\leq j$. If $i>j$,
then the index range $[i,j]_R$ consists of
$\{i,i+1,\ldots,h^*,1,2,\ldots,j\}$. For this case,
we first apply the above query algorithm on the range $[1,j]_R$. If the query
does not return none, then we return $r^*$ as the answer to the
original query on $[i,j]_R$.
Otherwise, if $\alpha_{h^*}$ is a bisector super-curve, then
we check whether $\tau$ intersects $\alpha_{h^*}$ by a ray-shooting
query; if there is an intersection, then we return $r^*=h^*$.
Otherwise, we apply the above query algorithm on the range
$[i,h^*]$, and the result of the query is the answer to the original
query on  $[i,j]_R$. The total time of the query algorithm is
still $O(\log h\log n)$.

The lemma thus follows.
\qed
\end{proof}

\subsection{Wrapping Things Up}
\label{sec:wrapup}

We summarize our overall result in the following theorem.

\begin{theorem}\label{theo:20}
Given $\spm(s)$, we can build a data structure of $O(n\log h + h^2)$ size in
$O(n\log h + h^2\log h)$ time, such that each quickest visibility query can be answered in  $O((K+h)\log h\log n)$ time, where $K$ is the size of the visibility polygon of the query point $q$.
\end{theorem}
\begin{proof}
In the preprocessing, we compute the visibility polygon query data structure in \cite{ref:ChenVi15} for computing $\vis(q)$, which is of $O(n+h^2)$ size and can be built in $O(n+h^2\log h)$ time. The rest of the preprocessing work includes building the decomposition $\calD$ and the segment query data structure as in Section~\ref{sec:segment}, performing the preprocessing in Lemmas~\ref{lem:map}, \ref{lem:100}, \ref{lem:regionpro}, \ref{lem:SP}, and \ref{lem:range}; these work takes $O(n\log h)$ time and space in total.

Given any query point $q$, we first compute $\vis(q)$ in $O(K\log n)$ time by the query algorithm in \cite{ref:ChenVi15}. Then, we obtain the extended window set $W$. Let $k=|W|$, which is $O(K)$. Next, we compute a closest point $q^*$ on a segment of $W$ in $O(k\log h\log n)$ time. To this end, we compute a set $S$ of $O(k)$ candidate points as follows. We first add $q,q_1,\ldots,q_{k}$ to $S$. Then, we compute the closest point $q_0^*$ of $\overline{u_0q_0}$ and add $q_0^*$ to $S$. Next we compute the point $q_l^*$ in $O((k+h)\log h\log n)$ time by using our pruning algorithm in Sections~\ref{sec:prune} and \ref{sec:computeq}. By a symmetric algorithm, we can also compute $q_r^*$. We add both $q_l^*$ and $q_r^*$ to $S$. By our analysis, $q^*$ must be one of the points of $S$. Since $|S|=O(k)$, we can find $q^*$ in $S$ in additional $O(k\log n)$ time by using the shortest path map $\spm(s)$.
\qed
\end{proof}

In fact, we have the following more general result, which might have independent interest.

\begin{corollary}\label{coro:10}
Given $\spm(s)$, we can build a data structure of $O(n\log h)$ size in
$O(n\log h)$ time, such that given $k=O(n)$ segments in $\calP$ intersecting at the same point, we can compute a shortest path from $s$ to all these segments in $O((k+h)\log h\log n)$ time.
\end{corollary}
\begin{proof}
The preprocessing step is the same as in Theorem~\ref{theo:20} except that
the visibility polygon query data structure \cite{ref:ChenVi15} is not necessary any more. Hence, the total preprocessing time and space is $O(n\log h)$.

Given a set $S$ of $k$ segments intersecting at the same point, denoted by $p$, we break each segment at $p$ to obtain two segments and we still use $S$ to denote the new set of at most $2k$ segments. Next we compute a closest point $p^*$ on the segments of $S$. To do so, we can apply the same algorithm as in Theorem~\ref{theo:20} for computing $q^*$ on the extended-windows of $W$. Indeed, the only key property of the segments of $W$ we need is that all segments of $W$ have a common endpoint at $q$. Now that all segments of $S$ have a common endpoint $p$, the same algorithm still works (some degenerate cases may happen, but can be handled easily).
\qed
\end{proof}

\section{The Quickest Visibility Queries: The Improved Result}
\label{sec:second}

In this section, we reduce the query time of Theorem~\ref{theo:20}
to $O(h\log h\log n)$, independent of $K$.
The key idea is the following. First, we show that for any query point
$q$, there exists a
subset $\calS(q)$ of $O(h)$ windows such that a closest point $q^*$ is on a segment of
$\calS(q)$. Second, we give an algorithm that can compute $\calS(q)$ in $O(h\log n)$
time, without computing $\vis(q)$. Our idea relies on the extended corridor
structure~\cite{ref:ChenL113STACS,ref:ChenVi15,ref:ChenCo17} and
modifying the query algorithm for computing $\vis(q)$ in~\cite{ref:ChenVi15}.

Below we first review the extended corridor structure in
Section~\ref{sec:corridor}. We then introduce the
set $\calS(q)$ in Section~\ref{sec:defS}. Finally we present our
algorithm for computing $\calS(q)$ in Section~\ref{sec:comS}.

\subsection{The Extended Corridor Structure}
\label{sec:corridor}

The corridor structure has been used for solving shortest path
problems, e.g.,~\cite{ref:ChenA11,ref:KapoorAn97}.
Later some new concepts such as ``bays,'' ``canals,'' and the ``ocean'' were
introduced, e.g.,~\cite{ref:ChenL113STACS,ref:ChenCo17},
referred to as the ``extended corridor structure''. We review it here
for the completeness of this paper and also for introducing the
notation that will be needed later.

\begin{figure}[t]
\begin{minipage}[t]{0.49\linewidth}
\begin{center}
\includegraphics[totalheight=1.3in]{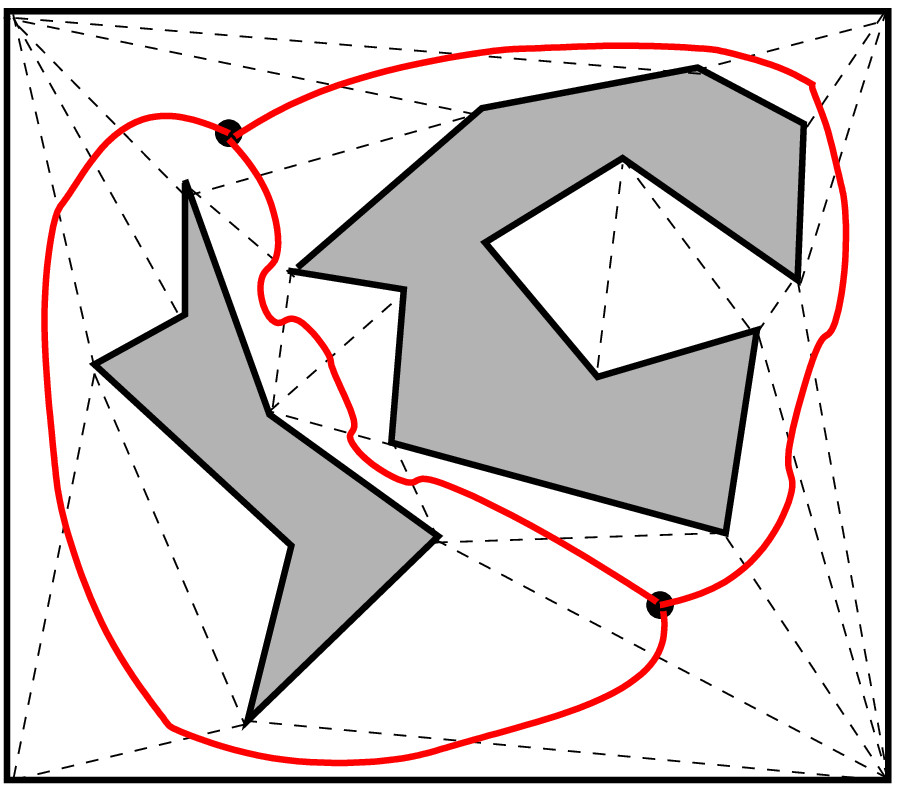}
\caption{\footnotesize
Illustrating a triangulation of the free
space among two obstacles and the corridors (with red solid curves).
There are two junction triangles indicated by the large dots inside
them, connected by three solid (red) curves. Removing the two
junction triangles results in three corridors.}
\label{fig:triangulation}
\end{center}
\end{minipage}
\hspace{0.02in}
\begin{minipage}[t]{0.49\linewidth}
\begin{center}
\includegraphics[totalheight=1.4in]{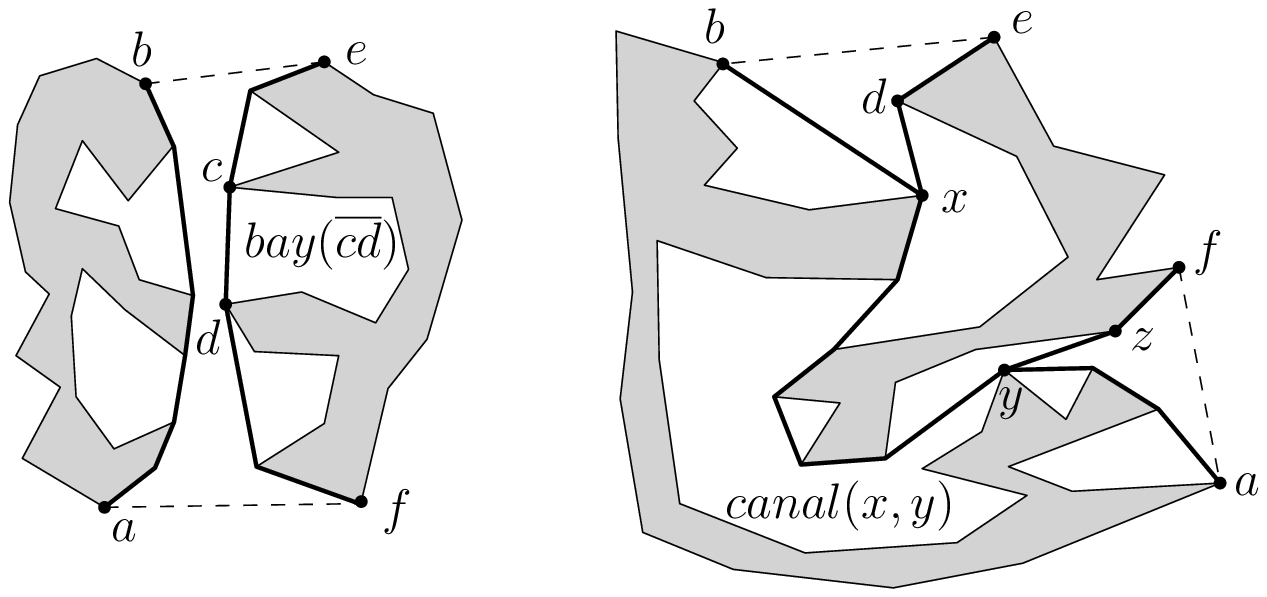}
\caption{\footnotesize
Illustrating an open hourglass (left) and a
closed hourglass (right) with a corridor path connecting the apices
$x$ and $y$ of the two funnels. The dashed segments are diagonals.
The paths $\pi(a,b)$ and $\pi(e,f)$ are marked by thick solid
curves. A bay with gate $\overline{cd}$ (left)
and a canal with gates $\overline{xd}$ and $\overline{yz}$
(right) are also shown.} \label{fig:corridor}
\end{center}
\end{minipage}
\vspace*{-0.15in}
\end{figure}

Let $\Tri(\calP)$ denote an arbitrary triangulation of $\calP$.  Each
edge of $\Tri(\calP)$  that is not an obstacle edge of $\calP$ is called a
{\em (triangulation) diagonal}.
Let $G(\calP)$ be the (planar) dual graph of $\Tri(\calP)$ (i.e., each
triangle defines a node and two triangles that share a diagonal
define an edge).  The degree of
each node in $G(\calP)$ is at most three.  Using $G(\calP)$, we
compute a planar 3-regular graph, denoted by $G^3$ (the degree of
each node in $G^3$ is three), possibly with loops and multi-edges,
as follows. First, remove every degree-one node from $G(\calP)$
together with its incident edge; repeat this process until no
degree-one node remains. Second, remove every degree-two node from
$G(\calP)$ and replace its two incident edges by a single edge;
repeat this process until no degree-two node remains. The
resulting graph is $G^3$ (see Fig.~\ref{fig:triangulation}), which has
$O(h)$ faces, nodes, and edges \cite{ref:KapoorAn97}. Each node of
$G^3$ corresponds to a triangle of $\Tri(\calP)$, which is called a
{\em junction triangle}. Removing all junction triangles results in $O(h)$
{\em corridors} (defined below), each of which corresponds to an edge
of $G^3$.

The boundary of a corridor $C$ consists of four parts (see
Fig.~\ref{fig:corridor}): (1) A boundary portion
of $\calP$ from a point $a$ to a point $b$; (2) a diagonal of a junction triangle from $b$ to $e$; 
(3) a boundary portion of $\calP$ from $e$ to a point $f$; (4) a diagonal of a
junction triangle from $f$ to $a$.
The above (1) and (3) are called the two {\em sides} of $C$.
The corridor $C$ is a simple polygon.

Let $\pi(a,b)$ (resp., $\pi(e,f)$) be the shortest path from $a$ to $b$
(resp., $e$ to $f$) in $C$. The region $H_C$ bounded by
$\pi(a,b), \pi(e,f)$, and the two diagonals $\overline{be}$ and
$\overline{fa}$ is called an {\em hourglass}, which is {\em open} if
$\pi(a,b)\cap \pi(e,f)=\emptyset$ and {\em closed} otherwise (see
Fig.~\ref{fig:corridor}). If $H_C$ is open, then both $\pi(a,b)$ and
$\pi(e,f)$ are convex chains and are called the {\em sides} of
$H_C$; otherwise, $H_C$ consists of two ``funnels'' and a path
$\pi_C=\pi(a,b)\cap \pi(e,f)$ joining the two apices of the two
funnels, called the {\em corridor path} of $C$.
Each side of every funnel is also a convex chain.

The triangulation $\Tri(\calP)$ can be computed in either $O(n\log n)$ time or
$O(n+h\log^{1+\epsilon}h)$ time for any
constant $\epsilon>0$ \cite{ref:Bar-YehudaTr94}.  After $\Tri(\calP)$ is produced,
computing all corridors and hourglasses takes $O(n)$ time.

Let $\calM$ be the union of all $O(h)$ junction triangles, open hourglasses, and funnels.
We call $\calM$ the {\em ocean}, which is a subset of $\calP$. Since the sides
of open hourglasses and funnels are all convex, the boundary
$\partial\calM$ of $\calM$ consists of $O(h)$ convex chains with
a total of $O(n)$ vertices.

The space of $\calP$ not in $\calM$, i.e., $\calP\setminus\calM$, consists
of two types of regions: {\em bays} and {\em canals}, defined as
follows. Consider the hourglass $H_C$ of a corridor $C$.

We first discuss the case where $H_C$ is open (see Fig.~\ref{fig:corridor}). The
boundary of $H_C$ has two sides. Let $c$ and $d$ be any two consecutive
vertices on one side of $H_C$ such that $\overline{cd}$ is
not an obstacle edge (see the left figure in
Fig.~\ref{fig:corridor}). Both $c$ and $d$ must be on the same
side of the corridor $C$.
The region enclosed by $\overline{cd}$ and the side of $C$ between $c$ and
$d$ is called a {\em bay}.
We call $\overline{cd}$ the {\em gate} of the bay, which is
a common edge of the bay and $\calM$.

If the hourglass $H_C$ is closed, let $x$ and $y$ be the two apices
of its two funnels. Consider two consecutive vertices $c$ and $d$ on
a side of a funnel such that $\overline{cd}$ is not an obstacle
edge. If $c$ and $d$ are on the same side of the corridor $C$, then
$\overline{cd}$ also defines a bay. Otherwise, one of $c$ and $d$ must
be a funnel apex, say, $c=x$, and we call $\overline{xd}$ a {\em
canal gate}
(see Fig.~\ref{fig:corridor}). Similarly, there is also a canal gate
at the other funnel apex $y$, say $\overline{yz}$. The
region of $C$ bounded by the two canal gates
$\overline{xd}$ and $\overline{yz}$ that contains the corridor path
is the {\it canal} of $H_C$.

Each bay or canal is a simple polygon.  While the
total number of all bays is $O(n)$, the total number of all canals is
$O(h)$ since
the number of corridors is $O(h)$. The two obstacle vertices of each bay/canal gate
are called {\em gate vertices}.

%

\subsection{Defining the Window Set $\calS(q)$}
\label{sec:defS}

We consider the source point $s$ as an obstacle and build the extended corridor
structure. This means that $s$ is on the boundary of the ocean $\calM$
and thus is not in any bay or canal.

Consider any query point $q$.
For any bay, if $q$ is not in the bay, since the bay has only one gate,
$q$ cannot see any point outside the bay ``through" its gate.
Although a canal has two gates,
the next lemma, proved in
\cite{ref:ChenCo17}, gives an important property that if $q$ is
outside a canal, then $q$ cannot see any point outside the canal
through the canal (and its two gates).

\begin{lemma}\label{lem:opaque} {\em \cite{ref:ChenCo17}}
\emph{(The Opaque Property)}
For any canal, for any line segment $\overline{pq}$ in
$\calP$ (i.e., $p$ is visible to $q$) such that neither $p$ nor $q$ is in the canal.
Then $\overline{pq}$ cannot contain any point of the canal that is not
on its two gates.
\end{lemma}

Consider any window $w_u=\overline{uq(u)}$ of $q$ defined by $u$, i.e.,
$q(u)$ is the first point on $\partial \calP$ hit by the ray from $u$
along the direction from $q$
to $u$. Clearly, the extended-window $\overline{qq(u)}$ is locally tangent at $u$, i.e., the two
incident obstacle edges to $u$ must be on the same side of the
supporting line of $\overline{qq(u)}$.
In the following, we partition all windows of $q$ into different
types.

Recall that $\partial \calM$ is comprised of $O(h)$ convex chains.
We call $w_u$ an {\em ocean window} if $u$ is a vertex of a
convex chain of $\partial\calM$ such that
$\overline{qq(u)}$ is outer tangent to that convex chain at $u$. Since
$q$ has at most two extended-windows outer tangent to each convex chain, $q$ has
$O(h)$ ocean windows.


Suppose $w_u$ is not an ocean window.
If $\overline{qu}\setminus\{u\}$ does not contain any point in
$\calM$, then $\overline{qu}$ is in a bay/canal $A$.
In this case, we say $w_u$ is an {\em outer-bay/outer-canal}
window defined by $A$ (we use ``outer'' because it is possible that
$w_u=\overline{uq(u)}$ contains points outside $A$); e.g., see Fig.~\ref{fig:outerbay}.

\begin{figure}[t]
\begin{minipage}[t]{\linewidth}
\begin{center}
\includegraphics[totalheight=1.0in]{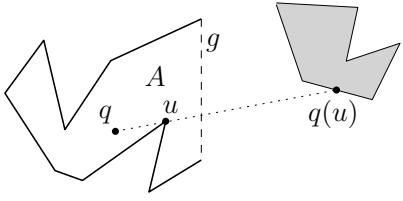}
\caption{\footnotesize
Illustrating an outer-bay window $w_u=\overline{uq(u)}$, where $q$ is in a bay $A$ ($g$ is the gate).} \label{fig:outerbay}
\end{center}
\end{minipage}
\vspace*{-0.15in}
\end{figure}

If $\overline{qu}\setminus\{u\}$ contains a point $q'$ in $\calM$,
then $q'\neq u$. Depending on whether $u$ is on $\partial M$, there
are two cases.

If $u$ is on $\partial\calM$, then $\overline{q'u}$
is in $\calM$ this is because $\overline{q'u}$ cannot traverse through
the interior of a canal due to the opaque property of
Lemma~\ref{lem:opaque}. If we move from $q'$ to $q(u)$ on
$\overline{qq(u)}$, since $w_u$ is not an ocean
window, after we pass $u$, we must move into the
inside of a bay/canal $A$, and further, regardless of whether $A$ is a
bay or a canal, we will never get out of $A$ due to the opaque
property, which implies that $w_u=\overline{uq(u)}$ must be in $A$. In this case, we say
that $w_u$ is an {\em inner-bay/inner-canal window} defined by $A$
(we use ``inner'' because $w_u$ is in $A$).

If $u$ is not on $\partial\calM$, then $u$ is a non-gate vertex of a
bay/canal $A$. This implies that if we move from $q'$  to $u$ on
$\overline{qq(u)}$, we must
cross a gate of $A$. Again, regardless of whether $A$ is a bay or a
canal, $w_u=\overline{uq(u)}$ must be in $A$. In this case, we also
call $w_u$ an {\em inner-bay/inner-canal window} (e.g., see Fig.~\ref{fig:innerbay} and Fig.~\ref{fig:innercanal}).

As a summary, a window $w_u$ may be an ocean window,
an outer-bay/canal window, or an inner-bay/canal window.

A window of $q$ is called a {\em closest window} if it contains a closest point
$q^*$ of $\vis(q)$.

The set $\calS(q)$ is defined as follows. We first add all $O(h)$ ocean windows to $\calS(q)$.
We will show several observations. First, no inner-bay
window can be a closest window. Second, among all
inner-canal windows defined by the same canal,
there are at most two that can be closest windows and we add them to
$\calS(q)$. Since there are $O(h)$ canals, $\calS(q)$ has $O(h)$ inner-canal
windows. Third, among all outer-bay windows, there are at most two
that can be closest windows; we add them to $\calS(q)$. Fourth, among all
outer-canal windows, there are at most four
that can be closest windows; we add them to $\calS(q)$.
This finishes the definition of $\calS(q)$.
In summary, $\calS(q)$ has $O(h)$ ocean windows, $O(h)$ inner-canal windows, at most two outer-bay windows, and at most four outer-canal windows. Thus, the size of $\calS(q)$ is $O(h)$.

For a window $w_u=\overline{uq(u)}$, we assume it is directed from $u$
to $q(u)$ and also assume $\overline{qq(u)}$ is directed from $q$ to $q(u)$.

\begin{observation}\label{obser:closestwin}
Suppose $w_u$ is a closest window, i.e., $q^*\in w_u$. If the two obstacle edges incident
to $u$ are on the left (resp., right) side of $\overline{qq(u)}$, then the
shortest path from $s$ to $q^*$ must be from the left (resp., right)
side of $w_u$.
\end{observation}
\begin{proof}
As discussed before, $\pi(s,q^*)$ is either from the left or from the
right side of $w_u$. Without loss of generality, we assume that
the two obstacle edges incident to $u$ are
on the left side of $\overline{qq(u)}$.

Assume to the contrary that $\pi(s,q^*)$ is from the right side of
$w_u$. Let $p$ be a point on $\pi(s,q^*)$ infinitely
close to $q^*$ but $p\neq q^*$. Since the two obstacle edges incident
to $u$ are on the left side of $\overline{qq(u)}$, $p$ is visible to
$q$, i.e., $p\in \vis(q)$. Since $d(s,p)<d(s,q^*)$, $q^*$ cannot be a closest
point of $\vis(q)$, a contradiction.
\qed
\end{proof}

\begin{lemma}\label{lem:innerbay}
None of the inner-bay windows is a closest window.
\end{lemma}
\begin{proof}
Suppose $w_u=\overline{uq(u)}$ is an inner-bay window
defined by a bay $A$. By definition, $w_u$ is in $A$.
Assume to the contrary that $w_u$ is a closest window.

Without loss of generality, assume the two obstacle edges of $\calP$
incident to $u$ is on the left side of $\overline{qq(u)}$ (e.g., see
Fig.~\ref{fig:innerbay}). Since both $u$ and $q(u)$ are on the boundary of $A$, $w_u$
partitions $A$ into two sub-polygons and one of them contains the only gate $g$
of $A$. Let $A'$ be the sub-polygon that does not contain $g$. Observe
that $A'$ must be locally on the left side of $w_u$.
By Observation~\ref{obser:closestwin}, since $q^*\in w_u$,
$\pi(s,q^*)$ must be from the left side of $w_u$, implying that
$p$ must be in the interior of $A'$, where $p$ is a point on
$\pi(s,q^*)$ infinitely close to $q^*$. Clearly, $s$ is not in $A'$.
Thus, $\pi(s,p)$ must cross $w_u$, but this is not possible since $q^*$ is on $w_u$.
Thus, $w_u$ cannot be an closest window. 
\qed
\end{proof}

\begin{figure}[t]
\begin{minipage}[t]{0.49\linewidth}
\begin{center}
\includegraphics[totalheight=1.2in]{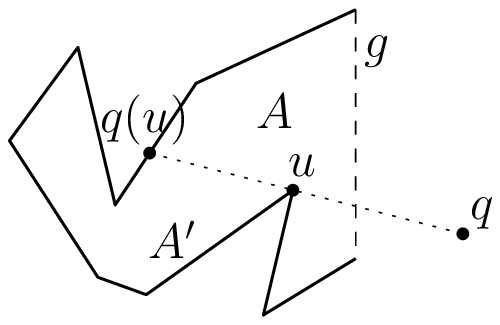}
\caption{\footnotesize
Illustrating an inner-bay window $w_u=\overline{uq(u)}$ in a bay $A$.} \label{fig:innerbay}
\end{center}
\end{minipage}
\hspace{0.02in}
\begin{minipage}[t]{0.49\linewidth}
\begin{center}
\includegraphics[totalheight=1.6in]{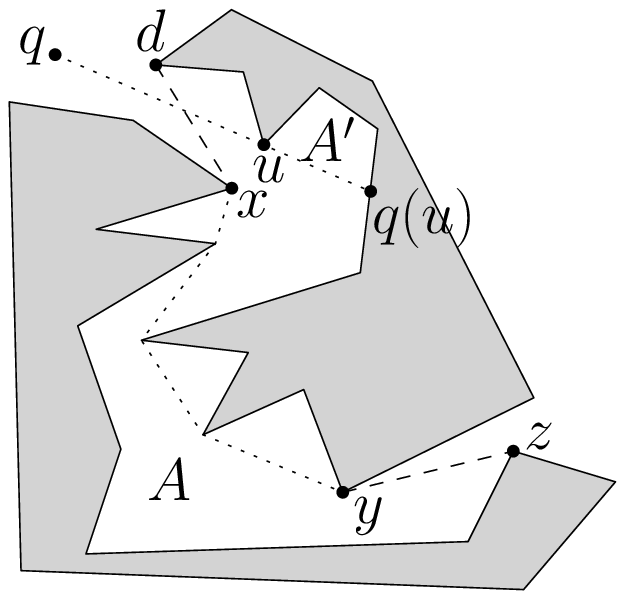}
\caption{\footnotesize
Illustrating an inner-canal window $w_u=\overline{uq(u)}$ defined by a
canal $A$ with two gates $\overline{xd}$ and $\overline{yz}$.}
\label{fig:innercanal}
\end{center}
\end{minipage}
\vspace*{-0.15in}
\end{figure}


%

\begin{lemma}
For any canal $A$ that defines an inner-canal window $w_u$,  if $u$ is
not an endpoint of the corridor path of $A$, then $w_u$ cannot be a
closest window.
\end{lemma}
\begin{proof}
Since $w_u$ is an inner-canal window defined by $A$, $w_u$ must be in $A$
and both $u$ and $q(u)$ are on the boundary of $A$. Further,
$\overline{qu(q)}$ has a point $q'\in \calM$ and $\overline{q'u}$
crosses a gate $g$ of $A$. Let $g=\overline{xd}$ such that $x$ is
the endpoint of the corridor path of $A$ on $g$ (e.g., see
Fig.~\ref{fig:innercanal}). Let $C$ be the corridor that defines the canal $A$.

Assume without loss of generality that the two obstacle edges of
$\calP$ incident to $u$ are on the left side of $\overline{qq(u)}$.
Since $u$ is not $x$,
according to the results in \cite{ref:ChenCo17} (see the proof of
Lemma 3) that $u$ and $q(u)$ must be on the same side of $C$ that
contains $d$ (e.g., see Fig.~\ref{fig:innercanal}).
This implies that $w_u$ partitions $A$ into two sub-polygons one of
which contains both gates of $A$, and let $A'$ be the sub-polygon that
does not contain the gates. Then, as in the proof of
Lemma~\ref{lem:innerbay}, $A'$ must be locally on the left side of
$w_u$, and by the similar analysis we can show that $w_u$ cannot be a closest window.
\qed
\end{proof}

Since each canal has one corridor path,
the preceding lemma implies that every canal can define at most two
inner-canal windows that are possibly closest windows.

Consider a bay $A$ with gate $g$ that defines an outer-bay window $w_u$.
By definition, $\overline{qu}$ is in $A$. Let $u_1$ be the vertex of
$A$ such that $\overline{qu_1}$ is in the shortest path  in $A$ from $q$
to an endpoint of $g$; similarly, define $u_2$
with respect to the other endpoint of $g$.

%

\begin{lemma}
If $w_u$ is an outer-bay window defined by $A$ and $u$ is neither $u_1$
nor $u_2$, then $w_u$ cannot be a closest window.
\end{lemma}
\begin{proof}
By the definitions of $u_1$ and $u_2$, since $A$ is a simple polygon and $u$ is neither $u_1$
nor $u_2$, $q(u)$ must be in $\partial A\setminus\{g\}$. Hence, the
window $w_u$ partitions $A$ into two sub-polygons and one of them
contains $g$. Let $A'$ be the sub-polygon that does not contain $g$.
Then, by using the same analysis as in Lemma~\ref{lem:innerbay}, $w_u$ cannot be a closest
window.
\qed
\end{proof}


Consider a canal $A$ that defines an outer-canal window $w_u$. This case
is similar to the above bay case except that we need to
consider both gates of $A$. Again, $\overline{qu}$ is in
$A$. Define $u_1$, $u_2$, $u_3$, and $u_4$ similarly as in the bay
case but with respect to the four gate vertices of $A$,
respectively.

\begin{lemma}
If $w_u$ is an outer-bay window defined by $A$ and $u$ is not in
$\{u_1,u_2,u_3,u_4\}$, then $w_u$ cannot be a closest window.
\end{lemma}
\begin{proof}
By the definitions of $u_i$ for $1\leq i\leq 4$, since $A$ is a simple
polygon and $u\not\in\{u_1,u_2,u_3,u_4\}$, $q(u)$ must be in $\partial
A$ and $q(u)$ is not on a gate of $A$. Further, it can be verified
that the window $w_u$ partitions $A$ into two
sub-polygons and one of them contains both gates of $A$. Let $A'$ be
the sub-polygon that does not contain the gates of $A$. Then, by using
the same analysis as in Lemma~\ref{lem:innerbay}, $w_u$ cannot be a closest window.
\qed
\end{proof}

The above discussions lead to the following lemma.

\begin{lemma}
Given any query point $q$, there is a set $\calS(q)$ of windows of $q$
such that $|\calS(q)|=O(h)$ and $\calS(q)$ contains a closest window.
\end{lemma}

\subsection{Computing the Window Set $\calS(q)$}
\label{sec:comS}

In this section we present our algorithm for computing $\calS(q)$,
by modifying the query algorithm in \cite{ref:ChenVi15} for computing
$\vis(q)$. Our result is summarized in the following lemma.

\begin{lemma}\label{lem:comset}
With $O(n+h^2\log h)$ time and $O(n+h^2)$ space preprocessing, given
any query point $q$ in $\calP$, we can compute the set $\calS(q)$ in $O(h\log n)$ time.
\end{lemma}

We first do the same preprocessing as in \cite{ref:ChenVi15},  which takes
$O(n+h^2\log h)$ time and $O(n+h^2)$ space.
In the following, we give our query algorithm for computing $\calS(q)$.
Depending on whether $q$ is in the ocean $\calM$, a bay, or a canal, there are three cases. In each case, we will first briefly review the algorithm in
\cite{ref:ChenVi15} for computing $\vis(q)$ and then modify it to compute
$\calS(q)$.

\subsubsection{The Ocean Case}
Suppose $q$ is in $\calM$. The algorithm in \cite{ref:ChenVi15} first computes the
region of $\calM$ that is visible to $q$, denoted by $\vis(q,\calM)$,
which is also the visibility polygon of $q$ in $\calM$ due to the
opaque property of canals.
Then, the algorithm computes the region in all bays and canals visible to
$q$. To this end, it traverses on the boundary of $\vis(q,\calM)$. If
a gate $g$ of a bay/canal $A$ is encountered, then the region
of $A$ visible to $q$ through $e$ is computed, where
$e$ is a maximal portion of $g$ on the boundary of $\vis(q,\calM)$.
The visible regions computed above for all such $e$'s are pairwise
disjoint. Hence, $\vis(q)$ is a trivial union of $\vis(q,\calM)$ and the visible
regions in all bays and canals.

We modify the above algorithm to compute $\calS(q)$, as follows.

The algorithm in  \cite{ref:ChenVi15} computes $\vis(q,\calM)$ by
using the visibility complex~\cite{ref:PocchiolaTo96,ref:PocchiolaTh96}. More
specifically, it uses the approach of crossing faces~\cite{ref:PocchiolaTh96}
such that all rays originating from
$q$ in the plane define a curve $\gamma$ in the visibility complex and each
intersection of $\gamma$ and the boundary of a cell of the visibility complex
corresponds to an outer tangent in $\calM$ from $q$ to a convex chain of
$\partial\calM$. Note that such tangents correspond exactly to our ocean windows.
If we traverse the curve $\gamma$ in the visibility complex, each such intersection can be
computed in $O(\log n)$ time. Hence, if there are $h'$ convex chains
of $\partial \calM$ that are visible to $q$, then the endpoints of the
maximal sub-chains $\xi$ of these convex chains that are visible to $q$ can be computed
in $O(h'\log n)$ time by using the approach of crossing faces.
Note that $h'=O(h)$~\cite{ref:ChenVi15}. After this, all ocean windows are computed.

\paragraph{Remark.} Traversing each such sub-chain $\xi$ can explicitly
construct $\vis(q,\calM)$. But for our problem of computing $\calS(q)$, we
can avoid this step; indeed, this is part of the reason our algorithm
avoids the $\Omega(K)$ time.
\paragraph{}

Next, we compute other windows of $\calS(q)$.
Since $q$ is in $\calM$, $\calS(q)$ does not have outer-bay/outer-canal
windows, and we only need to compute the inner-canal windows, as
follows.

The above has computed  the endpoints of  each such sub-chain $\xi$
that is visible to $q$. If $\xi$ does not contain any portion of
any canal gate, then we simply ignore $\xi$. Otherwise, we need to
compute the inner-canal windows through $g$ for each canal gate $g$
that has a portion in $\xi$. To this end, we need to first
find these canal gates. For this, in the preprocessing step, for each
convex chain $C$ of $\calM$, we maintain a list of canal gates
on $C$ by a balanced binary search tree
such that given the two endpoints $a$ and $b$ of $\xi$, we can
determine whether $\xi$ contains any portion of any canal gate in
$O(\log n)$ time, and if yes, report all these portions in $O(k+\log
n)$ time, where $k$ is the number of these portions. The number of such $k$ in the
entire algorithm is $O(h)$ since the total number of canal gates is
$O(h)$. For each such canal gate portion $e$, we compute the
corresponding inner-canal window (if any)  as follows.

Let $g$ be the canal gate containing $e$ and let $A$ be the canal.
Let $x$ be the endpoint of the corridor path of $A$ at $g$.
If $x$ is not on $e$, then we ignore $e$. Otherwise, $x$ is visible to
$q$ and $x$ defines an inner-canal window $w_u$ with $u=x$. Our goal
is to compute $q(u)$. This can be easily done by using a ray-shooting query in
$A$ as follows.
Consider the ray originating from $x$ with direction from $q$ to $x$.
Using a ray-shooting query on $A$, we find the first point
$p$ on the boundary of $A$ that is hit by the ray. Again, due to the
opaque property of canals, $p$ must be on an obstacle edge of $\calP$, and thus
$q(u)=p$. For answering each ray-shooting query in $A$ in $O(\log n)$
time, we need to
preprocess each canal for ray-shooting queries in linear time since a canal is a
simple polygon, and this requires $O(n)$ time in total for all canals.

Since the number of all visible sub-chains is $O(h)$,
we can compute all inner-canal windows in $O(h\log n)$ time.

In summary, we can compute the set $\calS(q)$ in $O(h\log n)$ time for the ocean case.

\subsubsection{The Bay Case}
If $q$ is in a bay $A$, then the algorithm in \cite{ref:ChenVi15}
for computing $\vis(q)$ first computes the
region of $A$ that is visible to $q$, denoted by $\vis(s,A)$. If
the gate $g$ of $A$ does not have any point on the boundary of
$\vis(s,A)$, then $g$ is not visible to $q$, which further implies
that no point outside the bay is visible to $q$ and thus
$\vis(s)=\vis(s,A)$.
If $g$ has a sub-segment $g'$ on the boundary of $\vis(s,A)$,
then the points of $\calP\setminus A$ visible to $q$ are all visible
to $q$ through $g'$. Next, the region $\vis(q,\calM)$ of $\calM$ that
are visible to $q$ through $g'$ is computed. After $\vis(q,\calM)$ is
computed, the rest of the algorithm
is the same as the ocean case. Namely, by traversing the boundary of
$\vis(q,\calM)$, other regions of $\calP$ in bays and canals visible to
$q$ can be computed.

Next we modify the above algorithm to compute $\calS(q)$.

Since $q$ is in $A$, we first compute the (at most two) outer-bay
windows. Let $a$ and $b$ be the two endpoints of $g$, respectively. In
the preprocessing, we compute the shortest path maps of $a$ and $b$ in
$A$, respectively. We also compute a ray-shooting data structure in
$A$. The total such preprocessing takes $O(n)$ time for all
bays. Then, using the shortest path maps of $a$ and $b$,
the two vertices $u_1$ and $u_2$ as defined before can be computed in $O(\log n)$ time.

If $u_1=u_2$, then consider the ray $\rho$ originating from $u_1$
along the direction from $q$ to $u_1$. Let $p$ be the first point on
the boundary of $A$ hit by $\rho$. Since $u_1=u_2$, $p$ must be on an
obstacle edge of $A$ (i.e., $p$ is not on the gate $g$ of $A$),
and thus $\overline{u_1p}$ is an outer-bay window.
In fact, in this case
$\overline{u_1p}$ is the only window in $\calS(q)$, and thus we can stop
our algorithm.

If $u_1\neq u_2$, then for each $u_i$ with $i=1,2$, the intersection
of $g$ with the supporting line of $\overline{qu_i}$ is an endpoint of
$g'$~\cite{ref:GuibasLi87}. Hence, $g'$ can be determined immediately
once $u_1$ and $u_2$ are available.
Similarly as in the above ocean case, the algorithm in~\cite{ref:ChenVi15} uses the approach
of crossing faces to compute $\vis(q,\calM)$ through $g'$, which is
actually a ``cone'' visibility query since the visibility of $q$ in
$\calM$ is delimited by the cone bounded by the ray from $q$ to $u_1$ and the ray from $q$
to $u_2$. All rays from $q$ in the cone define a segment
$\gamma'$ of the curve $\gamma$ (discussed in the ocean case)
in the visibility complex. To use the
approach of crossing faces, the algorithm in~\cite{ref:ChenVi15} first
finds the cell $\sigma$ of the visibility complex that contains an endpoint of $\gamma'$, which
is done in $O(\log n)$ time by a point location data structure on the
visibility complex. After this, the rest of the algorithm is the same
as the ocean bases. This is also the case for our problem for
computing $\calS(q)$. After locating the cell $\sigma$, we can use the
crossing face approach to compute the $O(h)$ maximal
sub-chains $\xi$ of the convex chains of $\partial\calM$ that are
visible to $q$ through $g'$.
As in the ocean case, this will also compute all ocean windows of $\calS(q)$. After that, we
use the same approach as in the ocean case to compute all inner-canal
windows. The total time is $O(h\log n)$.

Finally, we compute the two outer-bay windows defined by $u_1$ and
$u_2$. Namely, we need to compute $q(u_1)$ and $q(u_2)$. For
each $i=1,2$, let $\rho_i$ be the ray originating from $q$ and along
the direction from $q$ to $u_i$. The above algorithm for computing the
sub-chains will also determine the point $p_i$ on $\partial\calM$
first hit by $\rho_i$.
If $p_i$ is on an obstacle edge of $\calP$, then $p_i$ is
$q(u_i)$.
Otherwise, $p_i$ is on a bay/canal gate $g_i$ of a bay/canal $A$.
Then, we use a ray-shooting
query on $A$ to find the first point $p_i'$ on the boundary of $A$ hit by
$\rho_i$. Regardless of whether $A$ is a bay or a canal, $p_i'$ is
always on an obstacle edge, and thus $p_i'$ is $q(u_i)$. Since the
ray-shooting query on $A$ takes $O(\log n)$ time, the two outer-bay
windows can be computed in $O(\log n)$ time.


In summary,  the window set $\calS(q)$ can be computed in
$O(h\log n)$ time for the bay case.

\subsubsection{The Canal Case}

If $q$ is in a canal $A$, then the algorithm is similar to the bay case with the
difference that we apply the same algorithm on the two gates of the canal separately.
Specifically, let $g=\overline{ab}$ be a gate of $A$. We first compute
the vertices $u_1$ and $u_2$ with respect to $a$ and $b$,
respectively. Then, we apply exactly the same algorithm as in the bay
case. After that, we consider the other gate of $A$ and apply the same
algorithm. Then $\calS(q)$ is computed and the total time is $O(h\log n)$
time.

This proves Lemma~\ref{lem:comset}.
After $\calS(q)$ is computed, we can apply the query algorithm of Theorem~\ref{theo:20} (or Corollary~\ref{coro:10}) on the windows of $\calS(q)$ to compute $q^*$.
Thus we can obtain the following result.

\begin{theorem}\label{theo:30}
Given $\spm(s)$, we can build a data structure of $O(n\log h + h^2)$ size in
$O(n\log h + h^2\log h)$ time, such that each quickest visibility query can be answered
in $O(h\log h\log n)$ time.
\end{theorem}

%

\section{Conclusions}
\label{sec:con}

In this paper, we present a new data structure for answering quickest
visibility queries. Our result is particularly interesting when $h$,
the number of holes of $\calP$, is relatively small. For example, when
$h=O(1)$, our result matches the best result for the simple polygon
case (i.e., $h=1$) and is optimal. To achieve the result, we also
solve many other problems that may be interesting in their own right. We
highlight some of them below. We assume that the shortest path map $\spm(s)$
of the source point $s$ has been given.
\begin{enumerate}

\item
We present an algorithm that can compute a shortest path from $s$ to
$\tau$ in $O(h\log\frac{n}{h})$ time for any query segment $\tau\in
\calP$, after $O(n)$ time and space preprocessing.

\item
We present an algorithm that can compute in $O(\log h\log n)$ time an intersection between
$\tau$ and the shortest path $\pi(s,t)$ for any segment $\tau$ and
any point $t$ in $\calP$, after $O(n\log h)$ time and space
preprocessing.

\item
We present an algorithm that can answer each $\calR$-region range
query in $O(\log h\log n)$ time, after $O(n\log h)$ time and space
preprocessing.

\item
We present an algorithm that can compute in $O((k+h)\log h\log n)$ time a
shortest path from $s$ to any set of $k=O(n)$ segments in
$\calP$ that intersect at a same point, after $O(n\log h)$ time and space preprocessing.
\end{enumerate}

These results are particularly interesting when $h$ is relatively
small, and at least the first three results are optimal when $h=O(1)$.

In addition, the decomposition $\calD$ of $\calP$, the regions of $\calR$, and
some other techniques proposed in the paper (e.g., bundles) may find other applications as well.



\bibliographystyle{plain}
\bibliography{reference}

%


\end{document}